\documentclass[format=acmsmall, review=false, screen=true]{acmart}
\usepackage{amsthm}
\usepackage{longtable}
\usepackage{amssymb,amsmath}
\usepackage{proof}
\usepackage{mathpartir}
\usepackage[obeyFinal]{todonotes}

\newcommand{\type}{\texttt{type}}
\newcommand{\prop}{\texttt{prop}}
\newcommand{\sing}{\mathfrak{S}}
\newcommand{\myG}{\Gamma}

\newcommand{\entS}{\ensuremath{\vdash_{\Sigma;\Xi}}}
\newcommand{\ent}{\vdash}

\newcommand{\SE}{\cdot \ent}
\newcommand{\GE}{\Gamma \ent}

\newcommand{\SPE}{\Psi \entS}

\newcommand{\DE}{\Delta \ent}

\newcommand{\DGE}{\Delta;\Gamma \ent}

\newcommand{\PE}{\Psi \ent}
\newcommand{\DPE}{\Delta;\Psi \ent}

\newcommand{\squig}{\rightsquigarrow}

\newcommand{\Gpluss}[2]{(#1\text{\small +}#2)}
\newcommand{\Gplus}[1]{\Gpluss{\Gamma}{#1}}
\newcommand{\hook}{\mapsto}
\newcommand{\bound}[1]{\textbf{BOUND}\ #1}
\newcommand{\free}{\textbf{FREE}}

\newcommand{\unify}{\texttt{unify}}
\newcommand{\unifyargs}{\texttt{unify\_args}}
\newcommand{\args}[1]{\langle #1 \rangle}
\newcommand{\cross}[1]{{\textrm{\small x}}[#1]}

\newcommand{\func}[2]{#1 \args{#2}}

\newcommand{\fails}{\ \texttt{fails}}
\newcommand{\heap}[1]{\{#1\}}

\newcommand{\hup}[3]{#1 \{ #2 \hook #3 \}}
\newcommand{\hups}[2]{#1 \{ #2 \}}
\newcommand{\thup}[3]{#1 \{ #2 \colon #3 \}}
\newcommand{\hext}[3]{#1 \{ \hskip -2pt \{ #2 \hook #3 \}\hskip -2pt\}}
\newcommand{\thext}[3]{#1 \{ \hskip -2pt \{ #2 \colon #3 \}\hskip -2pt\}}
\newcommand{\hexts}[2]{#1 \{ \hskip -2pt \{ #2 \}\hskip -2pt\}}
\newcommand{\sub}[1]{[\hskip -1pt[#1]\hskip-1pt]}

\newcommand{\step}{\longmapsto}

\newcommand{\putvar}[2]{\texttt{put\_var }[#1] #2}
\newcommand{\dputvar}[2]{\texttt{put\_var }#1, #2}

\newcommand{\mov}[2]{\texttt{mov}\ #1, #2}

\newcommand{\getval}[2]{\texttt{get\_val}\ #1, #2}
\newcommand{\setval}{\texttt{set\_val}\ }
\newcommand{\getstr}[2]{\texttt{get\_str}\ #1, #2}
\newcommand{\putstr}[2]{\texttt{put\_str}\ #1, #2}
\newcommand{\unifyvar}[1]{\texttt{unify\_var}\ #1}
\newcommand{\unifyval}[1]{\texttt{unify\_val}\ #1}

\newcommand{\jmp}[1]{\texttt{jmp}\ #1}

\newcommand{\trustme}[1]{\texttt{trust\_me}}
\newcommand{\ret}[1]{\texttt{ret}}
\newcommand{\fresh}{\ \texttt{fresh}}

\newcommand{\ells}{\args{\ell_1, \ldots, \ell_n}}
\newcommand{\ellsh}{\args{\ell_1^H, \ldots, \ell_n^H}}

\newcommand{\cwsh}{c\ellsh}

\newcommand{\ok}{\ \texttt{ok}}

\newcommand{\done}{\ \texttt{done}}
\newcommand{\succeed}{\texttt{succeed}}

\newcommand{\edn}{\texttt{end}}
\renewcommand{\to}{\rightarrow}

\newcommand{\dom}[1]{\mathit{Dom}(#1)}

\newcommand{\case}[1]{\noindent \textbf{case } #1 \vskip 5pt}
\newcommand{\subcase}[1]{\noindent \textbf{subcase } #1 \vskip 5pt}

\newcommand{\code}{\texttt{code}}

\newcommand{\bt}{\texttt{backtrack}}

\newcommand{\branch}[1]{\texttt{push\_bt}\ #1}
\newcommand{\unwind}{\texttt{unwind}}

\newcommand{\uptrail}{\texttt{update\_trail}}
\newcommand{\puttuple}[1]{\texttt{put\_tuple}\ #1}
\newcommand{\proj}[1]{\texttt{proj}\ #1}

\newcommand{\close}[1]{\texttt{close}\ #1}
\newcommand{\closure}{\texttt{close}}

\newcommand{\mread}{\texttt{read}}
\newcommand{\mwrite}{\texttt{write}}
\newcommand{\twrite}{\texttt{twrite}}
\newcommand{\reads}{\texttt{ reads }}
\newcommand{\writes}{\texttt{ writes }}
\newcommand{\arity}{\texttt{arity}}
\newcommand{\pred}{\texttt{pred}}
\newcommand{\elems}{\texttt{elems}}
\renewcommand{\succ}{\texttt{succ}}

\newcommand{\canon}{\texttt{ canon}}
\renewcommand{\path}{\texttt{ path}}
\newcommand{\eval}{\Downarrow}
\newcommand{\size}{\mathit{size}}

\newtheorem{thm}{Theorem}
\newtheorem*{thm*}{Theorem}
\newtheorem{lem}{Lemma}
\newtheorem*{lem*}{Lemma}
\newtheorem{claim}{claim}
\newtheorem{subclaim}{subclaim}

\theoremstyle{definition}
\newtheorem{exmp}{Example}[section]
\acmJournal{TOCL}
\acmVolume{1}
\acmNumber{1}
\acmArticle{1}
\acmYear{2017}
\acmMonth{4}
\copyrightyear{2017}
\setcopyright{acmlicensed}
\acmDOI{0000001.0000001}

\received{April 2017}
\begin{document}

\setlength{\pdfpageheight}{\paperheight}
\setlength{\pdfpagewidth}{\paperwidth}

%\conferenceinfo{CONF 'yy}{Month d--d, 20yy, City, ST, Country}
%\copyrightyear{20yy}
%\copyrightdata{978-1-nnnn-nnnn-n/yy/mm}
%\copyrightdoi{nnnnnnn.nnnnnnn}

% Uncomment the publication rights you want to use.
%\publicationrights{transferred}
%\publicationrights{licensed}     % this is the default
%\publicationrights{author-pays}

%\titlebanner{banner above paper title}        % These are ignored unless
%\preprintfooter{short description of paper}   % 'preprint' option specified.

\title{TWAM: A Certifying Abstract Machine for Logic Programs}

\author{Brandon Bohrer}
\orcid{0000-0001-5201-9895}
\author{Karl Crary}
\affiliation{
  \institution{Carnegie Mellon University}
  \department{Computer Science Department}
  \city{Pittsburgh}
  \state{PA}
  \postcode{15213}
  \country{USA}
}
%\authorinfo{Brandon Bohrer \and Karl Crary}
%           {Carnegie Mellon University}
%           {\{bbohrer,crary\}@cs.cmu.edu}
 
\begin{abstract}
 Type-preserving (or typed) compilation uses typing derivations to certify correctness properties of compilation. 
 We have designed and implemented a type-preserving compiler for a simply-typed dialect of Prolog we call T-Prolog.
 The crux of our approach is a new \textit{certifying abstract machine} which we call the Typed Warren Abstract Machine (TWAM).
 The TWAM has a dependent type system strong enough to specify the semantics of a logic program in the logical framework LF.
 We present a soundness metatheorem which constitutes a partial correctness guarantee:  well-typed programs implement the logic program specified by their type.
 This metatheorem justifies our design and implementation of a certifying compiler from T-Prolog to TWAM.
\end{abstract}

\begin{CCSXML}
<ccs2012>
<concept>
<concept_id>10003752.10003790.10002990</concept_id>
<concept_desc>Theory of computation~Logic and verification</concept_desc>
<concept_significance>500</concept_significance>
</concept>
<concept>
<concept_id>10003752.10003790.10003795</concept_id>
<concept_desc>Theory of computation~Constraint and logic programming</concept_desc>
<concept_significance>500</concept_significance>
</concept>
<concept>
<concept_id>10011007.10011006.10011041</concept_id>
<concept_desc>Software and its engineering~Compilers</concept_desc>
<concept_significance>500</concept_significance>
</concept>
<concept>
<concept_id>10003752.10003790.10011740</concept_id>
<concept_desc>Theory of computation~Type theory</concept_desc>
<concept_significance>100</concept_significance>
</concept>
</ccs2012>
\end{CCSXML}

\ccsdesc[500]{Theory of computation~Logic and verification}
\ccsdesc[500]{Theory of computation~Constraint and logic programming}
\ccsdesc[500]{Software and its engineering~Compilers}
\ccsdesc[100]{Theory of computation~Type theory}
\keywords
{Certifying Compilation, Proof-Producing Compilation, Prolog, Warren Abstract Machine}
\nocite{*}

\maketitle

%Done
%TODO: Check for missing T in unify
%TODO: H vs S in unify_args
%TODO: Mention minor syntactic changes in depwam
%TODO: Notation S(l)
%TODO:  unifyval:   unify_val r, x:a. unifyvar:   unify_var r, x:a.
%TODO: Update :_s notation for spine typing
%TODO: Make it obvious that 1 program trace = 1 query
%TODO: Make obvious why sometimes we want to ignore elements in a heap
%TODO: Be real clear about use of \vec
%TODO: motivate failure vs. success continuation early
%TODO: One reviewer didnt seem to get how static everything was, make sure to draw that out
%TODO: 
% Simplify notation to:
% =: unify, unifyargs, \sqcap
% ::, epsilon : unifyargs list, trails, trail frames, machine lists
% \langle\rangle: empty tuple ONLY
% {}: Empty heap type/rftype ONLY 
%Additions: Karl intro, compiler impl, tail-call optimization
%TODO: make sure all frees are annotated
%TODO: Get terminology clear structure vs functor vs prolog term vs. constructor
%TODO: consistency z and s vs zero and succ
%TODO: Mention sometimes we write down arity but don't actually care
%TODO: Clear up ell^C vs. \ell^H
%TODO: Ensure trail judgement always has right contexts
%TODO: Maybe mention canons in heap value

%Notdone

\section{Introduction}
Compiler verification is important because of the central role that compilers play in computing infrastructure, and because compiler bugs are easy to make but often difficult to catch.
Most work on compiler verification has been done in the setting of imperative or functional programming; very little has been done for logic programming languages like Prolog.

Compiler verification is an equally interesting problem in the case of logic programming.
Logic programs are often easier to write correctly than programs in other paradigms, because a logic program is very close to
being its own specification.  However, the correctness advantages of logic
programming cannot be fully realized without compiler verification.
Compiler correctness is a concern for logic programming given the scale of realistic language implementations;
for example, SWI-Prolog is estimated at over 500,000 lines of code~\cite{OpenHub}.

Certifying compilation ~\cite{necula1998design} is an approach to verification wherein a
compiler outputs a formal proof that the compiled program satisfies
some desired properties.  Certifying compilation, unlike
conventional verification, has the advantage that the certificates can
be distributed with the compiled code and checked independently by
third parties, but the flip side is that compiler bugs are not
found until the compiler sees a program that elicits the bug.  In
the worst case, bugs might be found by the compiler's users, rather than its
developers.

%the properties being certified include type and memory safety, they do not include preservation of dynamic semantics.
%In contrast, we show preservation of dynamic semantics in the case of logic programming.%
%We achieve this by specifying the dynamic semantics of a logic program as a logical  LF signature.
%namely the order in which rules are tried, with the advantage that the type {\em are\/} the program.
In most work on certifying compilation ~\cite{necula1998design}, an additional disadvantage is that while type and memory safety are certified, dynamic correctness is not.
In contrast, we certify the dynamic behavior of logic programs, whose semantics we give as a signature in the logical framework LF \cite{Harper93aframework}.
This semantics abstracts away the low-level operational details of Prolog semantics such as order of execution.
This brings the type system into close harmony with our source programs, allowing type correctness to naturally encompass dynamic correctness.

In this work, we develop the Typed Warren Abstract Machine (TWAM), a dependently-typed \emph{certifying abstract machine} suitable as a compilation target for certifying compilers.
Section \ref{sec:dependent-wam} formalizes and proves the following claim that TWAM certifies partial correctness:

\textbf{Theorem 1}(\textit{Soundness}): If a query \verb+?-+$G.$ succeeds, there exists a proof of $G$ in LF.
That is, well-typed TWAM programs satisfy partial correctness with respect to their LF semantics.
Because this theorem says that well-typed TWAM programs implement sound proof search procedures for the LF specification, we also call this theorem \textit{soundness}.
We show that TWAM is a suitable compilation target by implementing a compiler from a simply-typed dialect of Prolog called T-Prolog to the TWAM.
The result is a certifying compiler with a small, domain-specific proof checker as its trusted core: the TWAM typechecker.

We ease the presentation of TWAM by first presenting its simply-typed variant (SWAM) in Section \ref{sec:simple-wam} along with standard progress and preservation theorems, which show type and memory-safety.
We then develop a dependently-typed variant in Section \ref{sec:dependent-wam} whose type system expresses the behavior of a TWAM program as a logic program in the logical framework LF \cite{Harper93aframework}, using an encoding demonstrated in Section \ref{sec:lf-encoding}.

%\subsection{Methods}
%\subsubsection{Certifying Compilation}

%\subsubsection{Warren Abstract Machine (WAM)}
%\section{Background$:$ Warren Abstract Machine}
\section{Source Language: T-Prolog}
Our compiler accepts programs for a simply-typed dialect of Prolog which we named T-Prolog.
It is worth noting that the language need not be typed for our approach to work: if we wished to
work in an untyped dialect of Prolog, we could simply add a compiler pass to collect a list of
all the constructors used in a particular Prolog program and construct a single type called \verb+term+ containing
every constructor we need. We choose a simply-typed language over an untyped one because our use of LF in the TWAM
makes this feature easy to implement and because the correspondence with LF is easier to present for a simply-typed
language.

T-Prolog programs obey the following grammar:
%\subsection{Syntax}
%\begin{figure}
 % \centering
  %\figrule
%  \footnotesize
\begin{center}
  \begin{tabular}{lll}
    programs                & $P$                &$::= D^*\ Q$\\
    query                   & $Q$                &$::= {\tt ?-}t{\tt.}$\\
    declaration             & $D$                &$::= D_\tau\ |\ D_c\ |\ D_p$\\
    type declaration        & $D_\tau$            &$::= {\tt ident:type}$\\
    constructor declaration & $D_c$              &$::= {\tt ident:}\ \tau_c$\\
    predicate declaration   & $D_p$              &$::= {\tt ident:}\ \tau_p\ C^*$\\
    constructor types       & $\tau_c$           &$::= {\tt type}\ |\ {\tt ident\ \verb+->+\ \tau_c}$\\
    predicate types         & $\tau_p$           &$::= {\tt prop}\ |\ {\tt ident\ \verb+->+\ \tau_p}$\\
    clause                  & $C$                &$::= t{\tt.}\ |\ t\ {\verb|:-|}\ G^*$\\
    goals                   & $G$                &$::=  t{\tt.}\ |\ t{\tt,}\ G^*$\\
    terms                   & $t$                &$::= {\tt Ident}\ |\ {\tt ident(}t{\tt, \ldots, }t{\tt )}$\\
  \end{tabular}
\end{center}
As our running example throughout the paper, we consider a series of arithmetic operations on the Peano representation of the natural numbers $\mathit{zero}$ and $\mathit{succ}(n)$.
To start, here is the \verb+plus+ function written in T-Prolog, with the query \verb|2 + 2 = X|.
As in standard Prolog, we will often annotate constructors and predicates with their arities as in {\tt plus/3}.
However, each identifier in T-Prolog has a unique type and thus a unique arity, so annotating identifiers with their arity is not strictly necessary.

\begin{exmp}\label{ex:tprolog-code}
\begin{verbatim}
nat : type.
zero/0 : nat.
succ/1 : nat -> nat.

plus/3 : nat -> nat -> nat -> prop.
plus(zero,X,X).
plus(succ(X),Y,succ(Z)) :-
  plus(X,Y,Z).

?- plus(succ(succ(zero)), succ(succ(zero)), X).
\end{verbatim}
\end{exmp}

There is no fundamental difference between \verb+type+ and \verb+prop+ (and in the theory they are identical): we differentiate them in the T-Prolog syntax because we find this notation intuitive and because it makes the language easier to parse.

\subsection{Semantics of T-Prolog}
In order to certify that a compiler preserves the dynamic semantics of T-Prolog programs, we must first ascertain those semantics.
As in typical Prolog, a T-Prolog program is defined as a signature of logical inference rules, and execution proceeds via depth-first proof-search under that signature, trying rules in the typical Prolog order.
Seeing as Prolog evaluation is proof search, the semantics of Prolog are often given operationally in terms of proof-search trees.
This operational treatment has the advantage that it can naturally express non-termination and the order in which rules are tried.
The disadvantage is that, in increasing operational detail, we diverge further from the world of pure logic, increasing the difficulty of verification.

For this reason, while the T-Prolog implementation does evaluate in the same order as Prolog, we do not take the operational search-based semantics of T-Prolog as canonical.
Rather, we take as the meaning of a T-Prolog program the set of formulas provable from its inference rules (Section ~\ref{sec:lf-encoding}), without regard to the order in which the proof steps are performed.
The abstractions made in this semantics are not significantly greater than those already made by a proof-search semantics.
The common insight is that a formal semantics for logic programs should exploit the close relationship to logic, ignoring implementation details that have no logical equivalent.
In both semantics, for example, it is typical to ignore Prolog's cut operator {\tt !}, which has the side effect of skipping any remaining backtracking opportunities for the current predicate, typically used as an optimization.
The cut operation is inherently about search rather than truth, informing the Prolog implementation to ignore some of the available proof rules.

Both the search semantics and provability semantics implicitly assume that all Prolog terms are finite.
The backbone of Prolog proof search is unification, and as usual for unification, finiteness cannot be taken for granted.
Allowing instances such as $X = f(X)$ to unify would result in infinite terms that do not have a logical interpretation.
In typical Prolog implementations, such terms are accepted out of the interest of performance.
In T-Prolog, we apply the standard solution of using an occurs check in unification, causing unification to fail on instances such as $X = f(X)$ whose only solutions are infinite.
This restores the close correspondence with logic, at the cost of decreased performance.

\section{The TWAM Instruction Set}
\label{sec:background-wam}
The TWAM borrows heavily from the Warren Abstract Machine, the abstract machine targeted by most Prolog implementations \cite{warren1983abstract}.
For a thorough, readable description of the WAM, see A\"it-Kaci~\cite{ait1999warren}. 
Readers familiar with the WAM may wish to skim this section and observe the differences from the standard WAM, while readers unfamiliar with the WAM will wish to use this section as a primer or even consult A\"it-Kaci's book.
In this section we present our simplified instruction set for the WAM using examples.
Notable simplifications include the usage of continuation-passing style and omission of many optimizations (with the exception of tail-call optimization in Section~\ref{sec:tco}) in order to simplify the formalism.
The description here is informal; the formal semantics are given in Section~\ref{sec:simp-op}.
\paragraph{Prolog and WAM Terminology}
The following terminology will be used extensively in this paper to describe Prolog and the WAM: 
a \emph{Prolog term} is an arbitrary combination of \emph{unification variables} $X$ combined with \emph{constructors} such as $\mathit{succ}$ and $\mathit{zero}$. What we call \emph{constructors} are generally called \emph{functors} in Prolog terminology.
We use the phrase unification variable when discussing Prolog source text and instead use \emph{free variable} to discuss WAM state at runtime.
The distinction becomes significant, e.g. because the Prolog source may specify that a parameter to some predicate is a unification variable, but at runtime the argument is a ground term.
We use the word \emph{constructor} only when discussing data and use the word \emph{predicate} to refer both to predicates for which a WAM program implements proof search and to the implementation itself.
We also say that certain WAM instructions are constructors because they construct some Prolog term, or destructors if they perform pattern matching on some Prolog term.
A \emph{structure} is the WAM representation of a constructor applied to its arguments.
A predicate consists of one or more \emph{clauses}, each of which specifies one inference rule and each of which consists of a \emph{head term} along with zero or more \emph{subgoals}.
A user interacts with the Prolog program by making a \emph{query}, which is compiled in the same way as a predicate with one clause with one subgoal.
In our discussion of TWAM programs, we consider programs with arbitrarily many predicates, one of which is designated as the query.

\paragraph{Term Destructors}
The instructions \texttt{get\_var, get\_val}, and \texttt{get\_str} are used the implementation of predicates to destruct the predicate arguments:
\begin{itemize}
\item $\texttt{get\_var}\ r_d, r_s$ reads (gets) $r_s$ into $r_d$.
  This is an unconditional register-to-register move and thus its use can be minimized by good register allocators.
  This is used to implement clauses where a unification variable is an argument.
\item $\texttt{get\_val}\ r_1, r_2$ reads (gets) $r_1$ and $r_2$ and unifies their values against each other.
  This is used to implement clauses where multiple arguments are the same unification variable.
\item $\texttt{get\_str}\ r_s, c$ reads (gets) $r_s$ and unifies it against the constructor $c$.
  For our initial examples, we will consider only the case where $c$ has no arguments.
  \texttt{get\_str} is effectively an optimized special-case of \texttt{get\_val} where we know the second unificand must be $c$.
  This is used to implement clauses where a constructor appears as a predicate argument.
\end{itemize}
For example, the Prolog predicate \texttt{both\_zero(zero, zero)}, which holds exactly when both arguments are zero, could be compiled in all of the following ways, with the naming convention
that the register for argument $i$ is named $A_i$ and the $i$'th temporary is named $X_i$:
\begin{exmp}[Implementing a Predicate]
\label{ex:both-zero}
\begin{center}
\begin{tabular}{ccc}
   \begin{minipage}{0.25\linewidth}
\begin{verbatim}
  # Implementation 1
  get_str A_1, zero/0;
  get_str A_2, zero/0;
\end{verbatim}
\end{minipage}
& \begin{minipage}{0.25\linewidth}
\begin{verbatim}
  # Implementation 2
  get_str A_1, zero/0;
  get_val A_1, A_2;
\end{verbatim}
\end{minipage}
&\begin{minipage}{0.25\linewidth}
\begin{verbatim}
  # Implementation 3
  get_var X_1, A_1;
  get_var X_2, A_2;
  get_str X_1, zero/0;
  get_val X_1, X_2;
\end{verbatim}
\end{minipage}
\end{tabular}
\end{center}
\end{exmp}
Generally speaking, Implementation 1 is most efficient, then Implementation 2, then Implementation 3.
Note that even though the Prolog predicate \texttt{both\_zero(zero, zero)} contains no unification variables, we can still use \texttt{get\_val} in the implementation,
because the unification problems $A_1 = zero, A_2 = zero$ and $A_1 = zero, A_1 = A_2$ are equivalent.
Observe that any instruction that uses unification, such as {\tt get\_val} and {\tt get\_str}, will fail if unification fails.
Should this occur, the runtime automatically backtracks if possible; backtracking is \emph{never} executed explicitly in the text of a TWAM program.

\paragraph{Term Constructors and Jumps}
To implement a query or subgoal that uses the predicate {\tt both\_zero}, we must first construct its arguments, then jump to the implementation:
\begin{itemize}
\item $\texttt{put\_var}\ r_d$  writes (puts) a \emph{new} free variable into $r_d$.
  This is used to implement passing a unification variable as an argument.
\item $\texttt{put\_val}\ r_d, r_s$ writes (puts) the value of an \emph{existing} unification variable into $r_d$, assuming it is already in $r_s$.
  This is an unconditional register move. Thus it is entirely identical to \texttt{get\_var}.
  For this reason, in our theory we will condense these into one instruction $\texttt{mov}\ r_d, r_s$ and only use the names \texttt{get\_var} and \texttt{put\_val} for consistency
  with traditional terminology in our examples.
\item $\texttt{put\_str}\ r_d, c$ writes a structure into $r_d$ using constructor $c$.
  \texttt{get\_str} is effectively an optimized special-case of \texttt{put\_val} where we are storing not an arbitrary unification variable, but specifically a constant $c$.
  This is used to implement passing a constructor as an argument to a predicate.
\item $\texttt{jmp}\ \ell^C$ passes control to the code location (address literal) $\ell^C$.
Arguments are passed through registers.
All code is in continuation passing style, and thus a continuation can be passed in through a register, which is named {\tt ret} by convention.
The queries in Example~\ref{ex:make-query} do not require returning from predicate calls, thus continuations are discussed separately.
\end{itemize}

\begin{exmp}[Making a Query]
\label{ex:make-query}
\begin{center}
 \begin{tabular}{cc}
   \begin{minipage}{0.5\linewidth}
\begin{verbatim}
# both_zero(X, X).
put_var A_1;
put_val A_2, A_1;
jmp both_zero/2;
\end{verbatim}
 \end{minipage}&
\begin{minipage}{0.5\linewidth}
\begin{verbatim}
# both_zero(X, zero).
put_var A_1;
put_str A_2, zero/0;
jmp both_zero/2;
\end{verbatim}
 \end{minipage}
   \end{tabular}
\end{center}
\end{exmp}
\paragraph{Constructors with Arguments}
We continue to use the \verb|put_str| and \verb|get_str| instructions to construct and destruct structures that contain arguments.
The difference is that when calling \verb|put_str| or \verb|get_str| with a constructor of arity $n > 0$, 
we now initiate a \emph{spine} (terminology ours) consisting of $n$ additional instructions using only the following:
\begin{itemize}
\item When $\texttt{unify\_var}\ r$ is the $i$'th instruction of a spine, it unifies the  $i'$th argument of the constructor with a \emph{new} unification variable, at register $r$.
\item When $\texttt{unify\_val}\ r$ is the $i$'th instruction of a spine, it unifies the  $i'$th argument of the constructor with an \emph{existing} unification variable, at register $r$.
\end{itemize}
The same instructions are used with both \texttt{put\_str} and \texttt{get\_str} spines.
However, at runtime a spine will execute in one of two modes, \emph{read mode} or \emph{write mode}.
\emph{Read mode} is used to destruct an existing value, meaning we are in the \texttt{get\_str $r_s, c$} and $r_s$ contains a structure whose constructor is $c$.
\emph{Write mode} is used to construct a new value, meaning we are either in $\texttt{put\_str}\ r_d, c$ or we are in $\texttt{get\_str}\ r_s, c$ but the content of $r_s$ is a free variable.
In both modes, each unification instruction processes one constructor argument:
\begin{itemize}
\item Read-mode \texttt{unify\_var} stores the next constructor argument in a register.
\item Read-mode \texttt{unify\_val} unifies the next constructor argument with the content of a register.
\item Write-mode \texttt{unify\_var} allocates a free variable as constructor argument, storing it also in a register.
\item Write-mode \texttt{unify\_val} uses the content of a register as constructor argument.
\end{itemize}

For example, the Prolog predicate \texttt{same\_pos(succ(X), succ(X))} which holds when the arguments are the same positive number, can be implemented and used as follows:

\begin{exmp}[Predicates with Prolog Spines]
\begin{center}
 \begin{tabular}{cc}
   \begin{minipage}{0.5\linewidth}
\begin{verbatim}
# Implementation
get_str A_1, succ/1;
  unify_var X_1;
get_str A_2, succ/1;
  unify_val X_1;
\end{verbatim}
 \end{minipage}&
\begin{minipage}{0.5\linewidth}
\begin{verbatim}
# Query same_pos(succ(X),succ(Y))
put_str A_1, succ/1;
  unify_var X_1;
put_str A_2, succ/1;
  unify_var X_1; # Overwrites X_1
jmp same_pos/2;
\end{verbatim}
 \end{minipage}\\[0.5in]
\begin{minipage}{0.5\linewidth}
\begin{verbatim}
# Query same_pos(succ(X),succ(X))
put_str A_1, succ/1;
  unify_var X_1;
put_str A_2, succ/1;
  unify_val X_1; # Reads X_1
jmp same_pos/2;
\end{verbatim}
\end{minipage}&
\begin{minipage}{0.5\linewidth}
\begin{verbatim}
# Query same_pos(zero,succ(succ(zero)))
put_str A_1, zero/0;
put_str X_1, zero/0; # Z = 0
put_str X_2, succ/1; # Y = 1
  unify_val X_1;
put_str A_2, succ/1; # X = 2
  unify_val X_2;
jmp same_pos/2;
\end{verbatim}
 \end{minipage}
 \end{tabular}\\
\end{center}
\end{exmp}
The last example demonstrates a compilation technique known as \emph{flattening}: The unification problem $X = succ(succ(zero))$ is equivalent to the problem $X = succ(Y), Y = succ(Z), Z = zero$.
This allows us to implement nested structures such as $succ(zero)$ or $succ(succ(succ(zero)))$ by introducing intermediate variables.
Thus each spine need only introduce one structure, and nested structures are reduced to the one-structure case by flattening.

\paragraph{Continuations, Closures, and Halting}
Prolog proof search can be structured using success and failure continuations.~\cite{Elliott91asemi-functional} 
When a predicate has multiple clauses, failure continuations are used to remember alternate clauses and implement backtracking.
When a clause has multiple subgoals, success continuations are used to remember the remaining subgoals.
In our system, success continuations can be stored in registers and passed to predicates, typically in a register named {\tt ret}, whereas failure continuations are stored in the trail.
Both success and failure continuations can access an environment value (generally a tuple) through the register {\tt env}.
Tuples are like structures, but can contain closures and cannot be unified.
The entry-point of a TWAM program is a top-level query, which specifies an initial continuation that terminates the program in success.
If all clauses fail, then the runtime will automatically report that the program failed.

\begin{itemize}
\item $\texttt{close}\ r_d, r_e, \ell^C$ places a new closure in $r_d$ containing an environment read from $r_e$.
 When that closure is invoked, control will pass to $\ell^C$ and the environment will be placed in a special-purpose register named \texttt{env}.
 This is used to construct success continuations.
\item $\texttt{push\_bt}\ r_e, \ell^C$ (push backtracking point) creates a new failure continuation.
 When that continuation is invoked, control will pass to $\ell^C$ and the environment will be placed in \texttt{env}.
 Note that \texttt{push\_bt} does not take a destination register: a stack of failure continuations is stored implicitly in the machine state,
 and they are only ever invoked implicitly, when unification instructions like {\tt get\_val} fail.
\item $\texttt{put\_tuple}\ r_d, n$ begins a \emph{tuple spine} of length $n$ which will put a new tuple in $r_d$.
  All following instructions of the tuple spine are \texttt{set\_val}.
\item $\texttt{set\_val}\ r_s$ copies $r_s$ in as the next tuple element.
\item $\texttt{proj}\ r_d, r_s, i$ copies the i'th element of the tuple at $r_s$ into $r_d$.
\item $\texttt{succeed}$ immediately terminates the program and indicates that the initial query has succeeded.
 (At this point, the runtime system will print out the solution to the query.)
\end{itemize}

As an example, consider implementing and calling the predicate $X + Y = Z$ with two clauses: {\tt plus(zero,X,X)} and {\tt plus(succ(X),Y,succ(Z)) :- plus(X,Y,Z)}:
\begin{exmp}[Implementing {\tt plus}]
\label{ex:plus-wam}
\centering
\begin{tabular}{cc}
\begin{minipage}{0.5\linewidth}
 \begin{verbatim}

# Entry point to plus, implements the
# case plus(zero,X, X) and tries 
# plus-succ on failure
plus-zero/3:
  put_tuple X_1, 3;
    set_val A_1;
    set_val A_2;
    set_val A_3;
  push_bt X_1, plus-succ/3;
  get_str A_1, zero/0;
  get_val A_2, A_3;
  jmp ret;
\end{verbatim}
\end{minipage}&
\begin{minipage}{0.5\linewidth}
\vspace{-0.25in}
\begin{verbatim}
# plus(succ(X), Y, succ(Z)) :- plus(X,Y,Z).
plus-succ/3:
  proj A_1, env, 1;
  proj A_2, env, 2;
  proj A_3, env, 3;
  get_str A_1, succ/1;
    unify_var A_1;
  get_str A_3, succ/1;
    unify_var A_3;
  jmp plus-zero/3;
\end{verbatim}
\end{minipage}
\end{tabular}
\end{exmp}

\begin{exmp}[Calling {\tt plus}]
%\label{ex:plus-wam}
\centering
\begin{tabular}{cc}
\begin{minipage}{0.5\linewidth}
\begin{verbatim}
    init-cont/0:
      succeed;    
\end{verbatim}
\end{minipage}
&\begin{minipage}{0.5\linewidth}
\vspace{0.08in}
\begin{verbatim}
# plus(succ(zero), succ(zero), X)
query/0:
  put_tuple X_1, 0;
  close ret, X_1, init-cont/0;
  put_str X_2, zero/0;
  put_str A_1, succ/1;
    unify_val X_2;
  put_str A_2, succ/1;
    unify_val X_2;
  put_var A_3;
  jmp plus-zero/3;
\end{verbatim}
\end{minipage}
\end{tabular}
\end{exmp}
In this example, \texttt{plus-zero/3} is the entry point for addition, and implements the base case.
Because \texttt{plus-zero/3} is not the last case, it constructs a failure continuation which tries the \texttt{plus-succ/3} case if an instruction fails.
This requires remembering the environment, implemented by creating a tuple.
In the example query, the first invocation of \texttt{plus-zero/3} will fail on the \texttt{get\_str} instruction because $A_1$ contains $succ(zero),$ not $zero$,
causing \texttt{plus-succ/3} to run (which will succeed after another call to \texttt{plus-zero}).
\texttt{plus-succ} contains several optimizations. 
The final subgoal of a clause can always apply tail-call optimization, so no success continuation is necessary.
Furthermore, it carefully avoids the use of intermediate registers. For example, when reading the argument $succ(X),$ the variable $X$ is written directly into $A_1$ to prepare for the recursive call.
The query $plus(succ(zero),succ(zero),X)$ must specify an initial continuation, which simply reports success. 
Because the success continuation is trivial, the empty tuple suffices as its environment.

\paragraph{Runtime State}
The runtime representation of a TWAM program differs from that of a WAM program, following the differences in their instruction sets.
Both languages have a fixed \emph{code section} containing the TWAM program text and a variable-sized \emph{heap}, which maintains all Prolog terms in a union-find data structure to enable fast unification.
The most significant difference is that the TWAM machine state does not have a stack, but instead allocates success continuations on the heap and allows them to be garbage collected.
Failure continuations, however, are stored in a separate area called the \emph{trail} as in standard WAM. 
In addition to storing a closure created with \texttt{push\_bt}, the trail automatically keeps track of all state changes which might have to be reverted during backtracking.
Traditional descriptions of the WAM contain a \emph{push-down list} or PDL area, which is used in unification to store a temporary list of unification subproblems.
Because this data structure is used only during unification, we found it easier to express the PDL merely as a part of unification and not as a permanent part of the state.

\paragraph{Differences Between WAM and TWAM Instruction Sets}
The key difference between WAM and TWAM is that the TWAM implements predicate calls and backtracking with success and failure continuations, while WAM implements both by maintaining a custom stack discipline, whose invariants are non-trivial.
The use of CPS significantly simplifies the formalism and unifies several instructions that are distinct in traditional WAM:
\begin{itemize}
\item Environments in TWAM are expressed as tuples with the instructions {\tt put\_tuple}, {\tt set\_val}, and {\tt proj}, which  replace {\tt allocate} and {\tt deallocate}.
\item The {\tt jmp} instruction of TWAM unifies  {\tt call}, {\tt execute}, and {\tt proceed} from WAM.
\item The {\tt push\_bt} instruction of TWAM replaces {\tt try\_me\_else}, {\tt retry\_me\_else}, and {\tt trust\_me} from WAM.
\item The {\tt succeed} instruction is added in TWAM for stylistic purposes; WAM reuses {\tt proceed} for this purpose.
\item The unification and spinal instructionns of TWAM correspond directly to WAM.
\item TWAM omits several optimizations such as cut and case analysis.
\end{itemize}

\section{The Simply-Typed WAM (SWAM)}
\label{sec:simple-wam}
The core contributions of this work are the design, metatheory, and implementation of a type system for the TWAM strong enough to certify compilation.
The certification guarantees provided by the dependently-typed TWAM in Section \ref{sec:dependent-wam} require significant complexity in the type-system.
In this section, we ease the presentation of that system by first presenting its simply-typed variant, the SWAM.
We prove progress and preservation for SWAM, which constitute a safety property analogous to those of other strongly-typed abstract machines such as the typed assembly language TAL \cite{Morrisett:1999:SFT:319301.319345}.
In Section~\ref{sec:dependent-wam}, this is subsumed by progress and preservation for TWAM, which is strong enough to certify partial dynamic correctness.

\subsection{Typechecking SWAM}
\label{sec:simple-wam-istat}
The text of a SWAM program is structured as a code area $C$ mapping code locations $\ell^C$ to code values $\code[\myG](I)$.
A code value is a single basic block $I$ annotated with a \emph{register file type} (rftype) $\myG$ which indicates, for each register $r_i$, the type expected by $I$. One of those code values is designated as the query (a predicate with one clause and one subgoal), which is the entry point of the program.
The type $\tau$ assigned to each register is either an \emph{atomic type} $a$ representing a Prolog term, a \emph{continuation type} $\neg \Gamma$ representing a closure that expects the registers to obey the types in $\myG$, or a \emph{tuple type} $\cross{\vec\tau}$, where the elements $\vec\tau$ can freely mix atomics and continuations.
Here $\vec{\tau}$ is an abbreviation for the sequence $\tau_1,\ldots,\tau_n$; similar abbreviations will be used extensively throughout the paper.

The main typing judgement in SWAM is $\myG \vdash_{\Sigma;\Xi} I \ok$, which says the basic block $I$ is well-typed assuming the registers obey $\myG$ initially,
and given signatures $\Sigma,\Xi$ which assign types to every constructor $c$ and code location $\ell^C,$ respectively.
We omit the subscripts $\Sigma;\Xi$ on rules where they are not relevant.
Throughout the paper, the notation $\thup{\Gamma}{r}{\tau}$ refers to updating the type of $r$ in $\Gamma$ to be $\tau$.
Analogous notation will be used for register values, etc.
Throughout this section, we alternate between inference rules for typechecking instructions and their descriptions.

%\begin{center}
%  $\boxed{\textrm{Instructions}\ i}$
%\end{center}
\begin{center}
{\footnotesize  \begin{tabular}{cccc}
    \infer[\textsc{Succeed}]{\GE \succeed;I \ok}{}
   &\infer[\textsc{PutVar}]{\GE \putvar{a}{r}; I \ok}
    {\hups{\Gamma}{r\colon{}a}\ent I \ok}
%   &\infer{\GE \getval {r_1}{r_2}; I \ok} {\Gamma(r_1) = a & \Gamma(r_2) = a}
   &\infer[\textsc{GetVal}]{\GE \getval {r_1}{r_2}; I \ok} { \Gamma(r_1) = a & \Gamma(r_2) = a & \Gamma \ent I \ok}
  \end{tabular}}

\end{center}
\begin{center}
  \footnotesize\begin{tabular}{cc}
    \infer[\textsc{Jmp}]{\GE_{\Sigma;\Xi}\jmp op, I\ok}{\Gamma\vdash_{\Sigma;\Xi} op : \neg\Gamma' & \SE \myG' \leq \myG}
&
  \infer[\textsc{Mov}]{\GE \mov{r_d}{r_s}; I \ok}
  {\myG(r_s) = \tau & \thup{\Gamma}{r_d}{\tau} \ent I \ok}
  \end{tabular}
\end{center}
\begin{center}
{\footnotesize\begin{tabular}{ccc}
  \infer[\textsc{PushBT}]{\GE_{\Sigma;\Xi} \branch r_e, \ell^C; I \ok}{\GE I \ok &
    \Gamma(r_e) = \tau & \Xi(\ell^C) = \neg\heap{\texttt{env}\colon{}\tau}}
&
  \infer[\textsc{Close}]{\GE \close r_d, r_e, \ell^C; I \ok}
  {\deduce{\GE \ell^C \colon{} \neg \thup{\myG'}{\texttt{env}}{\tau}} {\Gamma(r_e)= \tau &
      \thup{\myG}{r_d}{\neg \myG'} \ent I \ok}}
\end{tabular}}
\end{center}

\begin{itemize}
\item {\tt succeed} always typechecks, and is typically the last instruction of its block.
\item ${\tt put\_var}\ [a]r$ allocates a free variable of type $a$ in $r$, thus updating $r$ to type $a$. We write the annotation $[a]$ in brackets to emphasize that it is used only for typechecking.
\item ${\tt get\_val}\ r_1, r_2$ unifies $r_1$ and $r_2$, so they must have the same (atomic) type.
\item ${\tt jmp}\ op$ transfers control to $op$, which in the general case is either a location $\ell^C$ (used in predicate calls) or register $r$ (used in returns, by convention generally named {\tt ret}).
The judgement $\cdot\vdash\Gamma'\leq\Gamma$ means $\forall r\in\dom{\Gamma'}. \Gamma'(r) = \Gamma(r)$ ($\Gamma'$ may omit some registers of $\myG$).
The judgement $\Gamma\vdash_{\Sigma;\Xi} op\colon{}\tau$ has the rules:
\begin{center}
\begin{tabular}{cc}
 $\infer[\textsc{Op-}\ell^C]{\myG\vdash_{\Sigma;\Xi} \ell^C\colon{}\tau}
        {\Xi(\ell^C) = \tau}$
& $\infer[\textsc{Op-}r]{\myG\vdash_{\Sigma;\Xi} r\colon{}\tau}
        {\myG(r) = \tau}$
\end{tabular}
\end{center}
\item${\tt mov}\ r_d,r_s$ copies $r_s$ into $r_d$.
\item${\tt push\_bt}\ r_e,\ell^C$ installs the failure continuation $\ell^C$ in the trail along with the environment from $r_e$, which will be in {\tt env} upon invocation of $\ell^C$.
\item${\tt close}\ r_d,r_s,\ell^C$ is analogous, but stores the resulting success continuation in $r_d$ before proceeding.
\begin{center}
{\footnotesize\begin{tabular}{cc}
 \infer[\textsc{Proj}]{\GE\proj r_d, r_s, i; I \ok}{\myG(r_s) = \cross{\vec\tau} & \thup{\myG}{r_d}{\tau_i}\vdash I\ok & \text{(where $i \leq |\vec\tau|)$}}
&\infer[\textsc{PutTuple}]{\GE\puttuple r_d, n; I \ok}{\GE I:_t (\vec\tau\to \{r_d:\cross{\vec\tau}\}) & \text{(where $n = |\vec\tau|$)}}
\end{tabular}}
\end{center}
\item${\tt proj}\ r_d,r_s,i$ puts the $i$'th element the tuple $r_s$ into $r_d$, typechecking only if $r_d$ has length at least $i$. 
  Here $\vec\tau$ is a sequence of types, one for each element.
\item${\tt put\_tuple}\ r_d,n$ initiates a \emph{tuple spine} of length $n$ with destination $r_d$.
The remainder of the tuple spine is checked using the auxiliary \emph{tuple spine typing} judgement $\GE I:_t (\tau_1 \to \cdots \to \tau_n \to {\tt Post})$, 
where {\tt Post} is a singleton rftype $\{r_d\colon{}\cross{\vec\tau}\}$.
The auxiliary judgement $\GE I:_t (\tau_1 \to \cdots \to \tau_n \to {\tt Post})$ should be read as 
``the next $n$ instructions construct tuple elements of type $\tau_i,$ with postcondition ${\tt Post} \leq \myG$, and all remaining instructions typecheck''.
The typing rules for the spine typing judgement are given in Section~\ref{sec:spine-typing}.
\begin{center}
{\begin{tabular}{cc}
  \infer[\textsc{GetStr}]{\GE \getstr {c} {r}; I \ok} {\deduce{\GE I:_s(\vec a \to  \{\})}{\Sigma(c) = \vec a \to a &
      \Gamma(r) = a}}
&
  \infer[\textsc{PutStr}]{\GE \putstr {c} {r}; I \ok}
  {\deduce{\thup{\Gamma}{r}{a} \ent I:_s (\vec a \to \{\})}
  {\Sigma(c) = \vec a \to a}}
\end{tabular}}
\end{center}
\item${\tt get\_str}\ c,r$ and ${\tt put\_str}\ c,r$ both initiate \emph{Prolog spines} which are checked with \emph{Prolog spine typing} judgement $\GE I:_s (a_1 \to \cdots \to a_n \to {\tt Post})$. Unlike tuple spines, Prolog spines contain only atomic types, and in SWAM always have an empty postcondition ${\tt Post} = \{\}$.
Intuitively one might expect ${\tt Post} = \{r : a\}$,  for {\tt put\_str}, but we choose to update the type of $r$ at the \emph{beginning} of the spine instead of the end, because this leaves {\tt put\_str} symmetric more symmetric with {\tt get\_str}
(a free variable is stored at $r_d$ until the spine completes to ensure type safety). 
\end{itemize}
\subsection{Spine Typing}
\label{sec:spine-typing}
When constructing compound data structures (either tuples or structures), we wish to know that the data structure has the intended number of arguments, each with the intended type.
For this reason, we apply the auxilliary typing judgements for tuple and Prolog spines.
Each spinal instruction produces one element, and so each rule application checks the type of one element.
Consider the rules for the tuple spine judgement $\GE I:_t (\tau_1 \to \cdots \to \tau_n \to {\tt Post})$:
\begin{center}
\begin{tabular}{cc}
 \infer[\textsc{TSpine-SetVal}]{\GE\setval r; I :_t (\tau\to J)}{\myG(r) = \tau & \GE I:_t J}
&
 \infer[\textsc{TSpine-End}]{\GE I :_t {\tt Post}}
 {\Gplus{\tt Post}\ent I \ok}
\end{tabular}
\end{center}
The rule for (TSpine-SetVal) says that each {\tt set\_val} contributes one element.
The rule (TSpine-End) resumes the main typing judgement $\GE I \ok$ when 
says that when a tuple is complete, we store the tuple according to {\tt Post} and resume normal typechecking. Specifically, $\Gplus{\myG'}$ is the rftype such that $\Gplus{\myG'}(r) = \myG'(r)$ for $\dom{\myG'}$ and $\Gplus{\myG'}(r) = \myG(r)$ otherwise.

Prolog spines have their own auxilliary judgment, $\GE I:_s (a_1 \to \cdots \to a_n \to {\tt Post})$.
The rule for ending a Prolog spine is analogous to \textsc{(TSpine-End)}.
The elements of a Prolog spine can be produced either by ${\tt unify\_val}\ r$ or ${\tt unify\_var}\ r$.
The ${\tt unify\_val}\ r$ instruction which requires the argument register type to match the constructor argument, 
whereas ${\tt unify\_var}\ r$ creates a new unification variable of the correct type, which appears both in $r_d$ and as a constructor argument.
\begin{center}
{\begin{tabular}{cc}
  \infer[\textsc{UnifyVal}]{\GE \unifyval r; I:_s (a \to J)} 
   {\Gamma(r) = a & \GE I:_s J}
&
  \infer[\textsc{UnifyVar}]{\GE \unifyvar r; I:_s (a \to J)}
  {\thup{\Gamma}{r}{a} \ent I:_s J}
\end{tabular}
}\end{center}

\subsection{State Representation and Invariants}
\label{sec:simp-rep-inv}
Following the traditional description of the WAM, the essential parts of the SWAM \emph{machine state} include the \emph{code section} $C$, \emph{heap} $H$, and \emph{trail} $T$ (backtracking structure). $H$ and $C$ are often considered together as the \emph{store} $S = (C,H)$.
Locations in the code section are written $\ell^C$ and locations in the heap are written $\ell^H$.
Where both are acceptable we write $\ell$. The notation $S(\ell)$ denotes either $H(\ell^H)$ or $C(\ell^C)$ as appropriate. 
We additionally have an explicit representation of the \emph{register file} $R$ and we represent the instruction pointer as the sequence $I$ of remaining instructions in the current basic block.
Machines also support three spinal execution modes: read spines, write spines, and tuple (write) spines.
In short, machine states are described by the syntax:
\[m ::= (T,S,R,I)\ |\ \mread(T,S,R,I,\vec\ell^H)\ |\ \mwrite(T,S,R,I,c,\ell^H,\vec\ell^H)\ |\ \twrite(T,S,R,I,r,n,\vec\ell^H)\]
We first consider the following typing invariant for normal states $(T,S,R,I)$ in depth and then revisit the additional invariants for spinal states:
\[\infer[\textsc{Mach}]{\cdot \ent_{\Sigma;\Xi} (T,S,R,I) \ok}
{\deduce{\SE S:(\Xi;\Psi) \hskip 0.1in \PE R \colon{} \Gamma \hskip 0.1in \GE I \ok}
{S \ent T \ok}}
\]
As in Section~\ref{sec:simple-wam-istat}, all judgments are parameterized by signatures $\Sigma$ and code section types $\Xi,$ which are elided when irrelevant.
The code section is well-typed when each basic block is well-typed according to the rules of Section \ref{sec:simple-wam-istat}.
The code section is allowed to be mutually recursive:
\[\infer[\textsc{Code-Sec}]{\cdot\ent_{\Sigma;\Xi} \heap{v^C_1, \ldots v^C_n} \colon{} \Xi}
        {\cdot\ent_{\Sigma;\Xi} v^C_1 \colon{} \tau_1\ \cdots\ \cdot\ent_{\Sigma;\Xi} v^C_n \colon{} \tau_n & \text{(where $\Xi = \heap{v^C_1 \colon{} \tau_1, \ldots, v^C_n \colon{} \tau_n}$)}}\]

Heap types are written $\Psi$ and are analogous to rftypes.
As in rftypes we write $\thup{\Psi}{\ell^H}{\tau}$ when updating the type of $\ell^H$.
We also write $\thext{\Psi}{\ell^H}{\tau}$ when adding a \emph{fresh} location $\ell^H$ with type $\tau$, or
$\heap{}$ for an empty heap or empty heap type.
We prohibit cycles in the heap because it simplifies implementing SWAM and simplifies the dependent type system of Section \ref{sec:dependent-wam} even further.
Specifically, a typing derivation $\mathcal{D} : (\SE H : \Psi)$ serves as a witness that $H$ is acyclic,
because $\mathcal{D}$ implicitly specifies a topological sorting on $H$: the rules below state that each value may
only refer to preceding values.
However, $\Psi$ need not assign a type to all values in H, so long as those values without types are never accessed.
This technicality is useful when reasoning about backtracking as in Lemma~\ref{lem:trail-up}.
\begin{center}
  \begin{tabular}{cc}
    \infer[\textsc{Heap-Nil}]{\SE H\colon{}\heap{}}{}
    &\infer[\textsc{Heap-Cons}]{\SE H\colon{}\thext{\Psi}{\ell^H}{\tau}}
    {\SE H\colon{}\Psi & H(\ell^H) = v^H & \PE v^H\colon{}\tau & \ell^H \notin\dom{\Psi}}
%    &\infer{\SE \hext{H}{\ell^H}{v^H}\colon{}\thext{\Psi}{\ell^H}{\tau}}
%    {\SE H\colon{}\Psi & \PE v^H\colon{}\tau & \ell^H \notin\dom{H}}
  \end{tabular}
\end{center}
Values are divided into \emph{heap values} $v^H$ which are arbitrarily large and \emph{word values} $w$ which are fixed size.
In SWAM, words are always heap locations $w ::= \ell^H$. 
The heap values $v^H$ follow the syntax:
\[v^H ::= \cwsh\ |\ \free{}[a]\ |\ \bound{\ell^H}\ |\ \close(w_{env}, \ell^C)\ |\ \langle{}w_1,\ldots,w_n\rangle\]
The values $\cwsh$ and $\free{}[a]$ introduce structures and free variables in Prolog terms, respectively.
The type annotation $a$ in $\free{}[a]$ is merely a convenience for the metatheory and not used at runtime (i.e. SWAM and TWAM support type erasure).

Combined with pointers $\bound{\ell^H}$, these values provide a union-find data structure within the heap, used by SWAM's unification algorithm. The $\bound{\ell^H}$ pointers are merely an artifact of that algorithm and are semantically equivalent to $\ell^H$.
In addition to Prolog terms, the heap contains closures $\close(w_{env}, \ell^C)$ where the machine word $w_{env}$ is the environment for executing $\ell^C$, as well as tuples $\langle w_1, \ldots, w_n\rangle$.
The typing invariants for heap values are:
\begin{center}
\begin{tabular}{ll}
\infer[\textsc{HV-Close}]{\PE \close(w_{env}, \ell^C) \colon{}\neg \myG}
        {\PE w_{env} \colon{} \tau &  \PE\ell^C\colon{}\neg\thup{\myG}{\texttt{env}}{\tau}}
&\infer[\textsc{HV-Tup}]{\PE\langle w_1,\ldots,w_n\rangle\colon{}\cross{\tau_1,\ldots,\tau_n}}
       {\PE w_1 : \tau_1&\cdots&\PE w_n : \tau_n}
\end{tabular}
\end{center}
\begin{center}
\begin{tabular}{lll}
\infer[\textsc{HV-Bound}]{\PE \bound{\ell^H}\colon{}a}{\PE \ell^H\colon{}a}
&\infer[\textsc{HV-Str}]{\PE_{\Sigma;\Xi} \cwsh\colon{}a}
        {\Sigma(c) = \vec a \to a & \PE \ell_i^H\colon{}a_i}
 &\infer[\textsc{HV-Free}]{\PE \free{}[a] \colon{}a}{}
\end{tabular}
\end{center}

\paragraph{Register Typing}
A register file $R$ simply maps registers $r_i$ to word values $w_i$ and is well-typed when all $w_i$ are well-typed.
Because the only word values in SWAM are heap locations, it suffices to consult the heap type $\Psi$:
\begin{center}
\begin{tabular}{cc}
\infer[\textsc{RF}]{\PE \heap{r_1 \hook w_1, \ldots, r_n \hook w_n} : \heap{r_1 :\tau_1, \ldots, r_n : \tau_n}}
        {\PE w_1 : \tau_1  & \cdots & \PE w_n : \tau_n}&
\infer[\textsc{WV-}\ell^H]{\PE \ell^H \colon{} \tau}{\Psi(\ell^H) = \tau}
\end{tabular}
\end{center}

\paragraph{Trail Typing}
When Prolog backtracks because a clause failed, it must revert all changes made by the failed clause.
The only such change is the binding of free variables during unification, 
thus it suffices to record bindings when they occur and revert them during backtracking.
The \emph{trail} is the data structure that records these variables.
In traditional presentations of the WAM, the trail contains variable addresses only and separate \emph{choice point} records in the call stack contain the failure continuation.
For our presentation, it simplified the formalism to store the continuation inline.
The trail is given as a list of \emph{trail frames} $(t, w_{env}, \ell^C)$ where $t$ is a list of heap locations (in the theory, annotated with types as in $\ell^H:a$), where we write the list of $t_i$ as $t_1::\cdots::t_n::\epsilon$.  The environment is $w_{env}$ and $\ell^C$ is the failure continuation. The function $\unwind(S,t)$ describes the process of backtracking one trail frame, i.e. $\unwind(S,(\ell^H:a)::t) = \unwind( \hup{S}{\ell^H}{\free{}[a]},t)$ and $\unwind(S, \epsilon) = S$.
\begin{center}
\begin{tabular}{cc}
\infer[\textsc{Trail-Nil}]{S \ent \epsilon \ok}{}
&\infer[\textsc{Trail-Cons}]{S \ent_{\Sigma;\Xi}(t,w_{env},\ell^C)::T \ok}
        {\deduce
          {\ent S'\colon{}(\Xi,\Psi')\hskip 0.1in \Psi' \ent w_{env}\colon{}\tau
            \hskip 0.1in \Psi'\ent\ell^C\colon{}\neg\heap{\texttt{env}\colon{}\tau}}
          {\unwind(S,t)= S' & S' \ent T \ok}}
\end{tabular}
\end{center}

\paragraph{Special Mode Invariants}
When the machine is in read or write mode, it maintains additional data. Read mode maintains a list of arguments not yet read, while the write modes maintain lists of arguments written so far with destination registers or locations. Prolog write mode also tracks the constructor being applied while tuple write mode tracks the number of elements left to be written in the tuple. In each case additional invariants are required, as given in the judgements $\PE\vec\ell^H \reads \vec a$, $\PE (\vec \ell^H, \ell^H,c) \writes (\vec a_2\to \{\})$ and $\PE (n,r,\vec\ell^H) \writes (\vec \tau_2 \to \{r\colon{}\cross{\vec\tau_1 \vec \tau_2}\})$. 
In each case the invariants ensure that the type of the constructor or tuple in question is consistent with both the values computed so far and the remaining spinal instructions.
\begin{center}{\footnotesize
  \begin{tabular}{cc}
\infer[\textsc{Mach-Read}]{\cdot \ent_{\Sigma;\Xi} \mread(T,S,R,I, \vec \ell^H) \ok}
        {\deduce{\PE R\colon{}\Gamma\hskip 0.1in\GE I:_sJ\hskip 0.1in \PE\vec\ell^H\reads J}
          {S\ent T\ok & \ent S\colon{}(\Xi;\Psi)}}&
\infer[\textsc{Reads}]{\PE\vec\ell^H\reads(\vec a\to \{\})}
{\PE\ell^H_i\colon{}a_i}
  \end{tabular}}
\end{center}
\begin{center}{\footnotesize
  \begin{tabular}{cc}
\infer[\textsc{Mach-Write}]{\cdot \ent_{\Sigma;\Xi} \mwrite(T,S,R,I,c,\ell^H, \vec \ell^H) \ok}
        {\deduce{ \PE R\colon{}\Gamma\hskip 0.1in\GE I:_s J \hskip 0.1in \PE (\vec \ell^H, \ell,c) \writes J}
          {S \ent T \ok&\SE S\colon{}(\Xi;\Psi)}}&
\infer[\textsc{Writes}]{\PE_{\Sigma;\Xi} (\vec\ell^H,\ell^H,c) \writes (\vec a_2\to \{\})}
        {\Sigma(c) = \vec a_1 \to \vec a_2 \to a & \PE \vec\ell^H\colon{}\vec a_1 & \Psi(\ell^H) = a}
  \end{tabular}}
\end{center}
\begin{center}{\footnotesize
  \begin{tabular}{cc}
\infer[\textsc{Mach-TWrite}]{\cdot \ent_{\Sigma;\Xi} \twrite(T,S,R,I, \vec w,r,n) \ok}
        {\deduce{\PE R\colon{}\Gamma\hskip 0.1in\GE I:_t J \hskip 0.1in \PE (\vec w, r, n) \writes J }
          {S\ent T\ok&\ent S:(\Xi;\Psi)}}&
\infer[\textsc{TWrites}]{\PE (n,r,\vec\ell^H) \writes (\vec \tau_2 \to \{r\colon{}\cross{\vec\tau_1 \vec \tau_2}\})}
{\PE \vec \ell^H \colon{} \vec \tau_1 & |\vec \tau_2| = n}
  \end{tabular}}
\end{center}

\subsection{Operational Semantics}
The dynamic semantics of SWAM are given as a small-step operational semantics.
We begin with an informal example trace executing the query \verb|?- plus(X,zero,succ(zero))| using the \verb|plus| function of Example \ref{ex:plus-wam} before
developing the semantics formally.

%We present the operational semantics by example here, with the full semantics given in Appendix B.
%We give an  evaluation trace of the query . 
For each line we describe any
changes to the machine state, i.e. the heap, trail, register file, and instruction pointer. As with the WAM, the TWAM
supports special execution modes for spines: \emph{read mode} and \emph{write mode}. When the program enters read mode,
we annotate that line with the list $\ell^H{}s$ of variables being read, and when the program enters write mode we annotate it
with the constructor $c$ being applied, the destination location $\ell^H$ and the argument locations $\vec\ell^H$. The final instruction
of a write-mode spine is best thought of two evalution steps, one of which constructs the last argument of the constructor and one of which
combines the arguments into a term.

%TODO: typeset nice
%TODO: better empty lists
% EVALEXMP
\newcommand{\indentone}{\hspace{0.1in}}
\newcommand{\indenttwo}{\hspace{0.2in}}
{\footnotesize
\begin{longtable}{ll}
  Code & Change\\
  \multicolumn{2}{l}{{\tt \# Query} ${\tt plus}(X,{\tt zero},{\tt succ(zero)})$}\\
  \multicolumn{2}{l}{{\tt query}$~\mapsto\code[\heap{}]($}\\
    \indentone{\tt put\_var}$\ r_1;$    &$H \leftarrow \hext{H}{\ell_1}{\free{}[{\tt nat}]}$, $R \leftarrow \hup{R}{r_1}{\ell_1}$\\
    \indentone{\tt put\_str}$\ r_2$, {\tt zero/0};      &$H \leftarrow \hext{H}{\ell_2}{\free{}[{\tt nat}]}, R \leftarrow \hup{R}{r_2}{\ell_2}, c={\tt zero}$\\
                             &$\ell = \ell_2, \vec\ell = \langle\rangle,$\\
                             &$H \leftarrow \hup{H}{\ell}{c\langle\vec\ell\rangle}$\\
    \indentone{\tt put\_str}$\ r_3$, {\tt succ/1};      &$H \leftarrow \hext{H}{\ell_3}{\free{}[{\tt nat}]}, R \leftarrow \hup{R}{r_3}{\ell_3}, c={\tt succ}$\\
                             &$\ell = \ell_3, \vec\ell = \langle\rangle$\\
    \indenttwo{\tt unify\_val}$\ r_2;$            &$\vec\ell \leftarrow \langle\ell_2\rangle, H \leftarrow \hup{H}{\ell}{c\langle\vec\ell\rangle}$\\
    \indentone{\tt put\_tuple}$\ r_4,0;$          &$H \leftarrow \hext{H}{\ell_4}{\langle\rangle},R\leftarrow\hup{R}{r_4}{\ell_4}$\\
    \indentone{\tt close}\ {\tt ret}$,r_4,{\tt success/0};$  &$H \leftarrow \hext{H}{\ell_5}{\close(\ell_4,{\tt success})},R\leftarrow\hup{R}{\tt ret}{\ell_5}$\\
    \indentone{\tt jmp\ }\verb|plus-zero/3|$;$         &$I \leftarrow C({\verb|plus-zero|})$ \\
  )&\\
~\\
  \multicolumn{2}{l}{{\tt plus-zero/3}$~\mapsto\code[{r_2:{\tt nat}, r_2:{\tt nat},r_3:{\tt nat},r_4:\neg\heap{}}]($}\\
    \indentone{\tt put\_tuple} $r_4,4;$          &$\vec\ell=\langle\rangle$\\
    \indenttwo   {\tt set\_val} $r_2;$           &$\vec\ell=\langle\ell_1\rangle$\\
    \indenttwo   {\tt set\_val} $r_2;$           &$\vec\ell=\langle\ell_1,\ell_2\rangle$\\
    \indenttwo   {\tt set\_val} $r_3;$           &$\vec\ell=\langle\ell_1,\ell_2,\ell_3\rangle$\\
    \indenttwo   {\tt set\_val} ${\tt ret};$          &$\vec\ell=\langle\ell_1,\ell_2,\ell_3,\ell_5\rangle$\\
                             &$H \leftarrow \hext{H}{\ell_6}{\langle\vec\ell\rangle},R\leftarrow\hup{R}{r_4}{\ell_6}$\\
    \indentone{\tt push\_bt} $r_4, {\verb|plus-succ/3|};$ &$T \leftarrow (\ell_6, {\verb|plus-succ/3|}, {\tt nil})::{\tt nil}$\\
    \indentone{\tt get\_str} $r_2, {\tt zero/0};$      &$c={\tt zero}, l = \ell_1, \vec\ell = \langle\rangle$\\
                             &$H \leftarrow \hup{H}{\ell_1}{\tt zero}, T \leftarrow (\ell_6,{\verb|plus-succ/3|}, \ell_1)::\langle\rangle$\\
    \multicolumn{2}{l}{\tt\# This instruction fails, backtrack to plus-succ/3}\\
    \indentone{\tt get\_val}$\ r_2, r_3;$          &$T \leftarrow {\tt nil}, I \leftarrow {\verb|plus-succ/3|}, H \leftarrow \hup{H}{\ell_1}{\free{}[\tt nat]}$\\
    \indentone{\tt jmp\ ret};\\
  )&\\
~\\
  \multicolumn{2}{l}{{\tt plus-succ/3}$~\mapsto \code[\{{\tt env}:\cross{{\tt nat},{\tt nat},{\tt nat},\neg\heap{}}\}]$ (}\\
    \indentone{\tt proj} $r_1, {\tt env}, 1;$& $R\leftarrow \hup{R}{r_1}{\ell_1}$\\
    \indentone{\tt proj} $r_2, {\tt env}, 2;$& $R\leftarrow \hup{R}{r_2}{\ell_2}$\\
    \indentone{\tt proj} $r_3, {\tt env}, 3;$& $R\leftarrow \hup{R}{r_3}{\ell_3}$\\
    \indentone{\tt proj\ ret}$, {\tt env}, 4;$&$R\leftarrow \hup{R}{{\tt ret}}{\ell_5}$\\
  \multicolumn{2}{l}{\# Here we are replacing a free variable with a concrete term}\\
    \indentone{\tt get\_str} ${\tt succ/1}, r_2;$      &$\ell=\ell_1, \vec\ell=\langle\rangle$\\
    \indenttwo{\tt unify\_var} $r_2;$            &$H \leftarrow \hext{H}{\ell_4}{\free{}[{\tt nat}]},R \leftarrow \hup{R}{r_2}{\ell_4},\vec\ell= \langle\ell_4\rangle$\\
                             &$H \leftarrow \hup{H}{\ell_1}{{\tt succ}\ \vec\ell},$\\
    \indentone{\tt get\_str} ${\tt succ/1}, r_3;$      &$\vec\ell = \langle\ell_2\rangle$\\
    \indenttwo{\tt unify\_var} $r_3;$            &$R \leftarrow \hup{R}{r_3}{\ell_2},$\\
                             &$H \leftarrow \hup{H}{\ell_1}{{\tt succ}\ \vec\ell},$\\
    \indentone{\tt jmp }\verb|plus-zero/3|$;$         &$I \leftarrow C(\verb|plus-zero|)$\\
  )&\\
~\\  
\multicolumn{2}{l}{{\tt plus-zero/3}$~\mapsto \code[\{r_2:{\tt nat}, r_2:{\tt nat},r_3:{\tt nat},r_4:\neg\heap{}\}]($}\\
    \indentone{\tt put\_tuple} $r_4,4;$          &$\vec\ell=\langle\rangle$\\
    \indenttwo   {\tt set\_val} $r_2;$           &$\vec\ell=\langle\ell_1\rangle$\\
    \indenttwo   {\tt set\_val} $r_2;$           &$\vec\ell=\langle\ell_1,\ell_2\rangle$\\
    \indenttwo   {\tt set\_val} $r_3;$           &$\vec\ell=\langle\ell_1,\ell_2,\ell_3\rangle$\\
    \indenttwo   {\tt set\_val} {\tt ret};       &$\vec\ell=\langle\ell_1,\ell_2,\ell_3,\ell_5\rangle$\\
                                       &$H \leftarrow \hext{H}{\ell_7}{\langle\vec\ell\rangle},R \leftarrow \hup{R}{r_4}{\ell_6}$\\
    \indentone{\tt push\_bt} $r_4, {\verb|plus-succ/3|};$ &$T \leftarrow(\ell_7,{\verb|plus-succ|},\langle\rangle) :: {\tt nil}$\\
    \indentone{\tt get\_str} $r_2, {\tt zero/0};$      &$\ell=\ell_4, \vec\ell=\langle\rangle, c = {\tt zero}$\\
                             &$H \leftarrow \hup{H}{\ell_4}{c\langle\vec\ell\rangle}$\\
    \indentone{\tt get\_val} $r_2, r_3;$          &\\
    \indentone{\tt jmp} $r_4$;                  &$I \leftarrow C(R(r_4)) = C({\tt success})$\\
)&\\
~&\\
\multicolumn{2}{l}{\tt {success/0} $\mapsto\code[\heap{}]$ (}\\
    \indentone{\tt succeed;}&\\
)&\\
\end{longtable}
\normalsize
\subsection{Formal Operational Semantics}
\label{sec:simp-op}
The small-step operational semantics consists of three main judgements: 
$m\step m', m\done,$ and $m\fails,$ where $m\fails$ indicates a negative result to a Prolog query, not a stuck state.
There are also numerous auxilliary judgements for unification, backtracking, trail management, etc.
We begin with conceptually simple cases and proceed to conceptually complex ones.
The simplest instructions are {\tt mov} and {\succeed}, requiring no auxilliary judgements:
\begin{center}
\begin{tabular}{cc}
\infer[\textsc{Mov}\step]{(T,S,R,\mov{r_d}{r_s}; I) \step (T,S,\hup{R}{r_d}{w},I)}{R(r_s) = w}&
\infer[\done]{(T,S,R,\succeed) \done}{}
\end{tabular}
\end{center}
\subsubsection{Operands}
The $\jmp{op}$ instruction takes at \emph{operand} which allows us to jump either to a literal location or a success continuation stored in a register.
The \emph{operand evaluation} judgement $R \ent op\eval w$ resolves an operand $op$ into a word $w$ by consulting the registers $R$ if necessary.
If the operand is a code location, {\tt jmp} simply transfers control, else if the operand is a closure, {\tt jmp} also loads the stored environment.
\begin{center}{\footnotesize
\begin{tabular}{cc}
\infer[\textsc{Jmp-}\ell^C]{(T,S,R,\jmp op; I) \step (T,S,R,I')}
{R \ent op \eval \ell^C & S(\ell^C) = \code[\Gamma]I'}&
\infer[\textsc{Jmp-}\ell^H]{(T,S,R,\jmp op; I) \step (T,S,\hup{R}{\texttt{env}}{w_{env}},I')}
{R \ent op \eval \ell^H & S(\ell^H) = \close(w_{env}, \ell^C)
&S(\ell^C) = \code[\Gamma](I')}
\end{tabular}
\begin{tabular}{cc}
\infer[\ell^C\eval]{R\ent \ell^C\eval\ell^C}{}&
\infer[r\eval]{R\ent r\eval w}{R(r) = w}
\end{tabular}}
\end{center}
\subsubsection{Environments}
Environment tuples are constructed with the {\tt twrite} spinal mode.
This mode begins after {\tt put\_tuple} and ends when the count of remaining tuple elements reaches 0.
When the spine completes, the resulting tuple is stored in the destination register specified by the initial {\tt put\_tuple}.
As before, $\epsilon$ denotes an empty sequence. We also use the notation $\vec{w}::w$ even when adding an element $w$  to the end of a sequence $\vec{w}$.
%TODO: Say C is fixed somewhere
Reading tuple elements with {\tt proj} does not require entering a spine.
\[\infer[\textsc{PutTuple}\step]{(T,S,R,\puttuple{r, n}; I) \step
\twrite(T,S,R,I,r,n,\epsilon)}{}\]
\[\infer[\textsc{SetVal}\step]{\twrite(T, S,R, \setval r_s; I,r_d,n,\vec w) \step
\twrite(T,S,R,I,r_d,n-1,(\vec w::w))}{R(r_s) = w & n > 0}\]
\vspace{0in}
\[\infer[\textsc{TWrite}\step]{\twrite(T,S,R,I,r,0,\vec w) \step (T,\hext{S}{\ell^H}{\args{\vec w}},\hup{R}{r}{\ell^H},I)}{}\]
\[\infer[\textsc{Proj}\step]{(T,S,R,\proj r_d, r_s, i; I) \step (T,S,\hup{R}{r_d}{w_i},I)}
{R(r_s)=\ell^H & S(\ell^H) = \args{w_1,\ldots,w_i,\ldots,w_n}}\]
\subsubsection{Continuations and Backtracking}
\label{sec:cont-bt}
The instructions {\tt close} and {\tt push\_bt} allocate new success and failure continuations, respectively.
The {\tt close} instruction puts the continuation in a register $r_d$ for use by a future {\tt jmp}, 
whereas {\tt push\_bt} puts the failure continuation in the trail.
As shown in Section \ref{sec:trailing}, the trail maintains an invariant that the location of every free variable
bound since the last {\tt push\_bt} is stored in the current trail frame, which is necessary when backtracking.
Because we have just created a new failure continuation, after a {\tt push\_bt} our new trail frame contains the empty list $\epsilon$.
Backtracking is handled automatically when unification fails, thus there is no need to make the failure continuation accessible via a register.
\[\infer[\textsc{Close}\step]{(T,S,R,\close r_d, r_e, \ell^C; I) \step (T,\hext{S}{\ell^H}{\close(w_{env},\ell^C)},\hup{R}{r}{\ell^H},I)}
{R(r_e) = w_{env}}\]
\[\infer[\textsc{PushBT}\step]{(T,S,R,\branch r_e, \ell^C; I) \step ((\epsilon,w_{env},\ell^C)::T,S,R,I)}
{R(r_e) = w_{env}}\]
The trail invariant is essential to the correctness of the $\bt$ operation, which succeeds when there is failure continuation on the trail, or signals query failure when the trail is empty.
Trail unwinding has its inverse in trail updating, which adds a recently-bound variable to the trail (or safely skips it if there are no failure continuations left).
\begin{center}{\small
  \begin{tabular}{cc}
\infer[\textsc{BT-Cons}]{\bt(S,(t,w,\ell^C)::T) = (T,S',\heap{\texttt{env} \hook w},I)}
        {\unwind(S,t) = S' & C(\ell^C) = \code[\heap{\texttt{env}:\tau}](I)}&
\infer[\textsc{BT-Nil}]{\bt(S,\epsilon) = \bot}{}\\[0.08in]
 $\unwind(S,(\ell^H:a)::t) = \unwind( \hup{S}{\ell^H}{\free{}[a]},t)$ & $\unwind(S, \epsilon) = S$\\[0.08in]
$\uptrail(\ell^H:a,(t,w_{env},\ell^C)::T) = ((\ell^H:a)::t,w_{env},\ell^C)::T$&$\uptrail((\ell^H:a),\epsilon) = \epsilon$
  \end{tabular}}
\end{center}

\subsubsection{Unification, Occurs Checks, and Trailing}
\label{sec:trailing}
The {\tt get\_val} instruction unifies two arbitrary Prolog terms stored at $r_1$ and $r_2$.
It does so using the auxilliary judgement $\unify(S,T,\ell^H_1,\ell^H_2) = (S',T')$, which computes the resulting store where $\ell^H_1$ and $\ell^H_2$ are unified, or $\bot$ if unification fails.
It also must compute an updated trail, because unification binds free variables, and backtracking must be able to undo those changes.
\[\infer[\textsc{GetVal}\step]{(T, S,R, \getval{r_1}{r_2}; I) \step (T', S',R,I)}
{R(r_1) = \ell^H_1 & R(r_2) = \ell^H_2 & \unify(S,T,\ell^H_1,\ell^H_2) = (S',T')}\]
The judgement $\unify(S,T,\ell^H_1,\ell^H_2) = (S',T')$ is defined by mutual recursion with the judgement $\unifyargs(S,T,\vec \ell^H, \vec \ell'^H)=(S',T')$ which simply unifies every $\ell_i^H$ with the corresponding $\ell_i'^H$. 
These lists $\vec\ell^H$ and $\vec\ell'^H$ correspond to the push-down list (PDL) in other presentations of the WAM.
An additional judgement $\edn(S,\ell^H)$ follows chains of $\bound{\ell^H}$ pointers to their ends.
Because the typing invariant for heaps ensures absence of cycles, this is guaranteed to terminate.
The basic unification algorithm says to recurse if both unificands are structures, or if either is a free variable, then bind it to the other unificand.
However, unification must also maintain the invariant that the heap is free of cycles, thus we employ an occurs check in our algorithm, 
writing $\ell_1^H \in_S \ell_2^H$ when $\ell_1^H$ occurs in $\ell_2^H$ (occurs check failure)  and $\ell_1^H \notin_S \ell_2^H$ otherwise (occurs check success).\footnote{Typing constraints ensure that at runtime, the occurs check is only ever invoked on Prolog terms. However, the occurs check is also of broader use in the metatheory proofs, and there it is convenient to define the occurs check on closures and tuples as well, such as in Lemma~\ref{lem:heap-str}.}
Additionally, we employ the $\uptrail$ function to maintain the trail invariants when binding free variables.
{\footnotesize\begin{center}
\begin{tabular}{c}
\infer[\in{}c\langle\rangle]{\ell^H_1 \in_S \ell^H_2}{S(\ell^H_2) = c\args{\ell'^H_1,\ldots,\ell'^H_n} & \ell^H_1 \in_S \ell'^H_i\ (\exists i \in [n])}
\end{tabular}
\end{center}
\begin{center}
\begin{tabular}{cc}
\infer[\in=]{\ell^H_1 \in_S \ell^H_2}{\ell^H_1 = \ell^H_2}&
\infer[\in\bound{}]{\ell^H_1 \in_S \ell^H_2}{S(\ell^H_2) = \bound{\ell'^H_2} & \ell^H_1 \in_S \ell'^H_2}
\end{tabular}
\end{center}

\begin{center}
\begin{tabular}{cc}
\infer[\notin\free{}]{\ell^H_1 \notin_S \ell^H_2}{\ell^H_1 \neq \ell^H_2 & S(\ell^H_2)=\free{}[a]}&
\infer[\notin\bound{}]{\ell^H_1 \notin_S \ell^H_2}{S(\ell^H_2) = \bound{\ell'^H_2} & \ell^H_1 \notin_S \ell'^H_2}
\end{tabular}
\end{center}
\begin{center}
\begin{tabular}{c}
\infer[\notin{}c\langle\rangle]{\ell^H_1 \notin_S \ell^H_2}{S(\ell^H_2) = c\args{\ell'^H_1,\ldots,\ell'^H_n} & \ell^H_1 \notin_S \ell'^H_i\ (\forall i \in [n])}
\end{tabular}
\end{center}
\begin{center}
\begin{tabular}{cc}
\infer[\in\langle\rangle]{\ell^H_1 \in_S \ell^H_2}{S(\ell^H_2) = \args{\ell'^H_1,\ldots,\ell'^H_n} & \ell^H_1 \in_S \ell'^H_i\ (\exists i \in [n])}&
\infer[\notin\langle\rangle]{\ell^H_1 \notin_S \ell^H_2}{S(\ell^H_2) = \args{\ell'^H_1,\ldots,\ell'^H_n} & \ell^H_1 \notin_S \ell'^H_i\ (\forall i \in [n])}
\end{tabular}
\end{center}
\begin{center}
\begin{tabular}{cc}
\infer[\in\close]{\ell^H_1 \in_S \ell^H_2}{S(\ell^H_2) = \close(\ell'^H_2,\ell^C) & \ell^H_1 \in_S \ell'^H_2}&
\infer[\notin\close]{\ell^H_1 \notin_S \ell^H_2}{S(\ell^H_2) = \close(\ell'^H_2,\ell^C) & \ell^H_1 \notin_S \ell'^H_2}
\end{tabular}
\end{center}
\begin{center}
\begin{tabular}{cc}
  \infer[\edn~\free{}]{\edn(S,\ell^H) = \ell^H}{S(\ell^H) = \free{}[a]}&
\infer[\edn~c\langle\rangle]{\edn(S,\ell^H) = \ell^H}{S(\ell^H) = \cwsh}
\end{tabular}
\end{center}
\begin{center}

\begin{tabular}{c}
\infer[\edn~\bound{}]{\edn(S,\ell^H) = \ell''^H}{S(\ell^H) = \bound{\ell'^H} & \edn(S,\ell'^H) = \ell''^H}\\
\infer[{\tt unify}=]{\unify(S,T,\ell_1^H,\ell_2^H) = (S,T)}
        {\edn(\ell_1^H) = \ell^H & \edn(\ell_2^H) = \ell^H}
\infer[{\tt unify}~\free{}]{\unify(S,T,\ell_1^H,\ell_2^H) = (\hup{S}{\ell'^H_2}{\bound{\ell_1^H}},T')}
        {\deduce{{\ell'}_2^H \notin_S \ell_1\hskip 0.1in \uptrail(T,({\ell'}^H_2:a))=T'}
          {\edn(S,\ell_1^H) = \ell'^H_1 & \edn(S,\ell^H_2) = {\ell'}^H_2 & S(\ell'^H_2) = \free{}[a]}}\\
\infer[{\tt unify}~c\langle\rangle]{\unify(S,T,\ell_1^H,\ell_2^H) = (S',T')}
        {\deduce{\edn(S,\ell_2^H) = \ell'^H_2\hskip 0.1in\unifyargs(S,T,\vec \ell^H, \vec \ell'^H) =(S',T')}
          {S(\ell'^H_1) = \func{c}{\vec \ell^H} & S(\ell'^H_2) = \func{c}{\vec \ell'^H} & \edn(S,\ell_1^H) = \ell'^H_1}}\\
\end{tabular}
\end{center}
\begin{center}
\begin{tabular}{cc}
\infer[\textsc{UA-Cons}]{\unifyargs(S,T,(\ell^H_1::\cdots::\ell_n^H),(\ell'^H_1::\cdots::\ell'^H_n)) = (S'',T'')}
        {\deduce{\unifyargs(S',T',(\ell^H_2::\cdots::\ell^H_n),(\ell'^H_2::\dots::\ell'^H_n)) = (S'',T'')}{\unify(S,T,\ell^H_1,\ell^H_2) = (S',T')}}
&\infer[\textsc{UA-Nil}]{\unifyargs(S,T,\epsilon,\epsilon) = (S,T)}{}
\end{tabular}
\end{center}}

The failure cases for unification are straightforward, but they are given here for completeness:
\begin{center}\footnotesize\begin{tabular}{cc}
\infer[\textsc{U}\bot1]{\unify(S,T,\ell^H_1,\ell^H_2) = \bot}
        {\deduce{\edn(S,\ell^H_1) = \ell'^H_1 \hskip 0.1in\edn(S,\ell^H_2) = \ell'^H_2}
        {S(\ell'^H_1) = \free{}[a] & \ell'^H_1 \in_S \ell^H_2}}&
\infer[\textsc{U}\bot2]{\unify(S,T,\ell^H_1,\ell^H_2) = \bot}
        {\deduce{\edn(S,\ell^H_1) = \ell'^H_1 \hskip 0.1in\edn(S,\ell^H_2) = \ell'^H_2}
          {S(\ell'^H_2) = \free{}[a] & \ell'^H_2 \in_S \ell^H_1}}\\
\infer[\textsc{U}\bot3]{\unify(S,T,\ell^H_1,\ell^H_2) = \bot}
        {\deduce{\edn(S,\ell^H_1) = \ell'^H_1 \hskip 0.1in\edn(S,\ell^H_2) = \ell'^H_2}
           {S(\ell'^H_1) = \func{c}{\vec w} & S(\ell'^H_2) = \func{c'}{\vec w'} & c \neq c'}}
&
\infer[\textsc{U}\bot4]{\unify(S,T,\ell^H_1,\ell^H_2) = \bot}
        {\deduce{\edn(S,\ell^H_1) = \ell'^H_1 \hskip 0.1in \edn(S,\ell^H_2) = \ell'^H_2 \hskip 0.1in S(\ell'^H_1) = \func{c}{\vec \ell^H} }
           {\unifyargs(S,T,\vec \ell^H, \vec \ell'^H) = \bot &  S(\ell'^H_2) = \func{c}{\vec \ell'^H}}}\\
\infer[\textsc{UA}\bot1]{\unifyargs(S,T,(\ell^H::\vec \ell^H),(\ell'^H::\vec \ell'^H)) = \bot}
        {\unify(S,T,\ell^H,\ell'^H) = \bot}&
\infer[\textsc{UA}\bot2]{\unifyargs(S,T,(\ell^H::\vec \ell^H),(\ell'^H::\vec \ell'^H)) = \bot}
        {\deduce{\unify(S,T,\ell^H,\ell'^H) = (S',T')}{\unifyargs(S',T',\vec\ell^H,\vec\ell'^H) = \bot}}
\end{tabular}
\end{center}
Lastly, if unification fails, due to the above failure rules, {\tt get\_val} tries backtracking.
If the trail is non-empty, it backtracks successfully, else execution stops and the query has failed.
\[\infer[\textsc{GetVal-BT}]{(T, S,R, \getval{r_1}{r_2}; I) \step m'}
{R(r_1) = \ell^H_1 & R(r_2) = \ell^H_2 & \unify(S,T,\ell^H_1,\ell^H_2) = \bot & \bt(S,T) = m'}\]
\[\infer[\textsc{GetVal-}\bot]{(T, S,R, \getval{r_1}{r_2}; I) \fails}
{R(r_1) = \ell^H_1 & R(r_2) = \ell^H_2 & \unify(S,T,\ell^H_1,\ell^H_2) = \bot &
\bt(S,T) =\bot}\]
\subsubsection{Term Constructors and Occurs Checks}
The {\tt put\_var} instruction immediately allocates a free variable.
Structures are constructed with a write spine, initiated by {\tt put\_str}.
For symmetry with {\tt get\_str}, we first allocate a free variable and replace it with a structure when the spine completes.
Within a write spine, {\tt unify\_var} allocates free variables while {\tt unify\_val} copies a value into the structure.
In the typical case, a write spine finishes when enough arguments have been computed (one for each constructor argument, i.e. $\arity(c)$) by replacing the free variable with a complete structure.
However, the formalism technically allows us to refer to the destination $r$ within its own write spine.
To prevent a cycle, we perform an occurs check ($\ell^H \notin_S \ell_i^H$) and backtrack on failure.
The choice to use an occurs check here was made for its resulting proof simplicity.
For implementation purposes, an equally correct and more efficient choice is to enforce a syntactic restriction that prohibits references to $r$ within its own write spine.
\vspace{0in}
\[\infer[\textsc{PutVar}\step]{(T,S,R,\putvar{a}{r};I) \step (T,\hext{S}{\ell^H}{\free{}[a]},\hups{R}{r \hook \ell^H},I)}{}\]
\[\infer[\textsc{PutStr}\step]{(T,S,R,\putstr{c}{r}; I) \step
\mwrite(T,
  \hext{S}{\ell^H}{\free{}[a]},\hup{R}{r}{\ell^H},I,c,\ell^H,\epsilon)}{\Sigma(c) = \vec a \to a}\]
\vspace{0.0in}
\[\infer[\textsc{UnifyVar}\step\textsc{W}]{\deduce{\mwrite(T, \hext{S}{\ell^H}{\free{}[a]},\hup{R}{r_s}{\ell^H},I,c,\ell^H_d,(\vec \ell^H::\ell^H))}{\mwrite(T,S,R,\unifyvar[a] r_s; I,c,\ell^H_d,\vec \ell^H) \step}}{}\]
\vspace{0.0in}
\[\infer[\textsc{UnifyVal}\step\textsc{W}]{\mwrite(T, S,R,\unifyval r_s; I,c,\ell^H_d,\vec \ell^H) \step
\mwrite(T,S,R,I,c,\ell^H_d,(\vec \ell^H::\ell^H))}{R(r_s) = \ell^H}\]
\[\infer[\textsc{Write}\step]{\mwrite(T,S,R,I,\ell^H,c,\vec \ell^H) \step
(T,\hup{S}{\ell^H}{c\args{\vec \ell^H}},R,I)}
{ S(\ell^H) = \free{}[a]  & \ell^H \notin_S  \ell^H_i & |{\vec \ell^H}| = \arity(c) }\]
\[\infer[\textsc{Write-BT}]{\mwrite(T,S,R,I,\ell^H,c,\vec \ell^H) \step m'}
{\ell^H \in_S \ell^H_i & |{\vec \ell^H}| = \arity(c) & \bt(S,T) = m'}\]
\[\infer[\textsc{Write-}\bot]{\mwrite(T,S,R,I,\ell^H,c,\vec \ell^H) \fails}
{\ell^H \in_S \ell^H_i & |{\vec \ell^H}| = \arity(c) & \bt(S,T) = \bot}\]

\subsubsection{Term Destructor {\tt get\_str}}
The instruction ${\tt get\_str}\ c,r$ starts a spine which unifies the content of register $r$ with a structure $\cwsh$ where each $w_i$ is provided by the $i$'th spinal instruction.
In the case where $r$ contains a free variable, this amounts to building a new structure, and thus the {\tt write} case of {\tt get\_str} simply reuses the {\tt write} spines of {\tt put\_str}.
\[\infer[\textsc{GetStr}\step\textsc{W}]{(T,S,R,\getstr {c}{r}; I) \step
\mwrite(T, S,R,I,c,\ell'^H,\epsilon)}
{R(r) = \ell^H & \edn(S,\ell^H) = \ell'^H & S(\ell'^H) = \free{}[a]}\]
When $r$ contains a structure, we perform structure-to-structure unification.
This can fail in two ways: either the head constructor $c$ does not match or one of the argument positions does not unify.
The first conditon is checked during {\tt get\_str} itself, the other during the ensuing spinal instructions.
In both cases, errors are handled by backtracking if possible:
\[\infer[\textsc{GetStr}\step\textsc{R}]{(T, S,R,\getstr{c}{r}; I) \step
\mread(T, S,R,I,\vec \ell^H)}
{R(r) = \ell^H & \edn(S,\ell^H) = \ell'_H & S(\ell'^H) = c\args{\ell^H_1,\ldots,\ell^H_n}}\]
\[\infer[\textsc{GetStr-BT}]{(T,S,R,\getstr {c}{r}; I) \step m'}
        {\deduce{S(\ell'^H)=c'\ells \hskip 0.1in c \neq c' \hskip 0.1in \bt(S,T) = m'}{R(r)= \ell^H &\edn(S,\ell^H) = \ell'^H}}\]
\[\infer[\textsc{GetStr-}\bot]{(T,S,R,\getstr {c}{r}; I) \fails}
 {\deduce{S(\ell'^H)=c'\ells \hskip 0.1in c \neq c' \hskip 0.1in \bt(S,T) = \bot}{R(r)= \ell^H & \edn(S,\ell^H) = \ell'^H}}
\]
\vspace{-0.05in}
\[\infer[\textsc{UnifyVar}\step\textsc{R}]{\mread(T,S,R,\unifyvar [a]r;I, (\ell^H::\vec \ell^H)) \step
\mread(T,S,\hup{R}{r}{\ell^H},I,\vec \ell^H)}{~\hskip 0.1in}\]
\[\infer[\textsc{UnifyVal}\step\textsc{R}]{\mread(T, S,R,\unifyval r; I, (\ell^H::\vec \ell^H)) \step
\mread(T',S',R,I, \vec \ell^H)}
{R(r) = \ell'^H & \unify(S,T,\ell^H,\ell'^H) = (S',T')}\]
\[\infer[\textsc{UnifyVal-BT}]{\mread(T,S,R,\unifyval r; I, (\ell^H::\vec \ell^H)) \step m'}
{R(r) = \ell'^H & \unify(S,T,\ell^H,\ell'^H) = \bot & \bt(T) = m'}\]
\[\infer[\textsc{UnifyVal-}\bot]{\mread(T,S,R,\unifyval r; I, (\ell^H::\vec \ell^H)) \fails}
{R(r) = \ell'^H & \unify(S,T,\ell^H,\ell'^H) = \bot & \bt(T) = \bot}\]
\subsection{Metatheory}
\label{sec:simple-metatheory}
%\label{sec:dep-met}
In both the SWAM and dependently-typed TWAM, the main metatheorems are progress and preservation.
\begin{thm*}[Progress] If $\SE m \ok$ then either $m \done$ or $m \fails$ or $m \step m'$.\end{thm*}
\begin{thm*}[Preservation] If $\SE m \ok$ and $m \step m'$ then $\SE m' \ok$.\end{thm*}
%\begin{thm*}[Safety] If $\SE m \ok$ and $m \step^* m'$ then $m' \done$ or $m' \fails$ or $m' \step m''$.\end{thm*}
Where $m \fails$ mean that a Prolog query terminated normally, but the query had no solution.

In Section \ref{sec:dep-met}, the progress and preservation results for the TWAM will be strong enough to enable certifying compilation.
In the SWAM, progress and preservation amount to type and memory-safety.
Because the theorem of Section \ref{sec:dep-met} subsumes progress and preservation for SWAM,
we restrict ourselves here to the commonalities and present the differences in Section \ref{sec:dep-met}.
For the sake of readability, both this section and Section \ref{sec:dep-met} give proof sketches where the reader might find a detailed proof tedious.
For the sake of exhaustiveness, an extended proof for the dependent system is given in the electronic appendix, however.

The metatheory for SWAM begins with standard preliminary lemmas such as canonical forms and weakening.
This is followed with the heart of the metatheory: our treatment of the occurs check and unification.

The key lemma Heap Update (Lemma~\ref{lem:heap-update}) shows that binding free variables preserves the acyclic heap invariant when the occurs check passes, 
which gives us preservation for unification and thus every instruction that depends on unification.
\subsubsection{Preliminaries}
\begin{lem}[Canonical Forms] 
\label{lem:scf}
Canonical forms consists of a subclaim for each relevant class of values.
\begin{itemize}
  \item \emph{Code Values:} If $\PE v^C : \tau$ then $\tau = \neg \Gamma$ and $v^C = \code[\Gamma](I)$
  \item \emph{Word Values:} If $\PE_{\Sigma;\Xi} w : \tau$ and $(C,H) : (\Xi;\Psi)$ then $w$ has form $\ell^H$ or $\ell^C$ where $\ell^H \in\dom{H}$ or $\ell^C\in\dom{C}$.
  \item \emph{Heap Values:} If $\PE v^H : \tau$ then
  \begin{itemize}
    \item Either $\tau = \cross{\vec \tau}$ and $v = \args{\vec w}$ and $\PE \vec w : \vec \tau$
    \item Or $\tau = a$ and either $v = \bound{w}$ or $v = \free{}[a]$ or $v = \func{c}{\vec w}$ where $\PE \vec w : \vec{a}$ and $\Sigma(c)=\vec{a}\to a$.
    \item Or $\tau = \neg \Gamma$ and  $v = \close(w,\ell^C)$ where $\PE w :\neg\myG({\tt env})$.
  \end{itemize}
\end{itemize}
\end{lem}
\begin{proof}
  Each claim is by inversion on the typing rules.
\end{proof}
\begin{lem}[Weak Unicity of Heap Value Typing] For any value $v$, at most one of the following holds:
\begin{enumerate}
\item $\PE v : a$
\item $\PE v : \cross{\vec \tau}$
\item $\PE v : \neg \myG$
\end{enumerate}
\end{lem}
\begin{proof}
By cases on $\PE v : \tau$. 
Each rule requires $v$ to have a specific form.
If $\PE v : a$ then $v = \free{}[a],\bound{\ell^H},$ or $\cwsh$. 
In each case the only rules that apply produce type $a$ (the type annotation enforces unicity for $\free{}[a])$.
If $\PE v : \cross{\vec \tau}$ then $v = \args{\vec w}$ and the only rule that applies produces type $\cross{\vec\tau}$. 
Otherwise $\PE v : \neg\myG$ and the rules force $v = \code[\myG]$ or $v = \closure(w_{env}, \ell^C)$. 
In either case the only rule that applies produces $\PE v : \neg\myG$, because the type is restricted by either the annotation $\myG$ or by $\Xi(\ell^C)$.
\end{proof}
\begin{lem}[Weakening]
\label{lem:sweak}
In SWAM we need weakening for word and heap value typing and occurs check:
\begin{itemize}
\item \emph{Word Values:} If $\PE w : \tau$ then $\thext{\Psi}{\ell_H}{\tau'} \vdash w : \tau.$
\item \emph{Occurs Check:} For fresh $\ell^H,$
 (a) If $\ell^H_1 \in_H \ell^H_2,$ then $\ell^H_1 \in_{\hext{H}{\ell^H}{v}} \ell^H_2$ and
 (b) If $\ell^H_1 \notin_H \ell^H_2,$ then $\ell^H_1 \notin_{\hext{H}{\ell^H}{v}} \ell^H_2$.
\item\emph{Heap Typing:} For all fresh $\ell^H$ and even ill-typed $v^H$, if $\SE H:\Psi$ then $\SE \hext{H}{\ell^H}{v^H}:\Psi$.
\end{itemize}
\begin{proof} By induction on the derivation $\PE w : \tau$, $\ell^H_1 \in_H \ell^H_2$, $\ell^H_1 \notin_H \ell^H_2$, or $\SE H:\Psi$ respectively, using the fact that for fresh $\ell^H,$ $\hext{H}{\ell^H}{v}(\ell'^H) = H(\ell'^H)$ for all $\ell'^H \in \dom{H}$. Heap Typing weakening uses the fact that $\SE H:\{\}$ for all $H$, i.e. heaps may contain inaccessible values not assigned types by $\Psi$.\end{proof}
\end{lem}
\begin{lem}[Register File Subtyping]\label{lem:reg-sub} If $\GE I \ok$ and $\SE \Gamma \leq \Gamma'$ then $\Gamma' \ent I \ok$. \end{lem}
\begin{proof}
By induction on the derivation $\GE I \ok$.
\end{proof}
\subsubsection{Occurs Check}
The following theory of occurs checks is used in the theory of unification.
\begin{lem}[Occurs is a total function]
\label{lem:occ-tot}
 If $\SE H : \Psi$ and $\SPE \ell^H_2:a$ then:
  \begin{itemize}
  \item\emph{Total:} For all $\ell^H_1,$ either $\ell^H_1 \in_H \ell^H_2$ or $\ell^H_1 \notin_H \ell^H_2$.
  \item\emph{Function:} If $\ell^H_1 \notin_H \ell^H_2$ is derivable, then $\ell^H_1 \in_H \ell^H_2$ is not derivable.
  \end{itemize}
% for all $w_1$, exactly one of  $w_1 \in_H w_2$ or $w_1 \notin_H w_2$ is derivable.
\end{lem}
\begin{proof}
Totality is by induction on typing derivation of $\SE H : \Psi$, appealing to Lemmas~\ref{lem:sweak} and \ref{lem:scf}. Functionhood is by induction on the derivation $\ell^H_1 \notin_H  \ell^H_2$.
In the cases $H(\ell^H_2) = \free{}[a]$, $H(\ell^H_2) = \bound{w},$ and $\close(w_{env},\ell^C),$ clearly no rules apply for $\ell^H_1 \in_H \ell^H_2$. Consider the case $\cwsh$ (the tuple case is symmetric):

\case{\infer[\notin{}c\langle\rangle]{\ell^H_1 \notin_H \ell^H_2}{H(\ell^H_2) = c\args{\ell'^H_1,\ldots,\ell'^H_n} & \ell^H_1 \notin_H \ell'^H_i (\forall 1 \leq i \leq n)}}
By the IH, $\forall i.$($\ell^H_1 \in_H \ell'^H_i$ is not derivable), so $\neg \exists i.( \ell^H_1 \in_H \ell'^H_i$ is derivable). But because
$H(\ell_2) = \cwsh$, the only rule that might apply requires $\exists i. (\ell_1 \in_H \ell'_i$ is derivable). 
\end{proof}
\begin{lem}[Transitivity of Occurs]\label{lem:occ-trans} If $\SE H : \Psi, \ell^H_1 \in_H \ell^H_2, $ and $\ell^H_2 \in_H \ell^H_3,$  then $\ell^H_1 \in_H \ell^H_3$.\end{lem}
\begin{proof}
By induction on the derivation $\ell^H_2 \in_H \ell^H_3$.
\end{proof}
\begin{lem}[Occurs Strengthening] If $\SE H : \Psi, \ell^H_2 \in \dom{H}, \ell'^H \notin \dom{H},$ and
$\ell^H_1 \notin_{\hext{H}{\ell'^H}{v'}} \ell^H_2$ then $\ell^H_1 \notin_H \ell^H_2$. \end{lem}
\begin{proof}
By induction on the derivation $\ell^H_1 \notin_{\hext{H}{\ell'^H}{v'}} \ell^H_2$, appealing to Lemma~\ref{lem:scf}.
\end{proof}

% \begin{lem}[Occurs Weakening]
% \begin{align*}
% (a) &\text{ If $\ell_1 \in_H \ell_2,$ then $\ell_1 \in_{\hext{H}{\ell}{v}} \ell_2$}.\\
% (b) &\text{ If $\ell_1 \in_H \ell_2,$ then $\ell_1 \in_{\hext{H}{\ell}{v}} \ell_2$}.
% \end{align*}
% \end{lem}
% \begin{proof}
% By induction on the derivations $\ell_1 \in_H \ell_2$ and $\ell_1 \not \in_H \ell_2$.

% \begin{lem}[Occurs Progress] If $\DE H : \Psi$ and $\SE (\Delta,\mu) : H$ and $\DE \ell_1 : \sing(x : A)$ and
% $\DE \ell_2 : \sing(M : A)$ then let $\edn(H,\ell_1) \squig \ell_1'$.
% Now if  $x \in M$ then $\ell_1' \in_H \ell_2$ and if $x \notin M$ then $\ell_1' \notin_H \ell_2$.  \end{lem}
% \begin{proof}
% By lexicographic induction on the derivation $x \in M$ or $x \notin M$ and
% on $|\Psi|$. In the free variable case, use the assumption that variables are mapped uniquely to heap addresses.
% \end{proof}

% As mentioned previously, it is key that we never introduce cycles in the heap. We
% prove that if the occurs check passes, then binding a free variable
% results in a heap which is well-typed, which by the typing 
% invariants is guaranteed to be acyclic. Not only is this property essential
% for logical soundness, but it also provides an inductive principle for heaps.

\subsubsection{Heap Modification}
The simply-typed metatheory culminates in the treatment of heap modification.
We begin with a strengthening lemma:
\begin{lem}[Heap Value Strengthening]
\label{lem:heap-str}
If $\thext{\Psi}{\ell^H_1}{\tau_1} \vdash v_2:\tau_2$ and $\ell^H_1 \notin_H \ell^H_2$, then $\PE v_2:\tau_2 $.\end{lem}
\begin{proof}
By cases on $\thext{\Psi}{\ell^H_1}{\tau_1} \vdash v_2:\tau_2$.
The case $\cwsh$ is representative:

\case{\infer[\textsc{HV-Str}]{\thext{\Psi}{\ell^H_1}{\tau_1} \vdash \cwsh : \tau_2 }{\Sigma(c) = \vec a \to \tau_2 & \Sigma;\thext{\Psi}{\ell^H_1}{\tau_1} \vdash \ell'^H_i : a_i }}
Then for each $\ell'^H_i$ the typing derivation has form \infer[\textsc{WV-}\ell^H]{\thext{\Psi}{\ell^H_1}{\tau_1} \vdash \ell'^H_i : \sing(M_i : a_i) }{\thext{\Psi}{\ell^H_1}{\tau_1}(\ell'^H_i) = \sing(M_i: a_i)}.
Note $\ell'^H_i \neq \ell^H_1$. Otherwise we would have
\[\infer[\in{}c\langle\rangle]{\ell^H_1 \in_H \ell^H_2}{\infer[\in=]{\ell^H_1 \in_H \ell'^H_i (\exists i)}{\ell'^H_i = \ell^H_1} & H(\ell^H_2)=c\langle \ell'^H_1,\ldots\ell'^H_n\rangle}\]
Which would contradict $\ell^H_1 \notin_H \ell^H_2$ because occurs is a function (Lemma~\ref{lem:occ-tot}). Because $\ell^H_1 \neq \ell'^H_i$ and $\ell'^H_i \in \dom{\thext{\Psi}{\ell^H_1}{\tau_1}}$ we have
$\ell'^H_i \in \dom{\Psi}$ and
\[\infer[\textsc{HV-Str}]{\SPE \func{c}{\ell^H_1, \ldots, \ell^H_n} : \tau_2 }{\Sigma(c)= a_1 \cdots \to a_n \to \tau_2 & \infer[\textsc{WV-}\ell^H]{\SPE \ell'^H_i : a_i }{\Psi(\ell'^H_i) = \sing(M_i : a_i) &(\forall i)}}\]
\end{proof}

\begin{lem}[Heap Update]\label{lem:heap-update} If $\SE H \colon \Psi$ and $\Psi(\ell^H_1) = a$ then
\begin{align*}
  (a)&\ \text{If $\Psi(\ell^H_2) = a$ and $\ell^H_1 \notin_H \ell^H_2,$ (the occurs check passes)}&\text{then $\SE \hup{H}{\ell^H_1}{\bound{\ell^H_2}} \colon \Psi.$}\\
  (b)&\ \text{If for all $i$, $\Psi(\ell'^H_i) = a_i$ and $\ell^H_1 \notin_H \ell'^H_i$ and $\Sigma(c) = \vec a \to a$,}&\text{then $\SE \hup{H}{\ell^H_1}{\func{c}{\ell'^H_1,\ldots,\ell'^H_n}} \colon \Psi.$}
  \end{align*} \end{lem} 
The simple statement of Heap Update belies the complexity of its proof.
Recall that heaps and heap types are unordered (identified up to permutation),
while heap typing derivations are ordered, serving as a witness that the heap is acyclic.
The proof of Heap Update must show that no cycles are introduced,
which requires exhibiting a new acyclic ordering in, 
e.g. the derivation of $\SE \hup{H}{\ell^H_1}{\bound{\ell^H_2}} \colon \Psi.$
\paragraph{Heap Typing Proof Terms}
In the interest of rigor, we introduce proof term notation for heap typing derivations,
which allows us to give a concise, explicit construction of the topological orderings required by Heap Update.
The reader may wish to skip this section on a first reading, as it introduces significant proof machinery that is not needed elsewhere.
Recall the typing rules for heaps:
\begin{center}
\begin{tabular}{cc}
\infer[\textsc{HT-Nil}]{\SE H\colon{}\heap{}}{}
&\infer[\textsc{HT-Cons}]{\SE H\colon{}\thext{\Psi}{\ell^H}{\tau}}
   {\SE H\colon{}\Psi & H(\ell^H) = v^H & \PE v^H\colon{}\tau & \ell^H \notin\dom{\Psi}}
\end{tabular}
\end{center}
These rules result in list-structured proof terms:
\[\mathcal{D} ::= nil_H\ |\ \mathcal{D};d_{\ell^H}\]
We write $\mathcal{D}:(\SE H:\Psi)$ when $\mathcal{D}$ is a proof term of $\SE H:\Psi$.
In this notation $nil_H$ is the proof term for HT-Nil applied to heap $H$ 
and $\mathcal{D};d_{\ell^H}$ is the proof term for HT-Cons applied to subderivations
$\mathcal{D}:(\SE H\colon{}\Psi)$ and $d_{\ell^H}:(\PE v^H\colon{}\tau)$ and it is a proof of $\SE H\colon{}\thext{\Psi}{\ell^H}{\tau}$.
To state the key lemma precisely, we exploit several functions over proof terms.

The notation $\pred(\mathcal{D}, \ell^H)$ denotes the set of heap locations assigned types by $\mathcal{D}$ that appear before $\ell^H$ within $\mathcal{D}$ and $\succ(\ell^H)$ denotes the set of locations that appear after it.
The notation $\elems(\mathcal{D})$ denotes all locations assigned types by $\mathcal{D}$.
They can be defined recursively by:
\begin{align*}
\elems(nil_H) = \emptyset &&\elems(\mathcal{D};d_{\ell^H}) =\{\ell^H\} \cup \elems(\mathcal{D})&&\\
\pred(nil_H, \ell^H) = \emptyset &&\pred((\mathcal{D};d_{\ell^H}), \ell^H) = \elems(\mathcal{D})&& \pred((\mathcal{D};d_{\ell'^H}), \ell^H)=\pred(\mathcal{D}, \ell^H)\\
\succ(nil_H, \ell^H) = \emptyset &&\succ((\mathcal{D};d_{\ell^H}), \ell^H) = \emptyset && \succ((\mathcal{D};d_{\ell'^H}), \ell^H)=\succ(\mathcal{D},\ell^H) \cup \{\ell'^H\}
\end{align*}
We now have the machinery to state a subclaim which entails both claims of Heap Update.
\begin{subclaim}[Heap Reordering]
  If $\mathcal{D}\colon{}(\SE H:\Psi), H(\ell^H_1)=\free{}[a], \ell^H_1 \in
  \dom{\Psi}, \ell^H_2 \in \dom{\Psi}, \ell^H_1 \notin_H \ell^H_2$ then exists
  $\mathcal{D'}\colon{}(\SE H:\Psi)$ where $\succ(\mathcal{D'},
  \ell^H_1) \subseteq \succ(\mathcal{D}, \ell^H_1)$ and
  $\pred(\mathcal{D'}, \ell^H_2) \subseteq \pred(\mathcal{D}, \ell^H_2)$
  and $\ell^H_2 \in \pred(\mathcal{D'}, \ell^H_1)$
\end{subclaim}
\begin{proof}
  By lexicographic induction on $|\mathcal{D}|$ and $|\succ(\mathcal{D}, \ell^H_1)|$.
  We give an explicit construction of the proof term $\mathcal{D}'$ as a function of $\mathcal{D},\ell^H_1,$ and $\ell^H_2$ in functional pseudocode, then show the construction obeys the desired properties in each case.
  Here $str(d_{\ell^H})$ and $weak(d_{\ell^H})$ refer to the typing derivations that result from appeals to the heap value strengthening and weakening lemmas, respectively (Lemmas~\ref{lem:heap-str} and \ref{lem:sweak}).

\newcommand{\DD}{\mathcal{D}}
\newcommand{\CA}{\Rightarrow}
\newcommand{\str}{\tt str}
\newcommand{\weak}{\tt weak}
\newcommand{\intwo}{\hspace{0.2in}}
\begin{tabular}{ll}
Case & $\DD'(\DD,\ell_1,\ell_2)=$\\\hline
  &${\tt case}~\DD~{\tt of}$\\
1 &\ \ ${\tt nil}~\CA~\DD$\\
2 &${\tt|\ nil'}, d_\ell~\CA~\DD$\\
3 &${\tt|\ }\DD_1;d_{\ell};d_{\ell 1}~\CA~\DD$\\
4 &${\tt|\ }\DD_1;d_{\ell 1};d_{\ell 2}~\CA~\DD_1;\str(d_{\ell 2});\weak(d_{\ell 1})$\\
5 &${\tt|\ }\DD_1;d_{\ell};d_{\ell 2}~{\tt where}~\ell\neq \ell_1~\CA$\\
     &\indentone{${\tt if}~(\ell_1' \notin \ell_2)~{\tt then}$}\\
5a   &\indenttwo{$\DD'((\DD_1;\str(d_{\ell 2})),\ell_1,\ell_2);\weak(d_{\ell})$}\\
  &\indentone{} ${\tt else}$\\
5b   &\indenttwo{${\tt let}~\DD_2=\DD'((\DD_1;d_{\ell}),\ell_1,\ell)~{\tt in}$}\\
     &\indenttwo{$\DD'((\DD_2;d_{\ell 2}),\ell_1,\ell_2)$}\\
  &${\tt|\ }\DD_1;d_\ell;d_{\ell'}~{\tt where}~ \ell'\neq \ell_1,\ell'\neq \ell_2~\CA$\\
  &\indentone $\DD'((\DD_1;d_\ell),\ell_1,\ell_2);d_{\ell'}$\\
\end{tabular}
% \begin{verbatim}
% Case  D'(D, l1, l2)=
%         case D of
% 1         nil => D
% 2       | nil', dl => D
% 3       | D1;dl; dl1 => D
% 4       | D1;dl1;dl2 => D1;str(dl2);weak(dl1)
% 5       | D1;dl; dl2 where l != l1 =>
%             if (l' notin l2) then
% 5a              D'((D1;str(dl2)),l1,l2);weak(dl)
%             else
% 5b             let D2 = D'((D1;dl),l1,l) in 
%                D'((D2;dl2),l1,l2)
%         | D1;dl;dl' where l'!=l1,l'!=l2 =>
% 6           D'((D1;dl),l1,l2);dl'
% \end{verbatim}
  \begin{itemize}
  \item Cases 1 and 2 hold vacuously because our preconditions only hold for  $|\mathcal{D}| \geq 2$.
  \item Case 3: In this case $\ell^H_2 \in \pred(\mathcal{D}, \ell^H_1)$ (either $\ell'^H = \ell^H_2$
or $\ell_2^H$ appears elsewhere in $H$), so there is no work to be done.
  \item Case 4: By the assumption that $\ell^H_1 \notin_H \ell^H_2$, we can apply Lemma~\ref{lem:heap-str}, yielding $\SE \hext{H'}{\ell^H_2}{v_2} : \thext{\Psi'}{\ell^H_2}{\tau_2}$.
By Lemma~\ref{lem:sweak}, $\thext{\Psi'}{\ell^H_1}{v_1} \vdash v_1 : \tau_1,$ so $\SE \hexts{H'}{\ell^H_2 \hook v_2, \ell^H_1 \hook v_1} : \hexts{\Psi'}{\ell^H_2 : \tau_2, \ell^H_1 : \tau_1}$, which satisfies the requirements.
  \item Case 5: By Lemma~\ref{lem:occ-tot}, either $\ell'^H \in_H \ell^H_2$ or $\ell'^H \notin_H \ell^H_2$.
  \item Case 5a: By Lemma~\ref{lem:heap-str}, $\hext{H}{\ell^H_2}{v_2} : \hext{\Psi}{\ell^H_2}{\tau_2}$ so we can apply the IH on $\hext{H'}{\ell^H_2}{v_2}$ giving a derivation $\mathcal{D}_1$. The result follows in combination with Lemma~\ref{lem:sweak} on $\ell^H$.
  \item Case 5b:
By Lemma~\ref{lem:occ-tot}, $\ell^H_1 \in_H \ell'^H$ or $\ell^H_1 \notin_H \ell'^H$. In this case $\ell^H_1 \notin_H \ell'^H$. 
Otherwise by Lemma~\ref{lem:occ-trans}
$\ell^H_1 \in_H \ell^H_2$, but we assumed $\ell^H_1 \notin_H \ell^H_2$ and this is a contradition because occurs is a function (Lemma~\ref{lem:occ-tot}).
Thus we can apply the IH on $\hext{H'}{\ell'^H}{v'}$ (because $|\hext{H'}{\ell'^H}{v'}| < |H|$) to swap $\ell^H_1$ with $\ell'^H$ resulting in a derivation $\mathcal{D}_2$.
The IH tells us $|\succ(\mathcal{D}_2, \ell^H_1)| < |\succ(\mathcal{D}, \ell^H_1)|$,
allowing us to apply the IH a second time on the derivation $\mathcal{D}_2;d_2$.
The second IH implies the result.
  \item Case 6: This case is direct by the IH.
  \end{itemize}
\end{proof}

Given Heap Reordering, the first claim of Heap Update follows directly with the following derivation
for $\bound{\ell^H_2}$
\[\infer[\textsc{HV-Bound}]{\Psi'\ent \bound{\ell^H_2} : \tau}
 {\infer[\textsc{WV-}\ell^H]{\Psi'\ent \ell^H_2 : \tau}
 {\Psi'(\ell^H_2)=\tau}}\]
where $\Psi'$ is the heap type assigned by $\mathcal{D}'$ to the prefix of $\ell^H_2$, which must
contain $\ell^H_1$.

The second claim follows by iterating Heap Reordering, and because Heap Reordering preserves the
predecessors of $\ell^H_1$.

\subsubsection{Trails}
\begin{lem}[Trail Update]\label{lem:trail-up} Introducing and binding free variables both preserve the validity of the trail:
  \begin{align*}
    (a)&\text{If $S \ent T \ok$ and $S(\ell^H)=\free{}[a]$ then $\hup{S}{\ell^H}{w}\ent \uptrail(\ell^H:a,T) \ok$.}\\
    (b)&\text{If $S \ent T \ok$ and $\ell^H$ is fresh then $\hext{S}{\ell^H}{\free{}[a]} \ent T \ok$}.
  \end{align*}
\end{lem}
First consider the typing rule for nonempty trails:
\[\infer[\textsc{Trail-Cons}]{S \ent(t,w_{env},\ell^C)::T' \ok}
{\SE S' \colon{}(\Xi,\Psi') \hskip 0.1in \Psi' \ent w_{env} \colon{} \tau
&\unwind(S,t)=S' & S' \ent T' \ok
&\Psi' \ent \ell^C\colon{}\neg\heap{\texttt{env}\colon{}\tau}}\]
Trails are well-typed so long as unwinding (as used in backtracking) results in a well-typed state.
%The typing rule for trails is based entirely on properties of the state we enter after backtracking.
%: showing that an update preserves validity consists simply
%of showing that it does not change the result of backtracking (modulo perhaps introducing unused values).

Claim (a) says binding a free variable $X$ to a term represented by word value $w$, then adding $X$ to the trail, results in a well-typed trail.
$S(\ell^H)=\free{}[a]$ then $\hup{S}{\ell^H}{w}\ent \uptrail(\ell^H:a,T) \ok$ iff unwinding results in a well-typed store.
Because it unwinds to the same store as does $T,$ this is true by assumption.

Claim (b) is weakening principle for trails, which follows from weakening for stores.
This claim shows that the trail does not need to be modified when a fresh variable is allocated, only when it is bound to a term. 
Claim (b) follows from the following subclaim:
\begin{subclaim}$\unwind(\hext{S}{\ell^H}{\free{}[a]},t) = (\hext{S'}{\ell^H}{\free{}[a]})$\end{subclaim}
The subclaim holds by induction on $t$, completing the proof of Lemma~\ref{lem:trail-up}.
\begin{lem}[Backtracking Totality]\label{lem:bt-tot} For all trails $T,$ if $\cdot\ent S:(\Xi;\Psi)$ and $S \ent_{\Sigma;\Xi} T \ok$ then
either $\bt(S,T)=m'$ and $\SE m' \ok$ or $\bt(S,T)=\bot$\end{lem}
\begin{proof}
By cases on $S\ent T \ok$.
\end{proof}

\subsubsection{Dynamic Unification}
Unification uses a simple lemma on pointer following:
\begin{lem}[End Correctness]\label{lem:end-corr}  If $\PE \ell^H :a$ and $\SE H : \Psi$
then $\edn(H,\ell^H)=\ell_H'$ and $\PE \ell'^H :a$ and either $H(\ell'^H) = \free{}[a]$
or $H(\ell'^H)=\func{c}{\vec w}$\end{lem}
  \begin{proof}
    By induction on the derivation $\PE \ell^H:a$.
  \end{proof}
Runtime unification is total and results in a well-typed store and trail.
\begin{lem}[Soundness of \unify]\label{lem:unify-sound} If $\SE H: \Psi,
\label{lem:unify-soundness} 
 \Psi \vdash w_1 : a, \Psi \vdash w_2 : a,$ then $\unify(S,T,w_1,w_2) = \bot$ or
$\unify(S,T,w_1,w_2) = (S', T')$ where $\SE S' :  (\Psi;\Xi)$ and $S'\ent T'\ok$\end{lem}
\begin{proof}
We prove the claim by simultaneous induction with the following subclaim:
\begin{claim} 
\label{claim:unify-args-soundness}
For all argument lists (push-down lists) $(w_1::\cdots::w_n)$ and $,(w_1'::\cdots::w_n'),$
if for all $i$, $\Psi \vdash w_i : a_i$ and $\Psi \vdash w_i' : a_i$ then $\unifyargs(S,T,(w_1::\cdots::w_n),(w_1'::\cdots::w_n'), \ldots, (w_n,w_n')) = \bot$ or\\
$\unifyargs(S,T,(w_1::\cdots::w_n),(w_1'::\cdots::w_n')) = (S',T')$ where $\SE S' : (\Xi;\Psi)$ and $S'\vdash T'\ok$.
\end{claim}
\begin{proof}
Lemma \ref{lem:unify-soundness} is by induction on the size of the type $\Psi$
and Claim \ref{claim:unify-args-soundness} is by induction on the argument lists $(w_1::\cdots::w_n),(w_1'::\cdots::w_n')$.
The inductive case proceeds by cases on the form of $\edn(S,\ell^H_1)$ and $\edn(S,\ell^H_2)$ using Lemma~\ref{lem:end-corr}.
We present only a few of the success cases here, for the remaining cases are straightforward.

\case{$(\free{}[a], \free{}[a]):$}
Case on  $\edn(S,\ell^H_1) = \edn(S,\ell^H_2)$ holds. If it does, we apply the first rule, else by case assumption we have $\edn(S,\ell^H_2) \notin_S \ell^H_1$ and apply the second, then apply Lemma~\ref{lem:heap-update} to get $\hup{S}{\ell^H_1}{\bound{\ell^H_2}} : (\Psi;\Xi)$ and Lemma~\ref{lem:trail-up} to get $\hup{S}{\ell^H_1}{\bound{\ell^H_2}} \ent T' \ok$.
:
{\footnotesize\begin{center}
\begin{tabular}{cc}
\infer[\unify=]{\unify(S,T,\ell^H_1,\ell^H_2) = (S,T)}
        {\edn(\ell^H_1) = \ell^H & \edn(\ell^H_2) = \ell^H}&
\infer[\unify\ \free{}]{\unify(S,T,\ell^H_1,\ell^H_2) = (\hup{S}{\ell'^H_2}{\bound{\ell^H_1}},T')}
        {\deduce{\Psi(\ell'^H_1) = a \hskip 0.1in \ell'^H_2 \notin_S \ell^H_1\hskip 0.1in \uptrail(T,(\ell'^H_2:a))=T'}
          {\edn(S,\ell^H_1) = \ell'^H_1 & \edn(S,\ell^H_2) = \ell'^H_2 & S(\ell'^H_2) = \free{}[a]}}
\end{tabular}
\end{center}}

\case{$(c\langle\ell^H_1,\ldots,\ell^H_n\rangle,c\langle\ell'^H_1,\ldots,\ell'^H_n\rangle)$}
We apply the IH on the subclaim for $\unifyargs$ and the result follows immediately (and similarly if $\unifyargs$ were to fail).
\[\infer[\unify\ c\langle\rangle]{\unify(S,T,\ell^H_1,\ell^H_2) = (S',T')}
        {\deduce{\edn(S,\ell^H_2) = \ell'^H_2\hskip 0.1in\unifyargs(S,T,\vec \ell^H, \vec \ell'^H) =(S',T')}
          {S(\ell'^H_1) = \func{c}{\vec \ell^H} & S(\ell'^H_2) = \func{c}{\vec \ell'^H} & \edn(S,\ell^H_1) = \ell'^H_1}}\]
\end{proof}
\end{proof}
The essential cases of progress and preservation for the simply-typed system are instructions such as {\tt get\_val} and {\tt unify\_val} that rely on unification.
Those cases of progress and preservation follow from the unification soundness lemma above.
Moreover, progress and preservation for the simply-typed system are subsumed by their dependently-typed equivalents.
Most of the proofs above carry over readily. Any important differences are covered in Section \ref{sec:dep-met}.

\section{Typed Compilation in Proof-Passing Style}
\label{sec:lf-encoding}
Our certification approach is based on specifying the semantics of a T-Prolog program as an LF signature.
Before we can certify the correctness of compilation, we give a mechanical translation from T-Prolog programs to LF signatures (see Figure \ref{lftrans} for example):
\begin{itemize}
\renewcommand\labelitemi{--}
\item A type $a$ in T-Prolog translates to an LF constant $a \colon \type$.
\item A constructor $c \colon \vec a \to a$ translates to an LF constant of the same type.
\item A predicate $p \colon \vec a \to \prop$ translates to an LF constant $p \colon \vec a \to \type$.
\item {A clause $C$ of form $G \verb+ :- + \mathrm{SG}_1, \ldots,
  \mathrm{SG}_n.$ translates to an LF constant (of dependent function type) $C\colon\Pi\ \Delta. \Pi\
  \vec{\mathrm{SG}}. G$ where $\Delta$ consists of the free variables
  of $\Pi\ \vec{SG}.G$.}
\item Executing a query \verb+?-+$G.$ translates to searching for a proof of $G$.
\end{itemize}

\begin{exmp}[T-Prolog Program with its LF Signature]\label{lftrans}
%\scriptsize
%\figrule
\begin{center}
\begin{tabular}{ll}
\verb+nat:type.+                                                        &\verb+nat:type.+\\
\verb+zero:nat.+                                                        &\verb+zero:nat.+\\
\verb+succ:nat +$\to$\verb+ nat.+                                      &\verb+succ:nat +$\to$\verb+ nat.+\\
\verb+plus:nat +$\to$\verb+ nat +$\to$\verb+ nat +$\to$\verb+ prop.+ &\verb+Plus:nat +$\to$\verb+ nat +$\to$\verb+ nat +$\to$\verb+ type.+\\
\verb+plus(zero,X,X).+                                                   &\verb+Plus-Z:+$\Pi$\verb+X:nat. plus zero X X.+\\
\verb+plus(succ(X),Y,succ(Z)) :-+                                       &\verb+Plus-S:+$\Pi$\verb+X:nat. +$\Pi$\verb|Y:nat.| $\Pi$\verb+Z:nat.+\\
\verb+  plus(X,Y,Z).+                                                    &\: \: $\Pi$\verb+D:plus X Y Z. +\\
                                                                           & \: \: \: \: \verb|plus (succ X) Y (succ Z).|
\end{tabular}
\end{center}
\end{exmp}

Now that we have defined the semantics of a Prolog program in LF, we can describe our certification approach.
The TWAM certification approach can be summed up in the following slogan:
\begin{center}
  Proof-Carrying Code + Programming As Proof Search = Proof-Passing Style
\end{center}

Proof-carrying code is the technique of packaging compiled code with a formal proof that the code satisfies
some property. Previous work \cite{necula1998design} has used proof-carrying code to build certifying compilers
which produce proofs that the programs they output are memory-safe. Our insight is that by combining this technique
with the programming-as-proof-search paradigm that underlies logic programming, our compiler can produce proofs of
a much stronger property: partial dynamic correctness.

The programming-as-proof-search paradigm tells us that partial dynamic correctness consists of the following theorem, stated
informally here and formally in Section \ref{sec:dep-met}:

\textbf{Theorem 1: } If a query \verb+?-+$G.$ succeeds, there exists a proof $M$ of $G$ in LF.

Our compiler need only output enough information that the TWAM typechecker can reconstruct the proof of Theorem 1.
This requires statically proving that whenever \textit{any} proof search procedure $p$ would return, the corresponding predicate $P$ would have a proof in LF. 
This proof boils down to accumulating facts when unification succeeds and annotating all return points with the resultant LF proof terms.
This \textit{proof-passing} style of programming is essential to the type system of the TWAM.
It is worth noting also that proof-passing style is needed only at compile-time, because TWAM also supports \textit{proof-erasure}.
Thus the only runtime cost of certifying compilation with TWAM is the cost of using SWAM vs. other variants of the WAM;
TWAM introduces no additional overheads compared to SWAM.

\section{Dependently-Typed WAM}
\label{sec:dependent-wam}
The simply-typed system presented in Section \ref{sec:simple-wam} is insufficient to prove that compiled programs implement a given LF signature. 
The \verb+succeed+ instruction
\[\infer{\GE \succeed; I \ok}{~\vspace{0.08in}}\]
trivially typechecks in any context, but we wish to prove that a program only succeeds if 
a proof $M$ of some query $A$ exists in LF. We begin our dependently-typed development by requiring
exactly that in the typing rule:
\[\infer{\DGE \succeed[M\colon{}A];I \ok}{\DE M \colon{} A}\]
Yet if this was the only change we made, we could never compile meaningful programs because the premise would be too difficult to fulfill.
We make this premise easier to meet by introducing the ability for continuations to accept LF proof terms as arguments.
Because the \verb+succeed[M:A]+ instruction generally occurs in the top-level success continuation for a query, we can make this continuation  accept a proof of $M$ as an argument $x$ and supply $x$ as the proof term for \texttt{succeed}, passing the burden of proof onto the caller.

In this way, we can decompose the proof argument for $M$ into one argument for each basic block of the proof search algorithm.
This too is nontrivial: whether the proof $M$ exists for a given query $A$ cannot be known until $A$ is executed at runtime, but certification occurs at compile-time.
In order to reason statically about runtime proof search, the type system must connect LF terms with runtime constructs such as registers and heap values.
Whenever a unification succeeds at runtime (i.e. we learn that we can apply a particular rule), we need some way to say the same terms should be unified in the statics.
Without a mechanism for translating between the runtime values and the static LF terms, we have no mechanism by which to learn new facts during proof search, and thus no way to construct nontrivial LF proofs.

The simplest possible relation between an LF term $M$ and a heap value $v$ is the notion of equality. Since heap values can also contain pointers
and can exhibit sharing structures not visible in the LF term, we might 
more accurately think of this equality relation as ``$v$ encodes $M$''.
We add this notion to our type system by introducing
\textit{singleton types} $\sing(M \colon a)$ for values that encode an
LF term $M$ which itself has type $a$. This is the only fundamentally new
value type, though other aspects of the type system will change as well to accommodate
the presence of proofs.

In particular, we introduce a context $\Delta$ that contains the types of all LF variables in scope, 
which are introduced either in the parameters of a code value or by the \verb+put_var+ instruction.
Furthermore, we introduce a notion of static unification $M_1 \sqcap M_2$ which allows us to import
knowledge learned from runtime unification into an LF proof.

To see the interaction between runtime and static unification concretely, 
consider the zero case of \verb+plus+, which (using the proof terms of Example~\ref{lftrans}) compiles to

%TODO: Typeset nice
\begin{exmp}[TWAM Compilation]
\label{ex:twam-comp}~\\
\begin{tabular}{l}
\verb|plus-zero/3|$~\mapsto~\code[\Pi~X,Y,Z:nat.
     \heap{r_1:\sing(X), r_2:\sing(Y), r_3:\sing(Z), r_4:\Pi \_:({\tt Plus~X~Y~Z}).~\neg\heap}]($\\
  \indentone{${\tt put\_tuple}~r_4, 4;$}\\
  \indenttwo{${\tt set\_val}~r_1;$}\\
  \indenttwo{${\tt set\_val}~r_2;$}\\
  \indenttwo{${\tt set\_val}~r_3;$}\\
  \indenttwo{${\tt set\_val}~{\tt ret};$}\\
\verb|# plus-succ/3, not shown, has the same three natural number|\\
\verb|# parameters X, Y, Z, so we pass them in when constructing|\\
\verb|# the failure continuation|\\
\indentone{${\tt push\_bt}~r_4, {\tt (\verb|plus-succ/3|~X~Y~Z)};$}\\
~\\
\indentone{${\tt get\_str}~r_1, {\tt zero/0};$}\\
\indentone{${\tt get\_var}~r_2, r_3;$}\\
\indentone{${\tt jmp~(ret~(\verb|Plus-Z|~Y))};$}\\
)
\end{tabular}
\end{exmp}
The syntax \verb+code+$[\Delta.\Gamma](I)$ denotes a code value with body $I$ which expects the register file to
have type $\Gamma$ and where $\Gamma$ may refer to the LF variables in $\Delta$. The line \verb+jmp (r4 (Plus-Z Y))+
is an example of proof passing in action. Here the success continuation \verb|r4| expects a proof of the relevant predicate:
in this case \verb|plus X Y Z|. The \verb+jmp+ instruction constructs a proof \verb+Plus-Z Y+ to satisfy this requirement.
The proof \verb+Plus-Z Y+ has type \verb+plus zero Y Y+, so this code only typechecks
if $X = zero$ and $Y = Z$, which is exactly what we learn when \verb|get_str| and \verb+get_var+ succeed, respectively.

\subsection{Instruction Statics and Dynamics}
\label{sec:istat-dyn}
In this section we detail the type system changes made to support LF terms and singleton types.
The instruction set and dynamic judgements are fundamentally identical to that of the SWAM, but both are augmented with additional annotations as needed by the addition of LF.
For example, because the {\tt put\_var} instruction introduces an LF variable $x$, we now write $\dputvar{r}{x:a.}~I$ instead of $\putvar{a}{r}; I$ to indicate that there is an LF variable $x:a$ in scope for the remaining instructions $I$, a feature we will use to write proof terms.
Complete dynamics for the dependent TWAM are given in the electronic appendix.
In the following sections we detail the judgements that differ significantly from the simply-typed system.
\begin{table}
\begin{tabular}{ccc}
\label{tab:judgements}
Judgement(s) & Meaning                & Defined In\\\hline
$\DGE_{\Sigma;\Xi} I \ok$   & Basic Block Well-Typed & \ref{sec:istat-dyn}\\
$\DGE {I :_s}_{\Sigma;\Xi}~J, {I :_t}_{\Sigma;\Xi}~J$ & Spine Well-Typed & \ref{sec:dep-spine-stat}, \ref{sec:dep-env-stat}\\ 
$\DE M_1\sqcap M_2 = \sigma, \bot$ & Static Unification & \ref{sec:dep-stat-unif}\\
$\DE M_1\in M_2, M_1\notin M_2$ & Static Occurs Check & \ref{sec:dep-stat-unif}\\
$\DGE op :\tau$ & Operand Well-Typed & \ref{sec:word-oper}\\
$\DE M : A$ & LF Term Well-Typed         & \cite{Harper93aframework}\\
$\DE A : K$ & LF Type Family Well-Kinded & \cite{Harper93aframework}\\ 
\hline
$\DGE v :\tau$ & Heap Value Well-Typed & \ref{sec:dep-hvtype}\\
$\DE  H :\Psi$ & Heap File Well-Typed & \ref{sec:simp-rep-inv}\\
$\DGE w :\tau$ & Word Value Well-Typed & \ref{sec:word-oper}\\
$\DPE R :\myG$ & Register File Well-Typed & \ref{par:rftype}\\
$\SE (\Delta;\mu) \colon{} H$ & Heap Mapping Unique&\ref{par:con-map}\\
$\DE  T \ok$ & Trail Well-Typed &\ref{par:trail-dep-inv}\\
$\DPE \vec \ell^H \reads J_s$ & Prolog Read Spine Invariant& \ref{par:mach-state-inv}\\
$\DPE (\vec \ell^H, \ell^H,c) \writes J_s$ & Prolog Write Spine Invariant& \ref{par:mach-state-inv}\\
$\DPE(n,r,\vec\ell^H)\writes J_t$ & Tuple Spine Invariant& \ref{par:mach-state-inv}\\
$\SE  m\ok$ & Machine Well-Typed &\ref{par:mach-state-inv}\\\hline
$R\ent op\eval w$ & Operand Evaluation& \ref{sec:word-oper}\\
$R\ent op\squig w$ & Operand Resolution& \ref{sec:word-oper}\\
$w\eval w'$ & Word Evaluation& \ref{sec:word-oper}\\
$w \path, w \canon$ & Word Canonical Forms &\ref{sec:word-oper}\\
$\edn(S,\ell^H)$ & Pointer Following & \ref{sec:trailing}\\
$\ell^H_1 \in_S \ell^H_2$ & Dynamic Occurs Check & \ref{sec:trailing} \\
$\unify(S,T,\ell^H_1,\ell^H_2)$ & Dynamic Unification & \ref{sec:trailing}\\
$\unifyargs(S,T,\vec \ell^H, \vec \ell'^H)$ & Dynamic Unification & \ref{sec:trailing}\\
$\uptrail(x@\ell^H:a,T)=T'$ & Trail Update & \ref{par:trail-dep-inv}\\
$\unwind(S,\Delta,t) = (\Delta,S)$ & Trail Unwinding & \ref{par:trail-dep-inv}\\
$\bt(S,T)=m,\bot$ & Backtracking & \ref{sec:cont-bt}\\
$m\step m',m\fails,m\done$ & Stepping & \ref{sec:simp-op}\\\hline
\end{tabular}
\caption{Index of Typing and Evaluation Judgements}
\end{table}
A listing is given in Table 1.
\subsubsection{LF Terms}
The typing rule for {\tt succeed} is as given in Section \ref{sec:dependent-wam}:
\[\infer[\textsc{Succeed}]{\DGE \succeed[M\colon{}A];I \ok}{\DE M : A}\]
Here $M$ is an LF term and $A$ is an LF type (we write $A$ for arbitrary LF type families and $a$ for types corresponding specifically to Prolog terms).
Thus we extend the syntax of TWAM with the syntax of LF (here $c$ stands for type family constants and term constants):
%\cite{Harper93aframework}
%.
%,necula1998design
\begin{center}
\begin{tabular}{ll}
  LF Kinds & $K ::= {\tt type}\ |\ \Pi x:A.K$\\
  LF Type Families &  $A ::= c\ |\ \Pi x:A.A\ |\ A\ M$\\
  LF Terms & $M ::= x\ |\ c\ |\ M\ M\ |\ \Pi x:A.M$\\
\end{tabular}
\end{center}
Note that the TWAM need not be instrumented with LF proof terms at runtime: LF proofs are merely given as type annotations as an aid to establishing the metatheorem of Section~\ref{sec:dep-met}.
LF terms make numerous appearances in TWAM.
For example, because \textsc{Putvar} introduces a free variable at runtime, it introduces an LF variable $x:a$ in the statics as well:
\[\infer[\textsc{Putvar}]{\DGE \dputvar{r}{x:a.}~I \ok}
     {\Delta,x\colon{}a;\hups{\Gamma}{r:\sing(x:a)}\ent I \ok}\]
\subsubsection{Words and Operands}
\label{sec:word-oper}
The typing rules \textsc{Jmp} and \textsc{Mov} appear as before:
\begin{center}
\begin{tabular}{cc}
\infer[\textsc{Jmp}]{\DGE_{\Sigma;\Xi}\jmp op, I\ok}{\Xi(\ell^C) = \neg\myG' & \SE \myG' \leq \myG}&
\infer[\textsc{Mov}]{\DGE \mov{r_d}{op}; I \ok}
  {\DGE op : \tau & \Delta;\thup{\Gamma}{r_d}{\tau} \ent I \ok}
\end{tabular}
\end{center}
However, both instructions rely on \textit{operands}.
In TWAM, we generalize operands (and word values) so they can accept LF terms as arguments:
\begin{center}
\begin{tabular}{cc}
  operands & $op ::= \ell^C\ |\ r\ |\ op\ M\ |\ \lambda x:A.~op$\\
  word values & $w ::=\ell^C\ |\ \ell^H\ |\ w\ M\ |\ \lambda x:A.~w$\\
\end{tabular}
\end{center}
With this change we also update the statics and (big-step) dynamics for operands and word values.
The main dynamic judgement is still $R \ent op \eval w$ (Operand Evaluation),
but we add auxilliary judgements $R\ent op\squig w$ (Operand Resolution) and $w\eval w'$ (Word Evaluation).
\begin{center}{\footnotesize
\begin{tabular}{cccc}
 \infer[\ell^C\eval]{\ell^C \eval \ell^C}{}
&\infer[\lambda\eval]{(\lambda x:A.~w) \eval (\lambda x:A.~w')}{w \eval w'}      
&\infer[\beta\eval]{op\ M \eval w''}{op \eval (\lambda x:A.~w') & [M/x]w' \eval w''}
&\infer[op\ M\eval]{op\ M \eval w\ M}{op \eval w & w \path}
\end{tabular}

\begin{tabular}{cccc}
\infer[op\ M\squig]{R \ent op\ M \squig w\ M}{R \ent op \squig w}&
\infer[\lambda\squig]{R \ent (\lambda x:A.~op) \squig (\lambda x:A.~w)}{R \ent op \squig w}&
\infer[r\squig]{R \ent r \squig w}{R(r) = w}&
\infer[op\eval]{R \ent op \eval w'}{R \ent op \squig w & w \eval w'}
\end{tabular}

\begin{tabular}{ccc}
\infer[op\ r]{\DGE r : \tau}{\myG(r) = \tau}
&\infer[op\ (w\ M)]{\DGE (w\ M) : [M/x] \tau}{\DGE w : \Pi x:A.~\tau & \DE M : A}
&\infer[op\ \lambda]{\DGE (\lambda x:A.~w) : (\Pi x:A.~\tau)}{\DE A : \type & \Delta,x:A;\myG\ent w : \tau}\\

\end{tabular}
\begin{tabular}{lll}
\infer[w\ \ell^H]{\DPE \ell^H : \tau}{\Psi(\ell^H) = \tau}
&\infer[w\ (w\ M)]{\DPE (w\ M) : [M/x] \tau}{\DPE w : \Pi x:A.~\tau & \DE M : A}
&\infer[w\ \lambda]{\DPE (\lambda x:A.~w) : (\Pi x:A.~\tau)}{\DE A : \type & \Delta,x:A;\Psi\ent w : \tau}
\end{tabular}}
\end{center}
In \textsc{Jmp}, the generalization of operands supports proof-passing.
In \textsc{Mov}, it supports tail-call optimization as used in Section~\ref{sec:tco}.
As in LF, we have a notion of canonical forms for words, written $w \canon$, with an auxilliary judgement $w \path$:
\begin{center}
\begin{tabular}{lllll}
&\infer[\ell\path]{\ell \path}{}
&\infer[w\ M\path]{w\ M \path}{w\path}
&\infer[w\canon]{w \canon}{w \path}
&\infer[\lambda\canon]{(\lambda x:A.~w) \canon}{w \canon}
\end{tabular}
\end{center}
To simplify the proofs, the typing invariants for machine states require canonicity.
However, because canonical forms always exist~\cite{Harper93aframework} and involve only static-level computation,
the choice of when to require canonical forms is irrelevant.
%Operands also enable \textsc{Close} to close over LF variables and are used in \textsc{Mov}
%to enable tail-call optimization.
\subsubsection{Continuations}
The rules \textsc{Close} and \textsc{BT} also use operands to track LF proof terms in closures,
but those operands are syntactically restricted to $\ell^C\ \vec{M}$ in order to avoid closures within closures,
which would needlessly complicate the dynamics.
Furthermore, we see in \textsc{Close} that the type of continuations has been generalized to $\Pi\vec{x}:\vec{A}.~\neg\myG'$: a continuation can take any number of LF terms, which may freely mix Prolog terms and proof terms. 
Here the terms $\vec{M}$ are a static component of the environment, stored in the closure, while the $\vec{x}$ are static arguments supplied by the caller:
\begin{center}
\begin{tabular}{ll}
 \infer[\textsc{Close}]{\DGE \close r_d, r_s, (\ell^C\ \vec M); I \ok}
  {\deduce{\DGE (\ell^C\ \vec M):\Pi \vec x:\vec A.~\neg
      \thup{\myG'}{\texttt{env}}{\tau}} {\Gamma(r_s)= \tau &
      \Delta;\thup{\myG}{r_d}{\Pi \vec x:\vec A.~\neg \myG'} \ent I \ok}}
  \infer[\textsc{BT}]{\DGE \branch r, (\ell^C\ \vec M);I \ok}{\DGE I \ok &
    \Gamma(r) = \tau & \DGE (\ell^C\ \vec M):
    \neg\heap{\texttt{env}:\tau}}
\end{tabular}
\end{center}
\subsubsection{Static Unification}
\label{sec:dep-stat-unif}
We arrive now at what is arguably the most novel and surprising technical result of the TWAM type system: Static unification as used in the TWAM type system is not only in harmony with Prolog-style runtime unification, but is strong enough to enable the type-checking of LF proofs.
Without our static unification mechanism, it would in general be impossible to show the proof-terms returned by a clause were well-typed (consider the {\tt jmp} in Example~\ref{ex:twam-comp}).

As before, {\tt get\_val} unifies two Prolog terms stored in registers $r_1, r_2$.
Thanks to the addition of singleton types, the type system now has access to LF terms $M_1, M_2:a$ describing the values of $r_1, r_2$.
The subtlety of static unification lies in the fact that because the exact values of $r_1$ and $r_2$ are unknown until runtime, the terms $M_1$ and $M_2$ cannot be the exact values of $r_1$ and $r_2$.
Rather, they will merely be some terms that \emph{unify} with the eventual values of $r_1$ and $r_2$.
What we find novel and surprising is that this partial knowledge represented by $M_1$ and $M_2$ 
is simultaneously strong enough to certify proof search, yet consistent with the actual behavior at runtime.

To typecheck {\tt get\_val}, we unify the terms $M_1, M_2$ at compile-time.
We write $\DE M_1\sqcap M_2 = \sigma$ to say they successfully unify with most-general unifier $\sigma$.
We apply the substitution $\sigma$ while typechecking the remaining instructions $I$.
The substitution notation $\sub{\sigma}\Delta$ indicates that $\sigma$ substitutes for an arbitrary set of variables from $\Delta$ (i.e. $\dom{\sigma} \subseteq\dom{\Delta}$) and that the replacees $\dom{\sigma}$ should be removed from $\Delta$ in the process: the need for this variant of substitution arises because the variables of a most-general unifier $\dom{\sigma}$ may appear at arbitrary positions throughout $\Delta$.

The rule \textsc{Getval-F} says it is also possible that we statically detect unification failure, written $\DE M_1\sqcap M_2 =\bot$, in which case the program is vacuously well-typed because unification will certainly fail at runtime, leading to backtracking.
In practice, this rule should not be necessary for useful programs, as it indicates the presence of dead code.
However, it is absolutely essential in the theory to ensure preservation in the presence of predicate calls.
{\small\begin{center}
\begin{tabular}{cc}
\infer[\textsc{Getval-S}]{\DGE \getval {r_1}{r_2}; I \ok} {\deduce{\DE M_1 \sqcap M_2
      = \sigma\hskip 0.1in \sub{\sigma}\Delta;[\sigma]\Gamma \ent
      [\sigma]I \ok} {\Gamma(r_1) = \sing(M_1:a) & \Gamma(r_2) =
      \sing(M_2:a)}}&
\infer[\textsc{Getval-F}]{\DGE \getval {r_1}{r_2}; I \ok} 
{\deduce{\Gamma(r_1) = \sing(M_1:a) \hskip 0.1in \Gamma(r_2) = \sing(M_2:a)}{\DE M_1 \sqcap M_2 = \bot}}
\end{tabular}
\end{center}}
All unification in T-Prolog and TWAM is first-order, thus the unification judgements $\DE M_1\sqcap M_2 = \sigma$ and $\DE M_1\sqcap M_2 = \bot$ 
correspond closely to standard algorithms in the literature \cite{Robinson:1965:MLB:321250.321253}.
As in dynamic unification, unification uses auxilliary occurs-check judgements $x\in M$ and $x\notin M$.
Substitutions in TWAM are capture-avoiding and simultaneous. For example, we write $[M_1/x_1,M_2/x_2]$ for a simulateous subtitution on $x_1$ and $x_2$ or $[\sigma_1,\sigma_2]$ for simultaneous composition of arbitrary substitutions $\sigma_1,\sigma_2$:
{\footnotesize\begin{center}
\begin{tabular}{ccccc}
\infer[\sqcap\cdot]{\DE x \sqcap x = \cdot}{}&
\infer[\sqcap{}x1]{\DE x \sqcap M = [M/x]}{x \not \in M}&
\infer[\sqcap{}x2]{\DE M \sqcap x = [M/x]}{x \not \in M}&
\infer[\notin{}x]{x \notin x'}{x \neq x'}&
\infer[\notin{}x\ \vec{M}]{x \notin c\ \vec M}{x \notin M_i (\forall i)}
\end{tabular}
\end{center}
\begin{center}
\begin{tabular}{cc}
\infer[\sqcap{}c]{\DE c\ M_1\ \ldots\ M_n \sqcap c\ M_1'\ \ldots\ M_n' = \sigma_n,\ldots,\sigma_1}
        {\deduce[\vdots]
          {[\sigma_{n-1},\ldots,\sigma_1]\Delta \ent [\sigma_{n-1},\ldots,\sigma_1]M_n \sqcap [\sigma_{n-1},\ldots,\sigma_1]M_n' = \sigma_n}
          {\DE M_1 \sqcap M_1' = \sigma_1}}&\infer[\in{}x]{x \in x}{}\\[0.15in]
\infer[\bot{}c1]{\DE c\ M_1\ \ldots\ M_n \sqcap c\ M_1'\ \ldots\ M_n' = \bot}
{\deduce[\vdots]{[\sigma_{i-1},\ldots,\sigma_1]\Delta \ent [\sigma_{i-1},\ldots,\sigma_1]M_i \sqcap [\sigma_{i-1},\ldots,\sigma_1]M_i' = \bot}{\DE M_1 \sqcap M_1' = \sigma_1}}&\infer[\in{}x\ \vec{M}]{x \in c\ \vec M}{x \in M_i (\exists i)}\\
\end{tabular}
\end{center}
\begin{center}
\begin{tabular}{ccc}
\infer[\bot{}c2]{\DE c\ M_1\ \ldots\ M_n \sqcap c'\ M_1'\ \ldots\ M_m' = \bot}{c \neq c'}&
\infer[\bot{}x1]{\DE x \sqcap M = \bot}{x \in M}&
\infer[\bot{}x2]{\DE M \sqcap x = \bot}{x \in M}
\end{tabular}
\end{center}}

%\begin{tabular}{c}
%\end{tabular}

\subsubsection{Spines}
\label{sec:dep-spine-stat}
Recall that a Prolog spine serves to unify some terms $M_1\sqcap M_2$, the distinction being that unlike in {\tt get\_val}, the outermost shape of $M_2$ is known statically.
As above, this unification must be made explicit in the type system.
As before, a spine type expresses a typing precondition on each unificand and a typing postcondition.
Previously the postcondition was trivial, but in our generalized \emph{dependent spine types}, the postcondition says that some unification problem $\DE M_1\sqcap M_2$ has succeeded.
We write dependent spine types as $\Pi
x_1:a_1. \cdots\Pi x_n:a_n.~(M_1 \sqcap M_2)$ to say that $M_1$ and
$M_2$ will be unified if the spine succeeds, where the $x_i$ stand for the unificand subterms associated
with each instruction of the spine.

In \textsc{Putstr}, we temporarily introduce a fresh LF variable $x$ for our new Prolog term, which is then unified with the concrete term resulting from the spine. 
In \textsc{Getstr}, the unificand is the existing term stored in $r$.
In \textsc{Unifyvar} we extend $\Delta$ with a fresh unification variable standing for the given argument (because this variable may be needed later in a proof term), 
while in \textsc{Unifyval} we do not extend $\Delta$ but rather supply an existing term as the spinal argument.
At the end of the spine, if the terms unify, then the rule $\sqcap\sigma$ applies the unifier $\sigma$ while typechecking $I$, 
else $\sqcap\bot$ says typechecking is vacuous because unification will fail at runtime.
As in {\tt get\_val}, if $\sqcap\bot$ applies at compile-time, it indicates the presence of dead code,
but it is of essential use in the preservation proof.
%As in the simply-typed setting, a spine type encodes the arguments to a unification problem and a postcondition
%on the state of the world. In the dependent setting, the postcondition of a spine states that the given
%unification has succeeded. We extend dependent spines with a binding structure that expresses the relationship
%between a unification's arguments and its postcondition. We write these spine types as 
{\small\begin{center}
\begin{tabular}{cc}
   \infer[\textsc{Putstr}]{\DGE \putstr {c} {r}; I \ok}
   {\deduce{\Delta,x:a;\thup{\Gamma}{r}{\sing(x:a)} \ent I:_s\Pi \vec
       x:\vec a.~(x \sqcap c\ \vec x)}{\Sigma(c) = \vec a \to a}}&
   \infer[\textsc{Getstr}]{\DGE \getstr {c} {r}; I \ok} {\deduce{\DGE I:_s\Pi \vec
       x:\vec a.~(M \sqcap c\ \vec x)}{\Sigma(c) = \vec a \to a &
       \Gamma(r) = \sing(M : a)}}\\[0.1in]
   \infer[\textsc{Unifyvar}]{\DGE \unifyvar r, x:a.~I:_s\Pi x:a.~J}
   {\Delta,x:a;\thup{\Gamma}{r}{\sing(x:a)} \ent I:_sJ}&
   \infer[\textsc{Unifyval}]{\DGE \unifyval r, x:a.~I:_s\Pi x:a.~J} 
     {\Gamma(r) = \sing(M:a) & \DGE [M/x] I:_s[M/x]J}\\[0.1in]
   \infer[\sqcap\sigma]{\DGE I:_s M_1 \sqcap M_2}
   {\DE M_1 \sqcap M_2 = \sigma & \sub{\sigma}\Delta;[\sigma]\Gamma
     \ent [\sigma]I \ok}&
       \infer[\sqcap\bot]{\DGE I:_s (M_1 \sqcap M_2)}
   {\DE M_1 \sqcap M_2 = \bot}
\end{tabular}
\end{center}}

\subsubsection{Environments}
\label{sec:dep-env-stat}
The typing rules for environment tuples are unchanged, since tuples are orthogonal to Prolog terms and LF in general:
{\begin{center}
\[\infer[\textsc{PutTuple}]{\DGE\puttuple r_d, n; I \ok}{\DGE I:_t (\vec\tau\to \{r_d:\cross{\vec\tau}\})& \text{(where $n = |\vec\tau|$)}}\]
\begin{tabular}{ccc}
 \infer[\textsc{Proj}]{\DGE\proj r_d, r_s, i; I \ok}{\myG(r_s) = \cross{\vec\tau} & \thup{\myG}{r_d}{\tau_i}\vdash I\ok & \text{(where $i \leq |\vec\tau|$)}}
&\infer[\textsc{SetVal}]{\DGE\setval r; I :_t (\tau\to J)}{\myG(r) = \tau & \DGE I:_t J}
\end{tabular}
\end{center}}
%\cite{Robinson:1965:MLB:321250.321253}. 

\subsection{Code and Heap Value Typing Invariants}
\label{sec:dep-hvtype}
In dependent TWAM, code values can accept LF terms as arguments, as reflected in \textsc{Code}.
Furthermore, because heap values can now have dependent types, the heap value typing judgement is now
$\DGE v^H:\tau$, where the added context $\Delta$ contains an LF variable for each free variable on the heap.
The rule \textsc{Close} is generalized to close over LF terms, while the rules $\free{},\bound{},$ and $c\langle\rangle$ are generalized to singleton types. 
As with register files, tuples and closures enforce that words are canonical for simplicity:
\begin{center}
\begin{tabular}{l}
\infer[\textsc{Close}]{\DPE \close(w_{env}, \ell^C\ \vec{M}) :\Pi\vec{x}:\vec{A}.~\neg \myG}
        {\DPE w_{env} : \tau &  \DPE\ell^C\ \vec{M}:\Pi\vec{x}:\vec{A}.~\neg\thup{\myG}{\texttt{env}}{\tau} & w_{env}\canon}\\
\infer[\langle\rangle]{\DPE\langle w_1,\ldots,w_n\rangle:\cross{\tau_1,\ldots,\tau_n}}
       {\DPE w_1 : \tau_1&w_1\canon&\cdots&\DPE w_n : \tau_n&w_n\canon}
\end{tabular}

 \begin{tabular}{cc}
    \infer[\free{}]{\DPE \free{}[x:a] : \sing(x:a)}{\Delta(x) = a}&
    \infer[\bound{}]{\DPE \bound{\ell^H} : \sing(M:a)}{\DPE \ell^H : \sing(M:a)}
  \end{tabular}

  \begin{tabular}{cc}
\infer[c\langle\rangle]{\DPE_{\Sigma;\Xi} c\args{\ell^H_1,\ldots,\ell^H_n} : \sing(c\ \vec M : a)}
    {\Sigma(c) = \vec a \to a & \DPE \ell^H_i : \sing(M_i :a_i) }&
\infer[\textsc{Code}]{\cdot \ent \code[\vec x:\vec A.~\myG](\lambda \vec x:\vec A.~I) : \Pi \vec x : \vec A.~\neg \myG}
        {(\vec x:\vec A);\GE I \ok}

  \end{tabular}
\end{center}

\subsection{Machine Typing Invariants}
The runtime behavior of a TWAM program does not depend on type information, i.e. TWAM is easily executed by first type-erasing it to SWAM and then executing the SWAM program.
However, just as TWAM adds typing and proof term annotations to instructions, our theoretical presentation
of the machine states is annotated with LF variables and proof terms, as well.
\paragraph{LF Contexts and Mappings}
\label{par:con-map}
When we prove soundness for TWAM (Theorem~\ref{dep:sound}) in Section~\ref{sec:dep-met}, we will show that for each successful execution trace, an LF proof term exists \emph{in some context} $\Delta$.
For convenience, we make that context an additional field of the machine state, but this is not strictly necessary because it contains one variable for each free variable in the heap $H$ and could thus be computed as a function of $H$.
For Theorem~\ref{dep:sound} to be meaningful, it is essential that $\Delta$ only contains Prolog terms and not arbitrary LF propositions.
Otherwise, if we wished to find a proof term for some query $A$, we could simply add $A$ to the context with {\tt put\_var} and obtain a trivial ``proof''.
Consider the following example (which assumes we have successfully defined the Riemann Hypothesis in Prolog):
\begin{verbatim}
put_var r1, x:Riemann_hypothesis.
succeed[x:Riemann_hypothesis]
\end{verbatim}
Luckily, we easily enforce that $\Delta$ contains only Prolog terms by adding a syntactic restriction in {\tt put\_var}.

The addition of LF variables affects the heap as well: free variables are now annotated as $\free{}[x:a]$ because they are in correspondence with LF variables $x$.
As a technical device to support our progress and preservation theorems, we maintain the invariant that this correspondence is unique with a mapping
$\mu$ between each variable and its unique location on the heap.
 \[\text{LF Mappings} \: \mu ::= \cdot\ |\ x@(\ell^H:a),\mu\]
The syntax $x@(\ell^H:a)$ says the LF variable $x$ has type $a$ and is located
at $\ell^H$.  The judgement $\SE (\Delta;\mu) \colon{} H$ says that $\mu$ correctly mediates $\Delta$ and $H$ (i.e. assigns a unique location in $H$ to each variable of $\Delta$):
\begin{center}
  \begin{tabular}{cc}
    \infer[\mu{}\textsc{-Nil}]{\cdot \ent (\Delta;\mu) : \heap{}}{\Delta=\cdot&\mu=\cdot}&
    \infer[\mu{}\textsc{-Skip}]{\cdot \ent (\Delta;\mu) : \hext{H}{\ell^H}{v}}
    {\cdot \ent (\Delta;\mu) : H & v \neq \free{}[x:a]}  \end{tabular}
\end{center}
\begin{center}
  \begin{tabular}{c}
\infer[\mu{}\textsc{-Cons}]{\cdot \ent (\Delta,x:a;\mu::(x@\ell^H:a)) : \hext{H}{\ell^H}{\free{} [x:a]}}
        {\cdot \ent (\Delta;\mu) : H}
  \end{tabular}
\end{center}
\paragraph{Trails}
\label{par:trail-dep-inv}
Trails are generalized in two straightforward ways.
First, failure continuations are now allowed to close over LF terms.
Second, trail typing annotations $\ell^H:a$ are now generalized to remember the corresponding LF variable name ($x@\ell^H:a$)
so that $\Delta$ can be updated accordingly in unwinding:
{\footnotesize\begin{center}
\begin{tabular}{cc}
\infer[\textsc{Trail-Cons}]{\Delta;S \ent(t,w,\ell^C\ \vec M)::{}T' \ok}
        {\deduce
          {\SE (\Delta';\mu') : H' \hskip 0.1in \Delta' \ent S' : (\Xi,\Psi')
            \hskip 0.1in \Delta';\Psi' \ent \ell^C\ \vec M :\neg\heap{\texttt{env}:\tau}}
          {\unwind(S,\Delta,t)=(\Delta';S') & \Delta';S' \ent T' \ok & \Delta';\Psi' \ent w : \tau}}&
\infer[\textsc{Trail-Nil}]{\Delta;S \ent \epsilon \ok}{}
\\[0.1in]
 $\unwind(S,\Delta,(x@\ell^H:a)::t) = \unwind(\hup{S}{\ell^H}{\free{}[x:a]},(\Delta,x:a),t)$ & $\unwind(S,\Delta,\epsilon) = (\Delta,S)$\\
$\uptrail(x@\ell^H:a,(t,w_{env},\ell^C)::T) = ((x@\ell^H:a)::t,w_{env},\ell^C)::T$&$\uptrail((x@\ell^H:a),\epsilon) = \epsilon$
\end{tabular}
\end{center}}
\paragraph{Register File Types}
\label{par:rftype}
Register file typing now requires that words are canonical, for the sake of simplicity:
\[\infer[\textsc{RF}]{\DPE \heap{r_1 \hook w_1, \ldots, r_n \hook w_n} : \heap{r_1 :\tau_1, \ldots, r_n : \tau_n}}
        {\DPE w_1 : \tau_1 & w_1 \canon & \cdots & \DPE w_n : \tau_n & w_n \canon}\]
\paragraph{Machine States}
\label{par:mach-state-inv}
The machine state typing invariants are updated to use the dependent forms of existing judgements in addition to the new invariant $\SE(\Delta,\mu):H$.
As before, spinal states each appeal to an auxilliary invariant.
\begin{center}{
\begin{tabular}{c}
\infer[\textsc{Mach}]{\cdot \ent (T,\Delta,(C,H),R,I) \ok}
{\deduce{\Delta \ent (C,H) : (\Xi;\Psi) \hskip 0.1in \DPE R : \Gamma \hskip 0.1in \DGE I \ok}
{\Delta;(C,H) \ent T \ok & \cdot \ent (\Delta,\mu) : H}}\\[0.1in]
\infer[\textsc{Mach-TWrite}]{\cdot \ent \twrite(T,\Delta,(C,H),R,I, \vec w,r,n) \ok}
        {\deduce{\Delta \ent (C,H) : (\Xi;\Psi) \hskip 0.1in  \DGE I:_t J \hskip 0.1in \DPE (\vec w, r, n) \writes J }
          {\Delta;(C,H) \ent T \ok & \cdot \ent (\Delta,\mu) : H&\DPE R : \Gamma }}\\[.25cm]
  \infer[\textsc{Mach-Read}]{\cdot \ent \mread(T,\Delta,(C,H),R,I, \vec \ell^H) \ok}
  {\deduce{\Delta \ent (C,H) : (\Xi;\Psi)
      \hskip 0.1in \DGE I:_s J \hskip 0.1in \DPE \vec \ell^H \reads J}
    {\Delta;(C,H) \ent T \ok & \cdot \ent
      (\Delta,\mu) : H & \DPE R : \Gamma}}\\[0.1in]
  \infer[\textsc{Mach-Write}]{\cdot \ent \mwrite(T,\Delta,(C,H),R,I,c,\ell^H, \vec \ell^H)
    \ok}
  {\deduce{\Delta \ent (C,H) : (\Xi;\Psi)
      \hskip 0.1in \DGE I:_s J \hskip 0.1in \DPE (\vec \ell^H, \ell^H,c)
      \writes J} {\Delta;(C,H) \ent T \ok & \cdot
      \ent (\Delta,\mu) : H & \DPE R : \Gamma}}
\end{tabular}}
\end{center}
% In read mode, each spinal instruction unifies
%%%one position in $M_1$ with the corresponding position in $M_2$, the postcondition being that $M_1$ and $M_2$ are equal. In the write-mode case, $M_1$ i a free variable
%which is bound to $M_2$ once all arguments are known.
However, the auxilliary invariants for Prolog spines have become more complex. 
The read spine invariant considers a term sequence $\vec{M}$ for the arguments already read and a second sequence $\vec{M'}$ for those remaining.
The invariant holds if (a) every $x_i$ can still unify with $M'_i$ and (b) every $\ell^H_i$ has the type expected by the spine type.
The write spine invariant requires that (a) the destination $\ell^H$ matches the result type of the constructor $c$, (b) the existing arguments $\vec{\ell^H}$ match the initial argument types, and (c) the remainder of the spine matches the remaining argument types.
\begin{center}{\footnotesize
\[\infer[\textsc{TWrites}]{\DPE (n,r,\vec\ell^H) \writes (\vec \tau_2 \to \{r:\cross{\vec\tau_1 \vec \tau_2})\}}
{\DPE \vec \ell^H : \vec \tau_1 & |\vec \tau_2| = n}\]
\begin{tabular}{cc}
  \infer[\textsc{Reads}]{\DPE \vec \ell^H \reads \Pi \vec x:\vec A.(c\ \vec M\ \vec M'
    \sqcap c\ \vec M\ \vec x)} {\deduce{\DPE \ell^H_i :
    \sing(M'_i:[M'_1,\ldots,M'_{i-1}/x_1,\ldots,x_{i-1}]A_i)}{\DE c\ \vec M\
    \vec x \sqcap c\ \vec M\ \vec M' = \sigma}}&
  \infer[\textsc{Writes}]{\DPE (\vec \ell^H, \ell^H,c) \writes \Pi \vec x : \vec
    a_2. x' \sqcap c\ \vec M\ \vec x}
  {\deduce{\Psi(\ell^H) = \sing(x' : a)}{\Sigma(c) = \vec a_1 \to \vec a_2 \to a & \DPE \vec \ell^H :
    \sing(\vec M : \vec a_1)} }\\[.25cm]
\end{tabular}}
\end{center}

\subsection{Metatheory}
\label{sec:dep-met}
In the dependent setting, we show our primary result that all TWAM 
programs are sound proof search procedures in the following sense:
\begin{thm}[Soundness]\label{dep:sound}
  If $\SE m \ok$ and $m \step^* m'$ and $m' \done$ then \\
  $m' = (T,\Delta,S,R,\succeed[M:A];I)$ and $\DE M : A$.
\end{thm}

This theorem is an immediate corollary of progress and preservation, by inversion on the typing rule for \verb|succeed|. 

Thus it suffices to show progress and preservation and their supporting lemmas.
We present here only the lemmas that are new or significantly different from the simply-typed versions.
%TODO: Electronic appendix? Write appendix
A detailed proof for the dependently-typed system is in the electronic appendix.

\subsubsection{Static Occurs Check}
\begin{lem}[Static Occurs Check Totality]
\label{lem:stat-occ-tot}
For all terms $M$ and all variables $x$, either $x \in M$ or $x \notin M$\end{lem}
\begin{proof}
By induction on the structure of $M$.
\end{proof}

\subsubsection{Static Unification}
\begin{lem}[Static Unification Totality]
\label{lem:stat-unif-tot}
 For all LF terms $M_1, M_2,$ if
$\DE M_2 : A$ and $\DE M_1 : A$ then $\DE M_1 \sqcap M_2 = \sigma$ or
$\DE M_1 \sqcap M_2 = \bot$.\end{lem}

\begin{proof}
By lexicographic induction on $|\Delta|$ and the structure of $M_1$.
The base cases hold by Lemma~\ref{lem:stat-occ-tot}.
The inductive case $M_1 = c\ \vec M$, $M_2 = c\ \vec M'$ (where $|M| = |M'|$) relies on a subclaim:
\begin{claim}For each $1 \leq i \leq |\vec M|,$ consider $\sigma = [\sigma_{i-1},\ldots,\sigma_1]$. Then $[\sigma]\Delta \ent [\sigma]M_i \sqcap M_i' = \sigma_i$ for some
$\sigma_i$ or $[\sigma]\Delta \ent [\sigma]M_i \sqcap [\sigma]M_i' = \bot$.\end{claim}
\begin{proof} By cases on $i$.

\case{$i=1$} By IH because  $M_1$ is structurally smaller than $c\ \vec M$.

\case{$i>1$} By IH: if $\sigma = \cdot$ then $M_i$ is structurally smaller than $c\ \vec M$, else $\sigma = [\vec M''/\vec x]$ where $x_i \in \dom{\Delta}$ and thus $|[\sigma]\Delta| < \Delta$.
\end{proof}
\end{proof}
\begin{lem}[Static Unification Correctness] If $\DE M : A$ and $\DE M' : A$ and $\DE M \sqcap M' = \sigma$, then
\label{lem:static-unif}
\begin{itemize}
\item $[\sigma]M = [\sigma]M'$
\item For all substitutions $\sigma',$ if $[\sigma']M = [\sigma']M'$ then there exists some $\sigma^*$ such that $\sigma' \equiv_{\alpha} \sigma^*,\sigma$.
\end{itemize}\end{lem}
\begin{proof} Analogous to standard results from the literature. \end{proof}

As a technical device for Lemma~\ref{lem:doom} we introduce a judgement $M_1\sqsubset M_2$ meaning ``$M_1$ is a strict substructure of $M_2$'':
\begin{center}
\begin{tabular}{cc}
  \infer[\sqsubset\textsc{-Base}]{M_i \sqsubset c\ \vec M}{} &
  \infer[\sqsubset\textsc{-Ind}]{M \sqsubset c\ \vec M}{M \sqsubset M_i}
\end{tabular}
\end{center}
The following lemmas support the proof of Lemma~\ref{lem:doom}.
\begin{lem}[Occurs to Substructure]\label{lem:occ-substr}If $x \in M$ and $x \neq M$ then $x \sqsubset M$\end{lem}
\begin{proof}
By induction on the derivation $x \in M$.
\end{proof}

\begin{lem}[Substructure to Occurs]\label{lem:substr-occ}If $x \sqsubset M$ then $x \in M$. \end{lem}
\begin{proof}
By induction on the derivation $x \sqsubset M$.
\end{proof}

\begin{lem}[Substitution Preserves Substructure]\label{lem:subst-substr}If $M_1 \sqsubset M_2$ then $[M/x]M_1 \sqsubset [M/x]M_2$ \end{lem}

\begin{proof}
By induction on the derivation $M_1 \sqsubset M_2$.
\end{proof}

\begin{lem}[Transitivity of Substructure]
\label{lem:substr-trans}
If $M_1 \sqsubset M_2$ and $M_2 \sqsubset M_3$ then $M_1 \sqsubset M_3$.
\end{lem}
\begin{proof}
By induction on the derivation $M_2 \sqsubset M_3$.
\end{proof}

\begin{lem}[Substructures Don't Unify]\label{lem:subst-no-unif} If $\DE M : A, \DE M' : A$ and $M \sqsubset M'$ then $\DE M \sqcap M' = \bot$.
\end{lem}
\begin{proof} By lexicographic induction on $|\Delta|$ and the structure of $M$.

Consider the cases for $M \sqsubset M'$.
The case $x \sqsubset c\ \vec M$ holds by Lemma~\ref{lem:substr-occ}.

\case{\infer[\sqsubset\textsc{-Base}]{c\ \vec M \sqsubset c'\ \vec M'}{\text{(because $c\ \vec M = M'_i$)}}}
If $c' \neq c,$ then unification fails immediately, so assume $c \neq c'$.
Observe $M_i \sqsubset c\ \vec M$ by rule $\sqsubset$\textsc{-Base}. Since $M'_i = c\ \vec M$, we have $M_i \sqsubset M_i'$.
By Lemma~\ref{lem:subst-substr}, for any substitution $\sigma,$ we have $[\sigma]M \sqsubset [\sigma]M'$. Note that when we unify $c\ \vec M$ and
$c\ \vec M'$ we either fail before $M_i$ or attempt to compute $[\sigma]M \sqcap [\sigma]M'$ for some $\sigma$. If we failed already,
 the case is done. If we succeeded, then by the IH $[\sigma]\DE [\sigma]M \sqcap [\sigma]M' = \bot$ and we fail here.

\case{\infer[\sqsubset\textsc{-Ind}]{c\ \vec M \sqsubset c'\ \vec M'}{c\ \vec M \sqsubset M'_i}}As in the previous case, AWLOG $c = c'$. Now since $c\ \vec M \sqsubset c\ \vec M'$ then by Lemma~\ref{lem:substr-trans}, $M_i \sqsubset M'_i$ for the $i$ such that $c\ \vec M \sqsubset M_i$. The rest of the case is analogous to the last one.
\end{proof}

\begin{lem}[Unification Lemma of Doom]
\label{lem:doom}
Unifications that fail are doomed to fail forever. That is, if $\Delta,x:A \ent M_1 \sqcap M_2 = \bot$ and $\DE M : A$ then $[M/x]\DE [M/x]M_1 \sqcap [M/x]M_2 = \bot$.
\end{lem}
\begin{proof}
By lexicographic induction on $|\Delta|$ and the unification derivation $\Delta,x:A \ent M_1 \sqcap M_2 =\bot$.

\case{\infer[\bot{}x1]{\Delta,x:A \ent x' \sqcap M_2 = \bot}{x' \neq M_2 & x' \in M_2}}
Case on whether $x = x'$.

\case{$x \neq x'$} In this case, $[M/x]x' = x'$ and $x' \in [M/x]M_2$ so the unification still fails.

\case{$x = x'$} This case reduces to the following claim:
\begin{claim}If $x \in M$ and $\DE M' : A$ and $\Delta,x:A \ent M : A$ then $[M/x]\DE M' \sqcap [M/x] = \bot$. \end{claim}
\begin{proof}
  By Lemma~\ref{lem:occ-substr}, $x \sqsubset M$. 
  By Lemma~\ref{lem:subst-substr}, $M' \sqsubset [M'/x]M$.
  By Lemma~\ref{lem:subst-no-unif}, $[M'/x]\DE M' \sqcap [M'/x]M = \bot$.
\end{proof}
\case{\infer[\bot{}x2]{\Delta,x:A \ent M \sqcap x' = \bot}{x' \neq M & x' \in M}}
This case holds by symmetry.

\case{\infer[\bot{}c2]{\Delta,x:A \ent c\ M_1\ \ldots\ M_n \sqcap c'\ M_1'\ \ldots\ M_m' = \bot}{c \neq c'}}
This case holds because substitution preserves head constructors.

\case{\infer[\bot{}c1]{\DE c\ M_1\ \ldots\ M_n \sqcap c\ M_1'\ \ldots\ M_n' = \bot}
{\deduce[\vdots]{[\sigma_{i-1},\ldots,\sigma_1]\Delta \ent [\sigma_{i-1},\ldots,\sigma_1]M_i \sqcap [\sigma_{i-1},\ldots,\sigma_1]M_i' = \bot}{M_1 \sqcap M_1' = \sigma_1}}}

If some $i' < i$ fails to unify, we're done. Otherwise we attempt to unify $ [\sigma]M_i \sqcap [\sigma]M_2$ where  $\sigma = \sigma'_{i-1},\ldots,\sigma_1',M/x$.
By Lemma~\ref{lem:static-unif}, $\sigma = \sigma^*, \sigma_{i-1},\ldots,\sigma_1$ so
we can apply the IH to get
and
\[\infer[\bot{}c1]{[M/x]\DE [M/x]c\ M_1\ \ldots\ M_n \sqcap [M/x]c\ M_1'\ \ldots\ M_n' = \bot}
{\deduce[\vdots]{[\sigma'_{i-1},\ldots,\sigma'_1,M/x]\Delta \ent [\sigma'_{i-1},\ldots,\sigma'_1,M/x]M_i \sqcap [\sigma'_{i-1},\ldots,\sigma'_1,M/x]M_i' = \bot}{[M/x]M_1 \sqcap [M/x]M_1' = \sigma'_1}}\]
\end{proof}

\subsubsection{Substitution}
\begin{lem}[Substitution] All appropriate typing judgements support substitution.
  \label{lem:subst}
  \begin{enumerate}
  \item \emph{LF terms:} If $\Delta,x\colon{}A \ent M_1\colon{}A'$, $\DE M_2\colon{}A$ then $[M_2/x]\DE[M_2/x]M_1\colon{}[M_2/x]A$.
  \item \emph{Operands:} If $\Delta_1,x\colon{}A;\Delta_2;\myG\ent op\colon{}\tau$, $\Delta_1\ent M\colon{}A$ then $\Delta_1,[M/x]\Delta_2;[M/x]\myG \ent [M/x]op\colon{}[M/x]\tau$.
  \item \emph{Word values:} If $\Delta_1,x\colon{}A,\Delta_2;\PE w\colon{}\tau$, $\Delta_1\ent{}M\colon{}A$ then $\Delta_1,[M/x]\Delta_2;[M/x]\PE [M/x] w\colon{}[M/x]\tau$.
  \item \emph{Register Files:} If $\Delta_1,x\colon{}A,\Delta_2;\Psi\ent R\colon{}\myG$, $\Delta_1\ent M:A$ then $\Delta_1,[M/x]\Delta_2;[M/x]\Psi \ent [M/x]R \colon{} [M/x]\myG$.
  \item \emph{Heap values:}  If $\Delta_1,x\colon{}a,\Delta_2;\Psi\ent v^H\colon{}\tau$,$\Delta_1\ent M\colon{}a$, $v^H\neq\free{}[x\colon{}a]$ then \\$\Delta_1,[M/x]\Delta_2;[M/x]\Psi\ent[M/x]v^H\colon{}[M/x]\tau$
  \item \emph{Basic blocks:} If $\Delta_1,x\colon{}A,\Delta_2;\Gamma \ent I\ok$, $\Delta_1\ent M\colon{}A$ then $\Delta_1,[M/x]\Delta_2;[M/x]\myG \ent [M/x] I \ok$.
\end{enumerate}
\end{lem}
\begin{proof}
Each claim holds by lexicographic induction on $|\Delta|$ and the structure of the typing derivation.
The first five claims are straightforward. The interesting cases of the final claim are the unification instructions,
because there is a subtle interaction between static and dynamic unification.
The case for  $\getval{r_1}{r_2}$ is representative:

\case{{\footnotesize\infer[\textsc{GetVal-S}]{\Delta_1,x:A,\Delta_2;\GE \getval {r_1}{r_2}; I \ok}
{[\sigma]\Delta_1,x:A,\Delta_2;[\sigma]\Gamma \ent [\sigma]I \ok&\Delta_1,x:A,\Delta_2 \ent M_1 \sqcap M_2 = \sigma&\Gamma(r_1) = \sing(M_1 : a) & \Gamma(r_2) = \sing(M_2 : a)}}}
By claim 1, $\Delta_1,[M/x]\Delta_2\ent [M/x]M_1 : a$ and $\Delta_1,[M/x]\Delta_2\ent [M/x]M_2 : a$ so by
Lemma~\ref{lem:stat-unif-tot}, either $\Delta_1,[M/x]\Delta_2\ent [M/x]M_1 \sqcap [M/x]M_2 = \sigma'$ or $\Delta_1,[M/x]\Delta_2\ent [M/x]M_1 \sqcap [M/x]M_2 = \bot$.

\subcase{$\Delta_1,[M/x]\Delta_2\ent [M/x]M_1 \sqcap [M/x]M_2 = \sigma'$}
By Lemma~\ref{lem:static-unif}, $[\sigma']([M/x]M_1) = [\sigma']([M/x]M_2)$ which we
can rewrite as $[\sigma',M/x]M_1 = [\sigma',M/x] M_2$. Also by Lemma~\ref{lem:static-unif}, $\sigma$ is a most general unifier of $M_1$ and $M_2$. Thus there exists $\sigma^*$ such that $\sigma',M/x = \sigma^*,\sigma$ (they need not be syntactically equal, but must be alpha-equivalent). In particular, alpha-vary $\sigma^*$ such that it substitutes for $x$. Then by iterating the IH (we can do this because $|\Delta|$ decreases every time), $[\sigma^*,\sigma](\Delta_1,x:A,\Delta_2);[\sigma^*,\sigma]\Gamma \ent [\sigma^*,\sigma]I \ok$. By the assumption that $\sigma^*,\sigma$ substitutes for $x,$ we have $[\sigma^*,\sigma](\Delta_1,\Delta_2);[\sigma^*,\sigma]\Gamma \ent [\sigma^*,\sigma]I \ok$ which suffices to show the result:
\[\infer[\textsc{GetVal-S}]{\Delta_1,[M/x]\Delta_2;[M/x]\GE \getval{r_1}{r_2}; [M/x]I \ok}
{\deduce{\Delta_1,[M/x]\Delta_2 \ent [M/x]M_1 \sqcap [M/x]M_2 = \sigma^*,\sigma}{([M/x]\myG)(r_1) = \sing([M/x]M_1 : a)}
&\deduce{[\sigma^*,\sigma]\Delta;[\sigma^*,\sigma]\GE [\sigma^*,\sigma]I \ok }{([M/x]\myG)(r_2) = \sing([M/x]M_2 : a)}}\]
\subcase{$\Delta_1,[M/x]\Delta_2\ent [M/x]M_1 \sqcap [M/x]M_2 = \bot$}In this case, since the
unification failed, the result is vacuously well-typed:
\[{\footnotesize\infer[\textsc{GetVal-F}]{\Delta_1,[M/x]\Delta_2;[M/x]\GE \getval{r_1}{r_2}; [M/x]I \ok}
{([M/x]\myG)(r_1) = \sing([M/x]M_1 : a) &
 ([M/x]\myG)(r_2) = \sing([M/x]M_2 : a) &
 \Delta_1,[M/x]\Delta_2\ent [M/x]M_1 \sqcap [M/x]M_2 = \bot}
}\]
\case{\infer{\Delta_1,x:A,\Delta_2;\GE \getval {r_1}{r_2}; I \ok}
{\deduce{\Delta_1,x:A,\Delta_2\ent M_1 \sqcap M_2 = \bot}{\Gamma(r_1) = \sing(M_1 : a) & \Gamma(r_2) = \sing(M_2 : a)}}}
By Lemma~\ref{lem:doom}.
\end{proof}
\subsubsection{Soundness of Unification}
\begin{lem}[Soundness of \unify] If $\DE M_1 \colon{}a$, $\DE M_2 \colon{}a$, $\DE S \colon{} (\Xi;\Psi)$, $\Delta;S\ent T\ok$,
$\SE(\Delta;\mu)\colon{}S$, $\DPE\ell^H_1\colon{}\sing(M_1\colon{}A)$, $\DPE\ell^H_2\colon{}\sing(M_2\colon{}A)$ then
\begin{itemize}
\item If $\DE M_1 \sqcap M_2 = \bot$ then $\unify(\Delta,S,T,\ell^H_1,\ell^H_2) = \bot$
\item If $\DE M_1 \sqcap M_2 = \sigma$ then $\unify(\Delta,S,T,\ell^H_1,\ell^H_2) = (\Delta',S',T')$ where
$\Delta' = [\sigma]\Delta$ and $[\sigma]\DE H' \colon{} [\sigma]\Psi$ and $\Delta',S'\ent T' \ok.$
\end{itemize}
\end{lem}
\begin{proof}
  As in SWAM, with the additional use of Lemma~\ref{lem:subst}.
\end{proof}
\subsubsection{Words and Operands}
\begin{lem}[Word Totality]\label{lem:word-tot} If $\DPE w : \tau$ then $w \eval w'$ and $w' \canon$.
\end{lem}
\begin{proof}
By induction on $\size(w)$, defined by $\size(\ell) = 0, \size(w\ M) = 1 + \size(w), \size(\lambda x:A.w) = 1 + \size(w)$,
appealing to Lemma~\ref{lem:subst} and the fact that substitution preserves $\size(M)$.
\end{proof}
\begin{lem}[Operand Resolution]\label{lem:oper-res} For all operands $op$, if $\DGE_{\Xi} op : \tau$
 and $\DPE R : \Gamma$ then $R \ent op \squig w$
for some word $w$ and $\DPE_{\Xi} w : \tau$. \end{lem}
\begin{proof}
By induction on $\DGE_{\Xi} op : \tau$ and inversion on register file typing.
\end{proof}
\begin{lem}[Word Preservation]\label{lem:word-pres} If $\DPE w : \tau$ and $w \eval w'$ then $\DPE w' : \tau$.\end{lem}
\begin{proof} By induction on the trace $w \eval w'$ and Lemma~\ref{lem:subst}. \end{proof}
\begin{lem}[Operand Preservation]\label{lem:oper-pres} If $\DPE R : \myG$ and $\DGE op : \tau$ and $R \ent op \eval w$ then $w \canon$ and $\DPE w : \tau$.\end{lem}
\begin{proof}
Follows directly from Lemmas~\ref{lem:oper-res} and \ref{lem:word-pres}.
\end{proof}
\begin{lem}[Operand Canonicalization]\label{lem:oper-canon} If $\DPE R : \myG$ and $\DGE op : \tau$ then $R \ent op \eval w$ and $w \canon$ and $\DPE w : \tau$.
\end{lem}
\begin{proof}
By Lemma~\ref{lem:oper-res}, $R \ent op \eval w'$ and $\DPE w' : \tau$. By Lemma~\ref{lem:word-tot}, $w' \eval w$, so by rule $op\eval$ we have
$R \ent op \eval w$
and by Lemma~\ref{lem:word-pres}, $\DPE w : \tau$.
\end{proof}
\begin{lem}[Word Inversion]\label{lem:word-inv} If $\DPE \ell\ \vec M : \tau$ then $\DPE \ell : \Pi \vec x:\vec A.\tau'$ where $[\vec M/\vec x]\tau'=\tau$
and $\DE M_i : [M_1,\ldots,M_{i-1}/x_1,\ldots,x_{i-1}]A_i$.
\end{lem}
\begin{proof}
By induction on the derivation $\DPE \ell\ \vec M : \tau$.
\end{proof}
%As before, the \verb|unify_val| case of preservation relies on the soundness of unification, which
%we generalize to the following:
%The proof is similar to before, with additional appeals to substitution lemmas. \\
\subsubsection{Progress and Preservation}
\begin{thm}[Progress] If $\SE m \ok$ then either $m \done$ or $m \fails$ or $m \step m'$.\end{thm}
\begin{proof}
  By cases on $m$ (specifically, cases on $I$), and by Lemmas~\ref{lem:scf}, \ref{lem:bt-tot}, \ref{lem:end-corr}, \ref{lem:unify-sound}, \ref{lem:stat-unif-tot}, \ref{lem:oper-pres}, \ref{lem:oper-canon}, and \ref{lem:word-inv}.
\end{proof}
\begin{thm}[Preservation] If $\SE m \ok$ and $m \step m'$ then $\SE m' \ok$.\end{thm}
\begin{proof}
  By cases on $m \step m'$, using the assumptions of $\SE m\ok$.
First consider simulatenously all the cases that end in backtracking. 
Those cases hold by Lemma~\ref{lem:bt-tot}.
Thus it suffices to show the cases that do not backtrack.

\case{\infer[\textsc{Jmp-}\ell^H]{(T,\Delta,S,R,\jmp op; I) \step (T,\Delta, S,\hup{R}{r_1}{w},[\vec M\vec M'/\vec x\vec x']I')}
{R \ent op \hook \ell^H\ \vec M & C(\ell^H) = \close(w, \ell^C\ \vec M')}}
By Lemma~\ref{lem:oper-pres}, $\DPE \ell^H\ \vec M : \neg \Gamma'$. By typing assumption, $\ell^C\ \vec M : \Pi \vec x':\vec A'.\neg \hup{\myG''}{r_1}{\tau}$
so by Lemma~\ref{lem:word-inv}, $\DPE \ell^C : \Pi \vec{x}\vec{x'}:\vec A\vec A'.\neg \myG'''$ where $[\vec M'/\vec x']\myG''' = \myG''$. By Lemma~\ref{lem:scf},
$C(\ell^C) = \code[\vec x' \vec x: \vec A' \vec A](\lambda \vec x' \vec x : \vec A' \vec A.I')$ and $\vec x' \vec x:\vec A' \vec A;\heap{r_1 : \tau} \ent I' \ok$.
By Lemmas~\ref{lem:sweak} and \ref{lem:subst}, $\Delta;[\vec M \vec M'/\vec x\vec x']\thup{\Gamma'}{r}{\tau} \ent [\vec M/\vec x]I' \ok$ where $\DPE w : \tau=[\vec M'\vec M/\vec x'\vec x]\tau$. By assumption, $\DE \myG' \leq \myG$, so by Lemma~\ref{lem:reg-sub}, $\Delta;\thup{\Gamma}{r_1}{\tau} \ent [\vec M'\vec M/\vec x'\vec x]I' \ok$. By assumption $w \canon$ so  $\DPE \hup{R}{r_1}{w} : \thup{\Gamma}{r_1}{\tau},$ then $\SE m' \ok$.

\case{\infer[\textsc{Jmp-}\ell^C]{(T,\Delta, S,R, \jmp op; I) \step (T,\Delta, S,R,[\vec M/\vec x]I')}
  {R \ent op \hook \ell^C\ \vec M & C(\ell^C) = \code[\vec x:\vec A.\Gamma'](\lambda \vec x : \vec A.I')}}
By Lemma~\ref{lem:oper-pres}, $\DPE \ell^H\ \vec M : \neg \Gamma'$. By assumption, $(\vec x : \vec a);\Gamma' \ent I' \ok$  and by Lemmas~\ref{lem:sweak} and \ref{lem:subst}, $\Delta;[\vec M/\vec x]\Gamma' \ent [\vec M/\vec x]I' \ok$.
Since $\DE [\vec M/\vec x]\Gamma' \leq \Gamma$, by Lemma~\ref{lem:reg-sub}, $\DGE [\vec M/\vec x]I' \ok$, so $\SE m' \ok$.

\case{\infer[\textsc{Close}]{(T,\Delta, S,R,\close r_d, r_s, \ell^C\ \vec M; I) \step (T,\Delta, \hext{S}{\ell^H}{\close(w,\ell^C\ \vec M)},\hup{R}{r}{\ell^H},I)}
{R(r_s) = w}}
By assumption and inversion on $\myG(r) = \tau$, we have $\DPE \close(w,\ell^C\ \vec M) : (\Pi \vec x:\vec A.~\neg \Gamma'),$ so
$\DE \hext{S}{\ell^H}{\close[\Gamma'](w,\lambda \vec x:\vec A.I')} : \thext{\Psi}{\ell^H}{\Pi \vec x:\vec A.~\neg \Gamma'}$ and
$\Delta;\thext{\Psi}{\ell^H}{\Pi \vec x:\vec A.~\neg \Gamma'} \ent \hup{R}{r}{\ell^H} : \thup{\Gamma}{r}{\Pi \vec x:\vec A.~\neg \Gamma'}$ and thus, $\SE m' \ok$.

\case{\infer[\textsc{PushBT}\step]{(T,\Delta, S,R,\branch r, \ell^C\ \vec M; I) \step ((\epsilon,w,\ell^C\ \vec M)::T,\Delta, S,R,I)}
{R(r) = w}}
By $\epsilon$ case of $\unwind$, we have $\unwind(\Delta,S,\epsilon) = (\Delta;S)$ so let $\Delta'= \Delta, S'=S,\mu'=\mu$. By Lemma~\ref{lem:scf} and inversion, $\Delta';\Psi'\ent w:\tau$ and $w \canon$ and the rest holds by assumption:
\[\infer[\textsc{Trail-Cons}]{\Delta;S\ent (\epsilon,w,\ell^C\ \vec M)::T \ok}
{\deduce{\Delta';\Psi'\ent w:\tau\hskip 0.1in\Delta';\Psi'\ent \ell^C\ \vec M : \neg\heap{r_1:\tau}\hskip 0.1in w\canon}
{\unwind(\Delta,S,\epsilon)=(\Delta';S')&\Delta';S'\ent T\ok&\SE(\Delta';\mu'):H' &\SE H':\Psi'}}\]
%{\deduce{\Delta';\Psi'\ent w : \tau  \hskip 0.1in \Delta';\Psi'\ent \ell_C\ \vec M : \neg\heap{r_1:\tau \hskip 0.1in w \canon}}
%{\unwind(\Delta,S,\epsilon) = (\Delta';S')  & \Delta';S'\ent T \ok  & \SE (\Delta';\mu'):H'  & \SE H' : \Psi'}}\]
so $\SE m' \ok$ as well by the assumption $\DGE I \ok$.

\case{\infer[\textsc{GetVal}\step]{(T,\Delta, S,R, \getval{r_1}{r_2}; I) \step (T',\Delta', S',R,I)}
{R(r_1) = w_1 & R(r_2) = w_2 & \unify(\Delta,S,T,w_1,w_2)=(\Delta',S',T')}}

By assumption $\DE M_1 \sqcap M_2 = \sigma$ where $\DPE w_1 : \sing(M_1:a)$
and $\DPE w_2 : \sing(M_2 : a)$. By Lemma~\ref{lem:unify-sound}, $\Delta' = [\sigma]\Delta$, $\DE S' : [\sigma]\Psi$ and $\SE (\Delta', [\sigma]\mu) : H'$ and $\Delta';S'\ent T' \ok$. By
substitution, $[\sigma]\Delta;[\sigma]\Psi \ent R : [\sigma]\Gamma$ and
by assumption $[\sigma]\Delta;[\sigma]\Gamma \ent I \ok$ so $\SE m' \ok$.

\case{\infer[\textsc{Mov}\step]{(T,\Delta, S,R, \mov{r_d}{op}; I) \step (T,\Delta, S,\hup{R}{r_d}{w},I)}
{R\ent op \eval w}}
By Lemma~\ref{lem:oper-pres}, $\DPE w : \tau$ and $w \canon$ so  $\DPE \hup{R}{r_d}{w} : \thup{\Gamma}{r_d}{\tau}$ and $\SE m' \ok$.

\case{{\small\infer[\textsc{PutVar}\step]{(T,\Delta, S,R, \dputvar{r}{x:a.}  I) \step(T,(\Delta,x:a),\hext{S}{\ell^H}{\free{}[x:a]},\hups{R}{r \hook \ell^H},I)}{}}}

Let $\mu' = (x@\ell^H:a,::\mu),$ then have $\SE ((x:a,\Delta), ((x@\ell^H:a)::\mu)) : \hup{H}{\ell^H}{\free{}[x:a]}$ by rule $\mu{}\textsc{-Cons}$.
%\[\infer{\SE ((x:a,\Delta), (x@\ell_H,\mu)) : \hup{H}{\ell_H}{\free{}[x:a]}}
%{\SE (\Delta,\mu) : H}\]
We also have $\Delta,x:a \ent \hext{S}{\ell^H}{\free{}[x:a]} : \Xi;\thext{\Psi}{\ell^H}{\sing(x:a)}$ and $\Delta,x:a;\thext{\Psi}{\ell^H}{\sing(x:a)} \ent \hup{R}{r}{\ell^H} : \thup{\Gamma}{r}{\sing(x:a)}$. By Lemma~\ref{lem:trail-up}, $\Delta,x:a;\hext{S}{\ell^H}{\free{}[x:a]} \ent T \ok$ which together with the
previous statements gives us $\SE m' \ok$.

\case{\infer[\textsc{GetStr}\step\textsc{W}]{(T,\Delta, S,R,\getstr {c}{r}; I) \step
\mwrite(T,\Delta, S,R,I,c,\ell'^H,\epsilon)}
{R(r) = \ell^H & \edn(S,\ell^H) = \ell'^H & S(\ell'^H) = \free{}[x:a]}}
By assumption, suffices to show $\DPE (\vec \ell^H, r, c) \writes J$. Since
$\vec \ell^H$ is empty, we need only know $\Sigma(c) = \vec a_2 \to a$ where
$J$ accepts $\vec a_2$ which is true by assumption. We also
need $S(\ell'^H) = \free{}[x:a]$ which is true by case.

\case{\infer[\textsc{GetStr}\step\textsc{R}]{(T,\Delta, S,R,\getstr{c}{r}; I) \step
\mread(T,\Delta, S,R,I,\vec w)}
{R(r) = \ell^H & \edn(S,\ell^H) = \ell'^H & S(\ell'^H) = \cwsh}}
Suffices to show $\vec \ell^H \reads J$, and in particular $\DE M_1 \sqcap M_2 = \sigma$. Since $M_1 = x, M_2 = c\ \vec x$ and $x \notin \vec x$ (because $x\fresh$), the unification succeeds with $x \sqcap c\ \vec x = [c\ \vec x/x]$.
%\[\infer{x \sqcap c\ \vec x = [c\ \vec x/x]}{x \notin \vec x}\]

\case{{\footnotesize\infer[\textsc{PutStr}\step]{(T,\Delta, S,R,\putstr{c}{r}; I) \step\mwrite(T,(x:a,\Delta),
      \hext{S}{\ell^H}{\free{}[x:a]},\hup{R}{r}{\ell^H},I,c,\ell^H,\epsilon)}{\Sigma(c) = \vec a \to a}}}
Similar to the $\mwrite$ case of {\tt getstr}. Define $\mu' = (x@\ell^H:a),\mu,$
then $\SE (((x@\ell^H:a)::\mu),(\Delta,x:a)) : \hext{S}{\ell^H}{\free{}[x:a]}$.
Now $S(\ell^H) = \free{}[x:a]$ as needed, and by assumption, $J$ accepts $\vec a_2$.

\vspace{0.1in}
\case{\infer[\textsc{PutTuple}\step]{(T,\Delta, S,R,\puttuple{r, n}; I) \step
\twrite(T,\Delta, S,R,I,r,n,\epsilon)}{}}
Suffice to show $\DPE (\vec w, r,n) \writes J$ where in this case $\vec w$ is empty and $J = \vec \tau \to \heap{r : \cross{\vec \tau}}$. Since $n = |\tau|,$ we have $\DPE (\epsilon, r, n) \writes \vec \tau \to \heap{r:\cross{\vec \tau}}$.
%\[\infer{\DPE (\args{}, r, n) \writes \vec \tau \to \heap{r:\cross{\vec \tau}}}
%        {|\vec\tau| = n}\]

\case{\infer[\textsc{Proj}\step]{(T,\Delta, S,R,\proj r_d, r_s, i; I) \step (T,\Delta, S,\hup{R}{r_d}{w_i},I)}
{R(r_s) \ell^H & S(\ell^H) = \args{w_1,\ldots,w_i,\ldots,w_n}}}
By inversion and Lemma~\ref{lem:scf}, $R(r_s) = \ell^H$ and $\DPE \ell^H : \cross{\vec \tau}$ so by Lemma~\ref{lem:scf}, $S(\ell^H) = \args{w_1, \ldots,w_i,\ldots,w_n}$ and
$\DPE w_i : \tau_i$. Therefore $\DPE \hup{R}{r_d}{w_i} : \thup{\Gamma}{r_d}{\tau_i}$ so $\SE m \ok$.

\case{\infer[\textsc{SetVal}\step]{\twrite(T,\Delta, S,R, \setval r_s; I,r_d,n,\vec w) \step
\twrite(T,\Delta, S,R,I,r_d,n-1,(\vec w::w))}{R(r_s) = w & n > 0}}
Consider $\vec \tau_1, \vec \tau_2$ from the derivation $\SE (\vec w, r, n) \writes J$. Observe by assumption $\vec \tau_2$ has form $\tau,\vec \tau_2'$.
Now let $\vec \tau_1' = \vec \tau_1, \tau$. This gives us $\DPE (\vec w::w) : \vec \tau_1'$ and $|\vec \tau_2'| = n-1$ and $\vec \tau_1'\vec\tau_2' = \vec\tau_1\vec\tau_2$ so $\DPE ((\vec w::w),r,n-1) \writes (\tau_2' \to \heap{r:\cross{\vec\tau_1'\vec\tau_2'}})$. The rest is by assumption.

\vspace{0.1in}
\case{\infer[\textsc{TWrite}\step]{\twrite(T,\Delta, S,R,I,r,0,\vec w) \step (T,\Delta, \hext{S}{\ell^H}{\args{\vec \ell^H}},\hup{R}{r}{\ell^H},I)}{}}
Since $n=0,$ the instruction typing derivation has form $\DGE I : \heap{r:\cross{\vec\tau}}$
where $\DPE \vec \ell^H : \vec \tau$ (from the derivation $\DPE(\vec \ell^H,r,0) \writes J$). This derivation must contain $\Delta;\thup{\Gamma}{r}{\cross{\vec \tau}} \ent I \ok$. We also have $\DPE \args{\vec w} : \cross{\vec \tau}$
%\[\infer{\DPE \args{\vec w} : \cross{\vec \tau}}
%{\DPE \vec w : \vec \tau}\]
so $\DE \hext{S}{\ell^H}{\args{\vec w}} : (\Xi;\thext{\Psi}{\ell^H}{\cross{\vec\tau}})$ and $\Delta;\thext{\Psi}{\ell^H}{\cross{\vec\tau}} \ent \hup{R}{r}{\ell^H} : \thup{\Gamma}{r}{\cross{\vec \tau}}$ so $\SE m' \ok$.

\vspace{0.1in}
\case{\infer[\textsc{Read}\step]{\mread(T,\Delta, S,R,I,\epsilon) \step (T,\Delta, S,R,I)}{}}
Then $J = c\ \vec M \sqcap c\ \vec M$ which unifies under the empty substitution, so by assumption, $\DGE I \ok$ so $\SE m' \ok$.

\case{
\infer[\textsc{Write}\step]{\mwrite(T,\Delta, S,R,I,\ell^H,c,\vec \ell^H) \step
(T',[c\ \vec M/x]\Delta,  \hup{S}{\ell^H}{c\args{\vec \ell^H}},R,[c\ \vec M/x]I)}
  {\deduce{|{\vec \ell^H}| = \arity(c) \hskip0.1in \uptrail(T,x@\ell'^H:a) = T'\hskip0.1in\edn(S,\ell^H) = \ell'^H}
    {S(\ell'^H) = \free{}[x : a] & \ell'^H \notin c\args{\vec \ell^H}}}}
%\infer{\mwrite(T,\Delta, S,R,I,\ell_H,c,\vec \ell_H) \step
%(T',[c\ \vec M/x]\Delta,  \hup{S}{\ell_H}{c\args{\vec \ell_H}},R,[c\ \vec M/x]I)}
%{\deduce{ |{\vec \ell_H}| = \arity(c) \hskip 0.1in \uptrail(T,x:a@\ell'_H) = T'}
%  {\edn(S,\ell_H) = \ell'_H & S(\ell'_H) = \free{}[x : a] & \ell'_H \notin c\args{\vec \ell_H}}}
%}
By Lemma~\ref{lem:trail-up}, $[M/x]\Delta,\hup{S}{\ell'^H}{c\args{\vec \ell^H}}\ent T' \ok$.
Then by assumption, $\DPE \vec \ell^H : \sing(\vec M : \vec a)$ and $\Sigma(c) = \vec a \to a$ so $\DPE\func{c}{\vec \ell^H} :\sing(c\ \vec M : a)$. By assumption and Lemma~\ref{lem:end-corr}, $S(\ell'^H) = \free{}[x : a]$. The typing derivation for $J$ has the form
\[\infer[\sqcap\sigma]{\DGE I:_s (x \sqcap c\ \vec M)}{\DE x \sqcap c\ \vec M = \sigma & [\sigma]\Delta;[\sigma]\Gamma \ent [\sigma]I \ok}\]
By inversion, $\sigma = [c\ \vec M/x]$. Therefore $[c\ \vec M/x]\Delta;[c\ \vec M/x]\Gamma \ent [c\ \vec M/x]I \ok$ so
it suffices to show $\DE \hup{S}{\ell^H}{c\args{\vec \ell^H}} : (\Xi;[c\ \vec M/x]\Psi)$ which it does by Lemma~\ref{lem:heap-update}.

\case{\infer[\textsc{UnifyVar}\step\textsc{R}]{\mread(T,\Delta, S,R,(\unifyvar r, x:a. I), \vec \ell^H) \step
\mread(T,\Delta, S,R,I,(\vec \ell^H::w))}{R(r) = w}}
By inversion, $\Gamma(r) = \sing(M:a)$ so by Lemma~\ref{lem:subst}, $\DGE [M/x]I :_s [M/x]J$ and $\SE m' \ok$.

\case{\infer[\textsc{UnifyVal}\step\textsc{R}]{\mread(T,\Delta, S,R,(\unifyval r, x:a. I), (\ell^H::\vec \ell^H)) \step
\mread(T',\Delta', S',R,I, \vec \ell^H)}
{R(r) = \ell'^H & \unify(\Delta,S,T,\ell^H,\ell'^H) = (\Delta',S',T')}}
By assumption $\DE x \sqcap M = \sigma$ for some $\sigma = [M/x]$. By Lemma~\ref{lem:unify-sound}, $\Delta' \ent S' : (\Xi;\Psi')$ and $\Delta';S'\ent T' \ok$.
By Lemma~\ref{lem:subst}, $\DGE [M/x]I : [M/x]J$.
We now also have $\DE \vec \ell^H \reads \Pi \vec x : \vec A. (c\ \vec M\ M\ \vec x' \sqcap c\ \vec M\ M\ \vec M'')$ as desired, so $\SE m' \ok$.

\vspace{0.1in}
\case{
%\infer{\mwrite(T,\Delta,S,R,\unifyvar[a] r_s; I,c,\ell_d,\vec \ell_H)\step\mwrite(T,(x:a,\Delta),\hext{S}{\ell_H}{\free{}[x:a]},\hup{R}{r_s}{\ell_H},I,c,\ell_d,(\vec\ell_H,\ell_H))}{}
%}
\infer[\textsc{UnifyVar}\step\textsc{W}]{\deduce{\mwrite(T,(x:a,\Delta), \hext{S}{\ell^H}{\free{}[x:a]},\hup{R}{r}{\ell^H},I,c,\ell^H_d,(\vec \ell^H::\ell^H))}{\mwrite(T,\Delta, S,R,(\unifyvar r, x:a. I),c,\ell^H_d,\vec \ell^H) \step}
}{}}
By Lemma~\ref{lem:trail-up}, $\Delta,x:a;\hext{S}{\ell^H}{\free{}[x:a]}\ent T \ok$. Let $\mu' = ((x@\ell^H:a),::\mu)$. Now $\SE ((\Delta,x:a),((x@\ell^H:a)::\mu)) : \hext{H}{\ell^H}{\free{}[x:a]}$.
Take $\vec a_1, \vec a_2$ from the derivation $\DPE (\vec \ell^H, \ell^H_d, c) \writes \Pi \vec x:\vec \tau.(M_1\sqcap M_2)$. Note $\vec a_2$ has form $a,\vec a_2'$.
Let $\vec a_1' = (\vec a_1::a')$. Now $(\Delta, x:a);\thext{\Psi}{\ell^H}{\free{}[x : a]} \ent (\vec \ell^H, \ell^H) : \vec a_1'$
so $\SE m' \ok$.

\case{{\small\infer[\textsc{UnivyVal}\step\textsc{W}]{\mwrite(T,\Delta, S,R,(\unifyval r, x:a. I),c,\ell^H_d,\vec \ell^H) \step
\mwrite(T,\Delta, S,R,I,c,\ell^H_d,(\vec \ell^H::\ell^H))}{R(r) = \ell^H}}}
Take $\vec a_1, \vec a_2$ from the derivation $\DPE (\vec \ell^H, \ell^H_d, c) \writes \Pi \vec x:\vec \tau.(M_1\sqcap M_2)$. Note $\vec a_2$ has form $a',\vec a_2'$.
Let $\vec a_1' = (\vec a_1,a')$. Now $\DPE (\vec \ell^H, \ell^H) : \vec a_1'$
so $\SE m' \ok$.
\end{proof}

\subsection{Operands, {\tt mov} and Tail Calls}
\label{sec:tco}
Now that we have proven TWAM sound, we illustrate its optimization potential
by describing our implementation of tail-call optimization (TCO), a common and
performance-critical optimization. We take as an example the predicate $f$:
\begin{verbatim}
f : t -> prop.
f(X) :- g(X).
\end{verbatim}
where the definition of $g$ is irrelevant. In our LF translation, this predicate has
one proof term constructor: \verb|F-X:|$\Pi$\verb| X:t.|$\Pi$\verb|D:g X.f X|.
If we compile this program na\"{i}vely, it would produce the following code:

\begin{exmp}[Before Tail-Call Optimization]
\label{ex:before-tco}~\\
\verb+f_main +$\mapsto$\verb+ code[X:nat. {r1:+$\sing(X)$\verb+, r2: +$\Pi D\colon{}f\ X.\ \neg$\verb+{}}(+\\
\verb+  +$\lambda X\colon{}nat.$\\
\verb+  mov f_tail, r2;+\\
\verb+  jmp (g X)+\\
\verb+),+

\verb+# The following code value should not be necessary+\\
\verb+f_tail +$\mapsto$\verb+ code[X:nat, D:g X. {r1:+$\sing(X)$\verb+, r2: +$\Pi D\colon{}f\ X.\ \neg$\verb+{}}(+\\
\verb+  +$\lambda X\colon{}nat, D\colon{}g\ X.$\\
\verb+  jmp (r2 (F-X D))+\\
\verb+)+
\end{exmp}
The problem is that to find a proof of \verb|f X| we must first find a proof \verb|D : g X| and then apply \verb|F-X X D|.
Because \verb|F-X| uses \verb|D| we must apply it after \verb|g| has succeeded,
which means we must apply it in \verb|g|'s success continuation.
This is a problem: \verb|g| is supposed to be a tail call, so its success continuation should be
the one passed to \verb|f|.

Since LF proofs are completely unnecessary at runtime, our ideal solution would be no-op of sorts:
an instruction that allows us to perform simple proof steps in LF, but which can
trivially be deleted after typechecking to avoid any runtime cost. This no-op is easily expressed
as a special case of the \verb+mov+ instruction. The following
\verb+mov+ instruction takes the success continuation and pre-composes an LF term that converts
proofs of \verb|g X| into proofs of \verb|f X| as needed by the continuation.

\begin{exmp}[After Tail-Call Optimization]\label{ex:aften-tco}~\\
\verb+f_main +$\mapsto$\verb+ code[X:t.{r1 :+$\Pi$\verb+D : f X.+$\neg$\verb+{}, r2 : +$\sing($\verb+X+$)$\verb+}](+\\
\verb+  +$\lambda$\verb+X:t.+\\
\verb+  mov r1 (+$\lambda$\verb+ D: g X. r1 (f/X X D));+\\
\verb+  jmp (g X)+\\
\verb+)+
\end{exmp}
\section{Implementation}
The full source code for our compiler implementation is available at \\\verb+http://www.cs.cmu.edu/~bbohrer/pub/twam.zip+, along with
a small test suite. The compiler consists of approximately 5000 lines of Standard ML, and has the following phases:
\begin{itemize}
\item Parsing and Elaboration: T-Prolog is parsed with {\tt ML-Lex} and {\tt ML-Yacc}, then check T-Prolog types.
\item Flattening: Terms with nested constructors are flattened.
\item Main Translation: Generates a variant of TWAM where failure continuations are stored inline.
\item Hoisting: Lifts failure continuations to the top level.
\item Certification: Runs the TWAM typechecker on hoisted code. On failure, signals a compile error.
\item Type Erasure: Trivial conversion from TWAM code to SWAM code.
\item Rechecking: As a sanity check, type-check the SWAM code.If this fails, signals a compiler error.
\item Interpretation: The SWAM code is interpreted, following the operational semantics.
\end{itemize}
The primary goal of this implementation, which it achieved, was to validate the design of TWAM,
especially to show that it is expressive enough to support an interesting source language.
For this reason, we intentionally omitted most optimizations except tail-call optimization, which we deemed too essential to ignore.
That being said, there are a number of avenues available to build a compiler with more competitive performance:
\begin{itemize}
\item Compile to machine code instead of interpreting TWAM code.
\item Replace our trivial register allocator with an efficient (e.g. graph-coloring-based) one.
\item Implement existing WAM optimizations (e.g. optimized switch statements).
\item 
Investigate the cost of continuation-passing style.
Implement optimizations to reduce the number of environments allocated, or develop a stack-based system should that be insufficient.
\item Reduce the use of the occurs check by adding a mode system.
\end{itemize}
Among these, the first two options can be implemented with no changes to the instruction set or type system, 
and the third can be implemented by adding new instructions without modifying the existing instructions.
The final two optimizations are more fundamental in nature, requiring changes to existing instructions or changes that affect the entire type system.

\section{Conclusion}

We have designed and implemented a typed compiler for T-Prolog by first creating a \textit{certifying abstract machine} for logic programs, called the TWAM.
Our implementation demonstrates that the TWAM is expressive enough to use as a compilation target for real programs,
and that the implementation burden of TWAM is acceptable.
This implementation result supports our primary contributions: the development of the TWAM and its metatheory.
The metatheory shows that TWAM typechecking suffices to enable certifying compilation in theory, and our compiler shows it in practice.

Our work differs from previous work on Prolog compilation because 
we are the first to take a typed compilation approach.
We have also produced a working compiler with a formal guarantee, whereas previous efforts stopped before implementing a compiler
\cite{Russinoff92averified,Borger95thewam---definition,Beierle92correctnessproof,Pusch96verificationof}.
Several optimizing compilers have been verified in proof assistants \cite{Leroy06formalcertification,cakeml} and some of them use
proof-producing compilation \cite{JFP14}, but these do not address logic programming languages. 
Typed compilation is from our perspective an instance of certifying compilation\cite{necula1998design},
and proof-passing style specifically is a variant that allows us reason about semantic preservation.
%\todo{should cite Flint or another typed compiler or something}

Our type system relies on the logical framework LF \cite{Harper93aframework}, and
is inspired by other languages with dependent type systems \cite{Xi98dependenttypes}, though
the languages differ greatly. Our formalisms are inspired by 
typed assembly languages, but we make major changes to provide stronger guarantees and support
logic programming \cite{Morrisett:1999:SFT:319301.319345}.

Future work includes developing a production-quality optimizing
compiler and runtime, including any changes to the core TWAM language
to enable greater efficiency. We also wish to extend our abstract
machine to support logic programming languages with advanced type
system features and investigate whether certifying abstract machines
can provide equally strong guarantees for non-logic programming languages.

\bibliographystyle{ACM-Reference-Format}
\bibliography{twam-long}

%\printbibliography
%\bibliography{thesis}{}
%\bibliographystyle{acmtrans}
%\newpage
%\addcontentsline{toc}{chapter}{Appendix A: Statics}
%\onecolumn
%\section*{Appendix A: Statics}
%\input{statics}
%\newpage
%\addcontentsline{toc}{chapter}{Appendix B: Operational Semantics}
%\section*{Appendix B: Operational Semantics}
%\input{dynamics}
%\newpage
%\addcontentsline{toc}{chapter}{Appendix C: Proofs}
%\section*{Appendix C: Proofs}
%\input{proofs}

\end{document}